\documentclass[11pt]{article}
\usepackage{graphicx}
\usepackage{subcaption}
\usepackage{todonotes}
\usepackage{mathtools}
\usepackage{soul}
\usepackage{amsthm}
\usepackage{amsmath,amssymb}
\usepackage{systeme}
\usepackage{gensymb}
\usepackage{comment}
\usepackage[ruled,vlined,noend]{algorithm2e}
\SetKwComment{Comment}{/* }{ */}
\SetKwInOut{Input}{input}\SetKwInOut{Output}{output}
\LinesNumbered
\SetKwRepeat{Do}{do}{while}

\usepackage{thm-restate}

\graphicspath{{./figs/}}
\usepackage[absolute]{textpos}

\usepackage{lineno}

\usepackage{fullpage,cite}
\usepackage[top=0.8in, bottom=0.9in, left=0.9in, right=0.9in]{geometry}

\usepackage[pdftex, plainpages = false, pdfpagelabels,
                 bookmarks=false,
                 bookmarksopen = true,
                 bookmarksnumbered = true,
                 breaklinks = true,
                 linktocpage,
                 pagebackref,
                 colorlinks = true,  
                 linkcolor = blue,
                 urlcolor  = cyan,
                 citecolor = red,
                 anchorcolor = green,
                 hyperindex = true,
                 hyperfigures
                 ]{hyperref}

\DeclareMathOperator*{\argmin}{arg\,min}

\newcommand{\clg}{{\lceil \log n \rceil}}
\newcommand{\os}[1]{\overset{(#1)}}
\newcommand{\F}{\mathcal{F}} 
\newcommand{\G}{G(S)} 
\newcommand{\R}{\text{Vor}} 
\newcommand{\E}{\mathcal{E}} 
\newcommand{\df}{d} 
\newcommand{\pif}{\pi} 
\newcommand{\DT}{\mathcal{DT}} 
\newcommand{\vd}{\mathcal{VD}}
\newcommand{\N}{N_{\DT}}
\newcommand{\calD}{\mathcal{D}}

\newcommand{\dist}{\mathrm{dist}}
\newcommand{\precxi}{\prec_{\Xi}}
\newcommand{\ca}{^\circ}
\newcommand{\sa}{^\diamond}
\newcommand{\avd}{\mathcal{AVD}}

\def\st{$s$-$t$}

\newtheorem{lemma}{Lemma}

\newtheorem{observation}{Observation}
\newtheorem{theorem}{Theorem}

\title{Shortest Paths in Geodesic Unit-Disk Graphs\thanks{A preliminary version will appear in {\em Proceedings of the 42nd International Symposium on Computational Geometry (SoCG 2026)}~\cite{ref:BrewerSh26}. This version further improves the result in the preliminary version by proposing a new dynamic data structure for priority-queue updates (Section~\ref{sec:priorityqueue}).}
}
\author{
Bruce W. Brewer\thanks{Kahlert School of Computing,
University of Utah, Salt Lake City, UT 84112, USA. {\tt bruce.brewer@utah.edu}}
\and
Haitao Wang\thanks{Kahlert School of Computing,
University of Utah, Salt Lake City, UT 84112, USA. {\tt haitao.wang@utah.edu}}
}
\date{}

\begin{document}


\maketitle

\vspace{-0.3in}
\begin{abstract}
Let $S$ be a set of $n$ points in a polygon $P$ with $m$ vertices.
The geodesic unit-disk graph $G(S)$ induced by $S$ has vertex set $S$ and contains an edge between two vertices whenever their geodesic distance in $P$ is at most one.
In the weighted version, each edge is assigned weight equal to the geodesic distance between its endpoints; in the unweighted version, every edge has weight $1$.
Given a source point $s \in S$, we study the problem of computing shortest paths from $s$ to all vertices of $G(S)$.
To the best of our knowledge, this problem has not been investigated previously.
A naive approach constructs $G(S)$ explicitly and then applies a standard shortest path algorithm for general graphs, but this requires quadratic time in the worst case, since $G(S)$ may contain $\Omega(n^2)$ edges.
In this paper, we give the first subquadratic-time algorithms for this problem.
For the weighted case, when $P$ is a simple polygon, we obtain an $O(m + n \log^{2} n \log^{2} m)$-time algorithm.
For the unweighted case, we provide an $O(m + n \log n \log^{2} m)$-time algorithm for simple polygons, and an $O(\sqrt{n}\,(n+m)\log(n+m))$-time algorithm for polygons with holes.
To achieve these results, we develop a data structure for deletion-only geodesic unit-disk range emptiness queries, as well as a data structure for constructing implicit additively weighted geodesic Voronoi diagrams in simple polygons. In addition, we propose a dynamic data structure that extends Bentley’s logarithmic method from insertions to priority-queue updates, namely insertion and delete-min operations. These results may be of independent interest.
\end{abstract}


\section{Introduction}
\label{sec:intro}
Let $S$ be a set of $n$ points in the plane.
The \emph{unit-disk graph} induced by $S$ has vertex set $S$ and contains an edge between two vertices if their Euclidean distance is at most one.
Equivalently, the unit-disk graph is the intersection graph of the set of congruent disks of radius $1/2$ centered at the points of $S$: each disk corresponds to a vertex, and two disks are adjacent precisely when they intersect.
Unit-disk graphs are extensively studied in computational geometry due to their rich geometric structure and numerous applications (see, e.g., \cite{ref:ClarkUn90}).
A primary motivation arises from wireless ad-hoc networks.
Here, each point of $S$ models a device capable of transmitting and receiving wireless signals within a fixed range, which we normalize to one.
Thus, two devices can communicate directly exactly when they are connected by an edge in the corresponding unit-disk graph.
Many fundamental algorithmic tasks for wireless networks reduce to classical graph problems on unit-disk graphs, such as coloring, independent set, and dominating set~\cite{ref:ClarkUn90}.
A central question is how to route a message between two devices $s$ and $t$.
If $s$ and $t$ are not within communication range, then the message must be relayed through intermediate devices.
In graph theoretic terms, the message must follow an \st\ path in the unit-disk graph, and ideally one of minimum length.
For this reason, the shortest path problem in unit-disk graphs plays a key role in both theory and applications.


Another natural application of unit-disk graphs arises in motion planning with recharging/refueling.
Suppose a vehicle (for example, a delivery drone) must travel from a start point $s$ to a target point $t$, but lacks sufficient battery capacity to traverse this distance directly.
Instead, it may need to stop at intermediate charging stations.
Consider the graph whose vertices represent charging stations and where two stations are connected by an edge if the vehicle can travel between them on a full charge.
After scaling the maximum travel distance on a full charge to one, this graph is precisely a unit-disk graph.
Thus, the motion planning problem reduces to finding a shortest $s$-$t$ path in the unit-disk graph.

Now consider the above applications in the presence of obstacles that must be avoided.
In wireless ad-hoc networks, an obstacle may prevent a signal from traveling directly between two devices, forcing it to bend around buildings or terrain features.
In the motion planning scenario, obstacles may correspond to physical barriers or restricted regions.
In such environments, Euclidean distance no longer reflects the true travel cost between two points.
Instead, the relevant distance is the \emph{geodesic distance}, i.e., the length of a shortest obstacle-avoiding path between two points.
Figure.~\ref{fig:unit_disk} illustrates a ``geodesic unit disk'' centered at $s$ and a unit-length \st\ path under geodesic distance.

\begin{figure}
    \centering
    \includegraphics[height = 5cm]{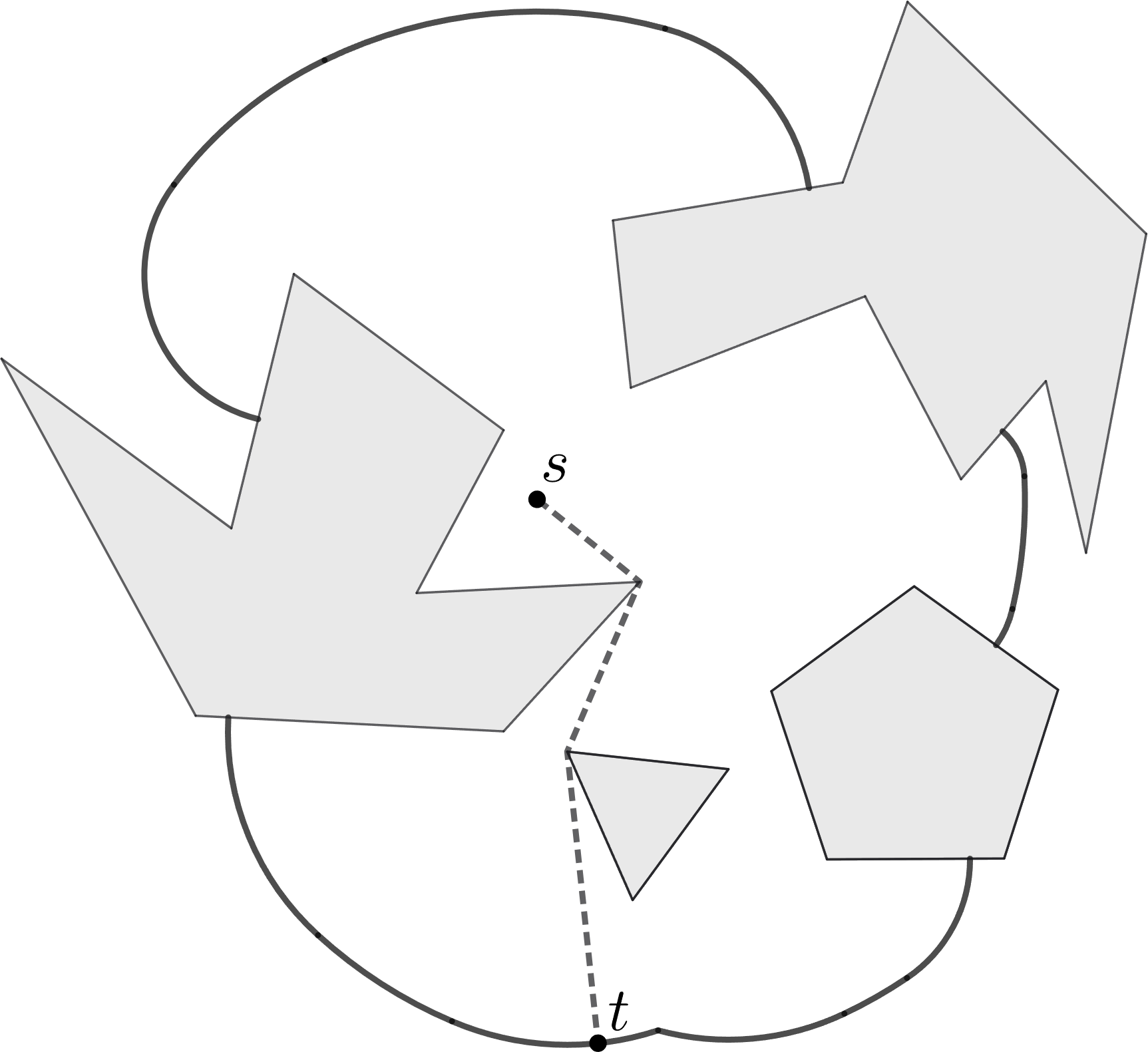}
    \caption{Example of a geodesic unit disk centered at $s$ (the black solid curve is the boundary of the disk).
    A unit-length \st\ path (the dashed path) is also shown.}
    \label{fig:unit_disk}
\end{figure}

\paragraph{Problem definition.}
Motivated by these applications, we study \emph{geodesic unit-disk graphs}, defined as follows.
Let $P$ be a polygonal domain with $h$ holes and a total of $m$ vertices in the plane.
Thus, $P$ is a closed, multiply connected region whose boundary consists of $h+1$ disjoint polygonal cycles.
The holes of $P$, together with the region outside $P$, are referred to as \emph{obstacles}.
For any two points in $P$, their \emph{geodesic distance} is the length of a shortest path between them in $P$.

Let $S$ be a set of points (also called \emph{sites}) contained in $P$.
The \emph{geodesic unit-disk graph} $\G$ induced by $S$ has vertex set $S$ and contains an edge between two sites whenever their geodesic distance is at most one.
In the \emph{unweighted case}, every edge of $\G$ has weight~$1$;
in the \emph{weighted case}, the weight of an edge is the geodesic distance between its endpoints.
The appropriate choice of model depends on the application.

Given a source site $s \in S$, we consider the single-source shortest-path (SSSP) problem in $\G$:
for a given polygonal domain $P$, site set $S$, and source $s$, compute shortest paths in $\G$ from $s$ to all other vertices.
To the best of our knowledge, this problem has not been studied previously.

\paragraph{Related work.}
As discussed above, problems on Euclidean unit-disk graphs (i.e., when $S$ is given in the plane without obstacles) are well studied~\cite{ref:ClarkUn90}.
For the unweighted case, Cabello and {Jej\v ci\v c}~\cite{ref:CabelloSh15} gave an $O(n \log n)$ algorithm for the SSSP problem, matching an $\Omega(n \log n)$ lower bound~\cite{ref:CabelloSh15}.
Chan and Skrepetos~\cite{ref:ChanAl16} later obtained an $O(n)$ algorithm under the assumption that the points are presorted by  $x$- and $y$-coordinates.

The weighted case has also received considerable attention.
Roditty and Segal~\cite{ref:RodittyOn11} first obtained an $O(n^{4/3+\epsilon})$-time algorithm for any $\epsilon>0$.
Cabello and {Jej\v ci\v c}~\cite{ref:CabelloSh15} improved this to $O(n^{1+\epsilon})$ time.
Their running time has since been reduced due to improvements in dynamic bichromatic closest-pair structures, first by Kaplan, Mulzer, Roditty, Seiferth, and Sharir~\cite{ref:KaplanDy20} and then by Liu~\cite{ref:LiuNe22}.
Most recently, Wang and Xue~\cite{ref:WangNe20} proposed a new algorithm that avoids dynamic bichromatic closest-pair structures, and their algorithm runs in $O(n \log^{2} n)$ time; Brewer and Wang~\cite{ref:BrewerAn24} further improved the time to $O(n \log^{2} n / \log \log n)$ by refining a key bottleneck subproblem.

We note that various other problems involving geodesic disks in polygons have also been studied, including separators~\cite{ref:deBergCl22}, clustering~\cite{ref:BorgeltGe07}, and piercing~\cite{ref:BosePi21}, among others.



\subsection{Our results}
\label{sec:results}

A straightforward approach is to construct $\G$ explicitly and then run a standard SSSP algorithm for general graphs.
However, $\G$ may contain $\Omega(n^{2})$ edges, leading to a quadratic running time in the worst case.
In this paper, we present three subquadratic-time algorithms for three problem variants, respectively.

\begin{enumerate}
    \item
    For the \emph{weighted} case, when $P$ is a simple polygon, we obtain an $O(m + n \log^2 n \log^{2} m)$-time algorithm.
    When applied to the Euclidean setting by taking $m = O(1)$ (i.e., $P$ is the entire plane), our running time matches the $O(n\log^2 n)$ time of the algorithm in \cite{ref:WangNe20} and is only a $\log\log n$ factor slower than the current best bound of $O(n \log^{2} n / \log \log n)$~\cite{ref:BrewerAn24}.
    The Euclidean algorithms in~\cite{ref:WangNe20,ref:BrewerAn24} rely crucially on a grid decomposition of the plane, a technique that does not extend naturally to the polygon setting because each grid cell may be fragmented by the boundary of $P$ multiple times. We therefore use a different algorithmic strategy.

    \item
    For the \emph{unweighted} case in a simple polygon, we give an $O(m + n \log n \log^{2} m)$-time algorithm.
    Applied to the Euclidean setting, this matches the $\Omega(n \log n)$ lower bound~\cite{ref:CabelloSh15}.

    \item
    When $P$ contains holes, the problem becomes substantially more challenging.
    For the unweighted case in this general setting, we develop an $O(\sqrt{n}\,(n+m)\log(n+m))$-time algorithm.
\end{enumerate}

A key component of our algorithms is the resolution of the following two fundamental subproblems, each of which may be of independent interest.
Indeed, designing efficient solutions to these subproblems is one of the main technical contributions of this paper.


\paragraph{Subproblem 1: Deletion-only geodesic unit-disk range emptiness queries.}
In this problem, we preprocess a set $S$ of $n$ points inside a simple polygon $P$ with $m$ vertices so that two operations can be performed efficiently:
(1) given a query point $q$, determine whether there exists a point $p \in S$ whose geodesic distance to $q$ is at most $1$, and if so return such a point; and
(2) delete a point from $S$.
We refer to the first operation as a \emph{geodesic unit-disk range emptiness query} (GUDRE query).

Using the dynamic randomized nearest neighbor data structure of Agarwal, Arge, and Staals~\cite{ref:AgarwalIm18}, each deletion can be supported in $O(\log^{7}n \log m + \log^{6}n \log^{3}m)$ expected amortized time, and each GUDRE query can be answered in $O(\log^{2}n \log^{2} m)$ time.
In the Euclidean case (i.e., when $P$ is the entire plane), one can preprocess $S$ in $O(n\log n)$ time and support both queries and deletions in $O(\log n)$ time~\cite{ref:EfratGe01,ref:WangCo26}.

We obtain the following result. After an $O(m)$-time \emph{global} preprocessing step for $P$, we can, for any given set $S$ of $n$ points in $P$, build a data structure in $O(n\log^{3}m \;+\; n\log^{2}m \log n)$
time such that deletions take $O(\log n \log^{3} m)$ amortized time and GUDRE queries take $O(\log n \log^{2} m)$ time.
Our result is remarkable in three aspects.
First, compared to applying~\cite{ref:AgarwalIm18}, our data structure improves deletion time by a factor of $O(\log^{5}n)$ and improves query time by a factor of $O(\log n)$, while being deterministic rather than randomized (note that \cite{ref:AgarwalIm18} also supports insertions and more general queries).
Second, when applied to the Euclidean case ($m = O(1)$), our bounds match the optimal Euclidean results of~\cite{ref:EfratGe01,ref:WangCo26}.
Third, our result has an interesting \emph{implicit} nature: after the global preprocessing of $P$, the construction time for any given $S$ depends only logarithmically on $m$.

\paragraph{Subproblem 2: Implicit additively weighted geodesic Voronoi diagram.}
In this problem, we wish to construct a nearest neighbor data structure for a set $S$ of $n$ additively weighted points in a simple polygon $P$ with $m$ vertices, under the requirement that the structure be \emph{implicitly} constructed.
That is, after a global preprocessing of $P$, the time to build the data structure for $S$ should depend only logarithmically on $m$. Such a structure is referred to as an \emph{implicit Voronoi diagram}.

An \emph{explicit} diagram for $S$ can be built in $O((n+m)\log(n+m))$ time using the algorithm of Hershberger and Suri~\cite{ref:HershbergerAn99}, after which weighted nearest neighbor queries can be answered in $O(\log(n+m))$ time via standard point location queries~\cite{ref:KirkpatrickOp83,ref:EdelsbrunnerOp86}.
However, this explicit approach cannot be used in the implicit setting, as even constructing a point location data structure requires $\Omega(m)$ time.

For the unweighted version of this implicit problem, the approach of Agarwal, Arge, and Staals~\cite{ref:AgarwalIm18} can obtain the following: after $O(m)$-time global preprocessing, one can build a data structure for $S$ in $O(n \log n \log m + n\log^{3} m)$ time that answers nearest neighbor queries in $O(\log n \log^{2} m)$ time.

The weighted case is more challenging, as several geometric properties used in~\cite{ref:AgarwalIm18} no longer hold in the presence of weights. We develop a new approach and show that, after $O(m)$-time global preprocessing of $P$, a data structure for $S$ can be built in $O(n\log n \log^{2} m \;+\; n\log^{3} m)$
time, and each weighted nearest neighbor query can be answered in $O(\log n \log^{2} m)$ time.



\paragraph{Dynamic data structure for priority-queue updates.}
Our SSSP algorithm for the weighted simple-polygon case requires a dynamic 
data structure that supports insertions and delete-min operations, 
which we collectively refer to as \emph{priority-queue updates}. 
We present a result that extends Bentley’s logarithmic method~\cite{ref:BentleyDe79,ref:BentleyDe80} 
from supporting insertions only to supporting full priority-queue updates. 
As in Bentley’s original method, our result applies more broadly to other 
decomposable searching problems. 
We illustrate the technique using the nearest neighbor searching problem 
as an example.

Consider nearest neighbor searching for a set $S$ of $n$ points in the plane. 
Given a query point $q$, the goal is to determine the point of $S$ closest to $q$. 
In the static setting, where $S$ does not change, each query can be answered 
in $O(\log n)$ time after constructing the Voronoi diagram of $S$ 
in $O(n \log n)$ time.

Now suppose that each point of $S$ is assigned a \emph{key} (or weight), 
and the following updates are allowed: 
(1) insert a new point into $S$; and 
(2) delete the point of $S$ with the minimum key. 
If only insertions are supported, Bentley’s logarithmic method yields 
a dynamic structure with $O(\log^2 n)$ amortized insertion time and 
$O(\log^2 n)$ worst-case query time. 
To additionally support delete-min operations, the best existing approach 
is to use a fully dynamic nearest neighbor data structure that handles 
both insertions and deletions. 
Using the currently fastest such structures~\cite{ref:ChanDy20,ref:LiuNe22,ref:KaplanDy20}, 
one obtains $O(\log^2 n)$ amortized insertion time, 
$O(\log^4 n)$ amortized delete-min time, 
and $O(\log^2 n)$ worst-case query time.

We introduce a new technique that transforms any static data structure 
for a decomposable searching problem into a dynamic one supporting 
priority-queue updates. 
Applied to nearest neighbor searching, our method achieves 
$O(\log^2 n)$ amortized insertion time, 
$O(\log^2 n)$ amortized delete-min time (which improves the above bound 
by a factor of $\log^2 n$), and 
$O(\log^2 n)$ worst-case query time.

In general, suppose there exists a static data structure $\calD(S)$ for a decomposable problem 
on a set $S$ of $n$ elements with complexity $(P(n), U(n), Q(n))$, where $P(n)$ denotes 
the preprocessing time, $U(n)$ the space usage, and $Q(n)$ the query time. 
Assume that both $P(n)$ and $U(n)$ are at least linear in $n$, and that $Q(n)$ is monotone increasing. 
Further assume that each element of $S$ is assigned a key (or weight). 
Then $\calD(S)$ can be transformed into a dynamic data structure that supports 
priority queue operations with the following performance guarantees: 
insertion time $O(P(n)\log n / n)$, delete-min time $O(P(n)\log n / n)$, 
and query time $O(Q(n)\log n)$. The total space usage remains $O(U(n))$.

Note that the running time of our SSSP algorithm for the weighted 
simple polygon case improves upon the preliminary version of this paper~\cite{ref:BrewerSh26} 
by a logarithmic factor. 
In particular, the running time in~\cite{ref:BrewerSh26} was 
$O(m + n\log^3 n \log^2 m)$, 
and the improvement is due to this new dynamic data structure for priority-queue updates. 

As discussed in~\cite{ref:BentleyDe79,ref:BentleyDe80}, many searching problems are decomposable. We therefore anticipate further applications of our technique (see Section~\ref{sec:extension} for more discussions about this). In addition, our technique is conceptually simple and arguably elegant.

\paragraph{Outline.}
The remainder of the paper is organized as follows. After introducing notation and preliminary concepts in Section~\ref{sec:pre}, we present our algorithm for the weighted case in simple polygons in Section~\ref{sec:weight}. This algorithm relies on the two data structures developed for the subproblems discussed above as well as the technique for handling priority-queue updates, which is discussed in Section~\ref{sec:priorityqueue}.
Our data structure for deletion-only GUDRE queries is given in Section~\ref{sec:line_seperated_disk_emptiness}, and our solution to the implicit nearest neighbor query problem appears in Section~\ref{sec:ls_gawvd}.
For the unweighted SSSP problem, we describe our algorithm for simple polygons in Section~\ref{sec:simple_pol}, and extend it to polygons with holes in Section~\ref{sec:pol_holes}. Thus, Section~\ref{sec:pol_holes} is the only section dealing with polygons with holes.

\section{Preliminaries}
\label{sec:pre}

We follow the notation already introduced in Section~\ref{sec:intro}, including $m$, $n$, $S$, $P$, $s$, and $\G$. Throughout the paper, $P$ is a simple polygon unless stated otherwise (the only exception is Section~\ref{sec:pol_holes}, where $P$ may have holes).

We use $\overline{pq}$ to denote the line segment connecting points $p$ and $q$, and $|\overline{pq}|$ to denote its (Euclidean) length. For any geometric object $R$ in the plane, let $\partial R$ denote its boundary.

For any two points $s,t \in P$, let $\pi(s,t)$ denote a shortest \st\ path in $P$.
If multiple shortest paths exist (which can happen only when $P$ has holes), $\pi(s,t)$ may refer to an arbitrary one. Let $d(s,t)$ denote the geodesic distance between $s$ and $t$, i.e., the length of $\pi(s,t)$. We assume that $\pi(s,t)$ is oriented from $s$ to $t$. It is known that every interior vertex of $\pi(s,t)$ (i.e., except $s$ and $t$) is a vertex of $P$. The vertex of $\pi(s,t)$ adjacent to $t$ is called the {\em anchor} of $t$ in $\pi(s,t)$, and the vertex of $\pi(s,t)$ adjacent to $s$ is called the {\em anchor} of $s$.

\paragraph{Notation for the simple polygon case.}
When $P$ is a simple polygon, for any three points $s,t,r \in P$, the vertex of $\pi(s,t)\cap \pi(s,r)$ farthest from $s$ is called the {\em junction vertex} of $\pi(s,t)$ and $\pi(s,r)$.
For a point $s \in P$ and a segment $\overline{tr} \subseteq P$, let $c$ be the junction vertex of $\pi(s,t)$ and $\pi(s,r)$. Following standard terminology~\cite{ref:LeeEu84,ref:GuibasLi87}, the paths $\pi(c,t)$ and $\pi(c,r)$ together with $\overline{tr}$ form a {\em funnel}, denoted $F_s(\overline{tr})$; see Figure~\ref{fig:funnel}. The point $c$ is the {\em cusp} of the funnel, and $\pi(c,t)$ and $\pi(c,r)$ are the two {\em sides}, each of which is a concave chain.

\begin{figure}[t]
\begin{minipage}[t]{\linewidth}
\begin{center}
\includegraphics[totalheight=1.3in]{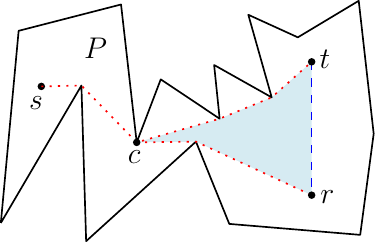}
\caption{Illustrating a the funnel $F_s(\overline{tr})$ (the light blue region).}
\label{fig:funnel}
\end{center}
\end{minipage}
\vspace*{-0.15in}
\end{figure}

When $P$ is a simple polygon, we will frequently use the data structure of Guibas and Hershberger~\cite{ref:GuibasOp89,ref:HershbergerA91}, referred to as the {\em GH data structure}. It can be built in $O(m)$ time and supports computing $d(s,t)$ in $O(\log m)$ time for any $s,t\in P$. It can also return the anchors of $s$ and $t$ on $\pi(s,t)$, and even a compact representation of $\pi(s,t)$ that supports binary search along the path. Furthermore, it can also compute the junction vertex of $\pi(s,t)$ and $\pi(s,r)$ for any three points $s,t,r\in P$.

\paragraph{Balanced polygon decomposition (BPD).}
For a simple polygon $P$, we will use the {\em balanced polygon decomposition} (BPD)~\cite{ref:GuibasLi87,ref:GuibasOp89,ref:ChazelleVi89}. A line segment inside $P$ connecting two vertices of $P$ that is not an edge of $P$ is called a {\em diagonal}. A classical result of Chazelle~\cite{ref:ChazelleCu93} shows that $P$ contains a diagonal $d$ that partitions $P$ into two subpolygons, each with at most $2n/3$ vertices. Recursively decomposing the subpolygons until each becomes a triangle yields the BPD.

The BPD can be represented by a tree $T_P$, called the {\em BPD-tree}. Each node $v$ of $T_P$ corresponds to a subpolygon $P_v$ of $P$: the root corresponds to $P$ itself, and each leaf corresponds to a triangle. Let $d_v$ denote the diagonal used to split $P_v$ into its two child subpolygons. The triangles at the leaves of $T_P$ form a triangulation of $P$. The entire decomposition and $T_P$ can be constructed in $O(m)$ time~\cite{ref:GuibasOp89}.

\paragraph{Geodesic Voronoi diagram.}
We will also make extensive use of the {\em geodesic Voronoi diagram} of $S$ in $P$~\cite{ref:AronovOn89,ref:OhOp19,ref:HershbergerAn99,ref:PapadopoulouA98}. For a site $u \in S$, its {\em Voronoi region} is defined as
\[
\R(u) = \{p \in P : d(u,p) \leq d(v,p) \text{ for all } v \in S \setminus \{u\}\}.
\]
Note that we treat $\R(u)$ as including its boundary. For two sites $u,v\in S$, the {\em Voronoi edge} $\E(u,v)$ is the locus of points $p\in P$ satisfying
\[
d(u,p) = d(v,p) \leq d(z,p)\quad\text{for all } z\in S\setminus\{u,v\}.
\]
If $P$ is a simple polygon, then $\E(u,v)$ has a single connected component~\cite{ref:AronovOn89}, while in general polygonal domains it may have multiple components~\cite{ref:HershbergerAn99}. Each connected component consists of a sequence of hyperbolic arcs (including line segments as a degenerate case). We do not use special notation for {\em Voronoi vertices}, i.e., points equidistant to at least three sites and strictly farther from all others. Let $\vd(S)$ denote the geodesic Voronoi diagram of $S$ in $P$.

The definitions extend naturally to the additively weighted case. Each site $u$ has a weight $w(u)$, and the {\em weighted distance} from $u$ to any $p\in P$ is $w(u)+d(u,p)$. Replacing geodesic distances by weighted distances yields the weighted geodesic Voronoi diagram. Note that all weights considered in this paper are additive.

\section{The weighted case in a simple polygon}
\label{sec:weight}

We now consider the SSSP problem in weighted geodesic unit-disk graphs within a simple polygon. Here, $S$ lies inside a simple polygon $P$, and for every edge $e(u,v)$ of $G(S)$ connecting two sites $u,v \in S$, the weight of $e(u,v)$ is defined to be the geodesic distance $d(u,v)$.

A natural starting point for designing an SSSP algorithm is to attempt to adapt the current best $O(n \log^2 n / \log \log n)$-time algorithm for weighted Euclidean unit-disk graphs~\cite{ref:WangNe20,ref:BrewerAn24}. However, that algorithm critically relies on constructing a uniform grid in the plane, and the technique that does not extend to the polygonal setting, since a grid cell may be fragmented into multiple disconnected components by the boundary of $P$. Another possible direction is to extend the $O(n^{1+\epsilon})$-time algorithm of Cabello and {Jej\v ci\v c}~\cite{ref:CabelloSh15}, which relies on maintaining a dynamic weighted {\em bichromatic closest pair (BCP)} data structure. When combined with the randomized dynamic BCP structures~\cite{ref:KaplanDy20,ref:LiuNe22}, the runtime improves to $O(n \log^4 n)$ expected time. Unfortunately, this line of attack does not seem to generalize to the polygon setting due to the absence of an efficient dynamic weighted BCP data structure under geodesic distance.

We instead take a different approach that avoids BCP altogether. Our algorithm runs in $O(m + n \log^3 n \log^2 m)$ deterministic time. Notably, when applied to the Euclidean case by treating $P$ as the whole plane (i.e., $m = O(1)$), the running time becomes $O(n \log^3 n)$, which is even faster than the best known approach based on Cabello and {Jej\v ci\v c}'s algorithm combined with dynamic BCP structures~\cite{ref:KaplanDy20,ref:LiuNe22}.

In what follows, we first describe the algorithm in Section~\ref{sec:weightdescription}, then prove its correctness in Section~\ref{sec:weightcorrect}, and finally discuss its implementation and time analysis in Section~\ref{sec:weighttime}

\subsection{Algorithm description}
\label{sec:weightdescription}

Given a source point $s \in S$, our goal is to compute shortest paths from $s$ to all other vertices $v$ in $G(S)$. Let $d_G(v)$ denote the shortest-path distance from $s$ to $v$ in $G(S)$. Our algorithm maintains an array $\dist[\cdot]$ such that, at termination, $\dist[v] = d_G(v)$ for all $v \in S$. With standard modifications, the predecessor information can also be maintained so that the shortest path tree from $s$ is obtained.

We maintain two disjoint subsets $A$ and $B$ of $S$. The set $A$ consists of all sites $v$ whose values $\dist[v]$ have already been finalized, and $B = S \setminus A$. We further partition $A$ into two subsets, $A_1$ and $A_2$, such that no site of $A_1$ is adjacent to any site of $B$ in $G(S)$. Initially, we set $\dist[s] = 0$, $A_2 = A = \{s\}$, $A_1 = \emptyset$, and $B = S \setminus \{s\}$. We also initialize $\dist[v] = \infty$ for every $v \in S \setminus \{s\}$.

In each iteration, we select
$a = \arg\min_{v \in A_2} \dist[v]$,
and perform the following two subroutines:

\begin{description}
    \item[Subroutine I.] Compute the set $B_a = \{\, b \in B : d(a,b) \le 1 \,\}$,
    i.e., all sites in $B$ that are adjacent to $a$ in $G(S)$.

    \item[Subroutine II.] For each site $b \in B_a$, compute
    $\sigma_b = \arg\min_{p \in A_2} \bigl( \dist[p] + d(p,b) \bigr)$,
    and set $\dist[b] = \dist[\sigma_b] + d(\sigma_b, b)$.
\end{description}

Afterward, we move all sites of $B_a$ from $B$ into $A_2$, and we move $a$ from $A_2$ into $A_1$. This completes one iteration. The algorithm terminates once $A_2$ becomes empty.

Although the algorithm looks simple, its correctness is not immediate. The main difficulty lies in implementing the two subroutines efficiently. Indeed, each subroutine corresponds precisely to one of the two subproblems introduced in Section~\ref{sec:results}.


\subsection{Algorithm correctness}
\label{sec:weightcorrect}

We now prove the correctness of the algorithm. In particular, we show that after the algorithm terminates, $\dist[v] = d_G(v)$ holds for every site $v \in S$. To this end, we establish that the algorithm maintains the following two invariants:

\begin{enumerate}
    \item[(I)] For every site $v \in A$, we have $\dist[v] = d_G(v)$.
    \item[(II)] For every site $v \in A_1$, every site $u$ adjacent to $v$ in $G(S)$ is already in $A$.
\end{enumerate}

Initially, we set $\dist[s] = 0$, $A_2 = A = \{s\}$, $A_1 = \emptyset$, and $B = S \setminus \{s\}$. It is immediate that both invariants hold at this point.

Consider any iteration of the algorithm, and assume that the invariants hold at the start of the iteration. At the end of the iteration, we move all sites of $B_a$ from $B$ into $A_2$, and we move $a$ from $A_2$ into $A_1$. By the definition of $B_a$, every neighbor of $a$ in $G(S)$ is now in $A$, and therefore Invariant~(II) continues to hold.

To verify Invariant~(I), it suffices to show that the values assigned during the iteration are correct for all sites newly added to $A$, namely those in $B_a$. This follows directly from the lemma below.

\begin{lemma}\label{lem:invariant}
For every site $b \in B_a$, we have $d_G(b) \;=\; \dist[\sigma_b] + d(\sigma_b, b)$.
\end{lemma}
\begin{proof}
Consider a point $b\in B_a$. Our goal is to prove $d_G(b)=\dist[\sigma_b]+d(\sigma_b,b)$.

We claim that $d(\sigma_b,b)\leq 1$. Indeed, since $a\in A_1$, by the definition of $\sigma_b$, $\dist[\sigma_b]+d(\sigma_b,b)\leq \dist[a]+d(a,b)$. By the definition of $a$, we have $\dist[a]\leq \dist[\sigma_b]$. Hence, we obtain $d(\sigma_b,b)\leq d(a,b)$. By the definition of $B_a$, $d(a,b)\leq 1$ holds. Therefore, $d(\sigma_b,b)\leq 1$.

Since $\sigma_b\in A_2\subseteq A$, according to the first algorithm invariant, we have $\dist[\sigma_b]=d_G(\sigma_b)$.

We now prove $d_G(b)=\dist[\sigma_b]+d(\sigma_b,b)$.
Assume to the contrary that $d_G(b)\neq \dist[\sigma_b]+d(\sigma_b,b)$. As $d(\sigma_b,b)\leq 1$, $\sigma_b$ is adjacent to $b$ in $G(S)$. Hence, $d_G(b)\leq d_G(\sigma_b)+d(\sigma_b,b)= \dist[\sigma_b]+d(\sigma_b,b)$. Since $d_G(b)\neq \dist[\sigma_b]+d(\sigma_b,b)$, we obtain
\begin{equation}\label{equ:contradict}
    d_G(b)< \dist[\sigma_b]+d(\sigma_b,b).
\end{equation}

Let $\pi_G(s,b)$ be a shortest path from $s$ to $b$ in $G(S)$ (with an orientation from $s$ to $b$). Let $v$ be the first vertex of $\pi_G(s,b)$ not in $A$. Since $b\not\in A$, such a vertex $v$ must exist. Let $u$ be the predecessor of $v$ in $\pi_G(s,b)$. See Figure~\ref{fig:path}. Hence, $u$ is adjacent to $v$ in $G(S)$ and $u\in A$.
Depending on whether $v$ is $b$, there are two cases.

\paragraph{The case $\boldsymbol{v=b}$.}
If $v=b$, then $d_G(b)=d_G(v)+d(v,b)$. By the definition of $\sigma_b$, we have $\dist[\sigma_b]+d(\sigma_b,b)\leq \dist[v]+d(v,b)$. Since $u\in A$, by our first algorithm invariant, $\dist[u]=d_G(u)$. Consequently, we can derive
\begin{equation*}
\begin{split}
    d_G(b)& =d_G(u)+d(u,b)
           =\dist[u]+d(u,b)
           \geq \dist[\sigma_b]+d(\sigma_b,b),
\end{split}
\end{equation*}
which contradicts with \eqref{equ:contradict}.

\paragraph{The case $\boldsymbol{v\neq b}$.}
If $v\neq b$, depending on whether $u$ is in $A_1$ or $A_2$ (recall that $u\in A$), there are further two subcases.

If $u\in A_1$, then by our second algorithm invariant, all vertices of $G(S)$ adjacent to $u$ must be in $A$. As $u$ is adjacent to $v$ in $G(S)$, $v$ must be in $A$, which contradicts the fact that $v\not\in A$. Hence, the case $u\in A_1$ cannot happen.

\begin{figure}[t]
\begin{minipage}[t]{\linewidth}
\begin{center}
\includegraphics[totalheight=1.0in]{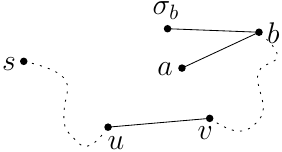}
\caption{The three solid segments represent edges in $G(S)$. The two dotted curves represent two subpaths of $\pi(s,b)$.}
\label{fig:path}
\end{center}
\end{minipage}
\vspace*{-0.15in}
\end{figure}

In the following, we assume that $u\in A_2$. By the definition of $a$, since $u\in A_2$, we have
\begin{equation}\label{equ:invariant20}
    \dist[a]\leq \dist[u].
\end{equation}

Since $v\neq b$ and $u$ is the predecessor of $v$ in $\pi(s,b)$, $u$ cannot be adjacent to $b$ in $G(S)$ since otherwise $u$ could connect to $b$ directly to shorten $\pi(s,b)$ by triangle inequality. Hence, $d(u,b)>1$. Since $d(a,b)\leq 1$, we have $d(a,b)<d(u,b)$. As both $a$ and $u$ are in $A_2$, by the definition of $a$, $\dist[a]\leq \dist[u]$ holds. We thus obtain
\begin{equation*}
    \dist[a]+d(a,b)<\dist[u]+d(u,b).
\end{equation*}

By the definition of $\sigma_b$, $\dist[\sigma_b]+d(\sigma_b,b)\leq \dist[a]+d(a,b)$. Hence, $\dist[\sigma_b]+d(\sigma_b,b) < \dist[u]+d(u,b)$.
Since $u\in A$, by our first algorithm invariant, $\dist[u]=d_G(u)$. Therefore,
\begin{equation*}
    \dist[\sigma_b]+d(\sigma_b,b) < d_G(u)+d(u,b)=d_G(b).
\end{equation*}
But this contradicts \eqref{equ:contradict}.
The lemma thus follows.
\end{proof}

In summary, both invariants continue to hold after each iteration. When the algorithm terminates, we have $A_2 = \emptyset$. We distinguish two cases depending on whether $B$ is empty.

If $B = \emptyset$, then every site of $S$ belongs to $A$. By Invariant~(I), we have $\dist[v] = d_G(v)$ for all $v \in S$, and hence the algorithm is correct.

If $B \neq \emptyset$, then since $A_2 = \emptyset$, Invariant~(II) implies that no site in $B$ is adjacent in $G(S)$ to any site in $A$. Consequently, no site in $B$ is reachable from $s$ in $G(S)$, and therefore $d_G(v) = \infty$ for all $v \in B$. Thus, the algorithm is also correct in this case.


\subsection{Algorithm implementation and time analysis}
\label{sec:weighttime}

We maintain the set $A_2$ in a min-heap with $\dist[v]$ as the key for all $v \in A_2$. Thus, extracting $a = \arg\min_{v \in A_2} \dist[v]$ takes $O(\log n)$ time, and each insertion into $A_2$ also takes $O(\log n)$ time. We now discuss the implementation of the two subroutines.

\paragraph{Subroutine I.}
To compute $B_a$, we repeatedly locate a point $b \in B$ with $d(a,b) \le 1$ and delete it from $B$, stopping once no such point remains. Thus, the task reduces exactly to supporting \emph{deletion-only geodesic unit-disk range emptiness} (GUDRE) queries, as introduced in Section~\ref{sec:results}. Given a set $B$ of points in $P$, we require a data structure that supports:

\begin{description}
    \item[Delete:] Remove a point from $B$.
    \item[GUDRE query:] Given a point $a \in P$, determine whether there exists $b \in B$ with $d(a,b) \le 1$, and if so, return such a point.
\end{description}

One of our main results is the following theorem, proved in Section~\ref{sec:line_seperated_disk_emptiness}.

\begin{theorem}\label{theo:rangeempty}
Let $P$ be a simple polygon with $m$ vertices, and assume that a GH data structure for $P$ has been built in $O(m)$ time. For any set $S$ of $n$ points in $P$, we can construct a geodesic unit-disk range emptiness data structure for $S$ in $O(n\log^3 m + n\log n \log^2 m)$ time such that each deletion takes $O(\log n \log^3 m)$ amortized time and each query takes $O(\log n \log^2 m)$ worst-case time.
\end{theorem}

Using Theorem~\ref{theo:rangeempty}, we implement Subroutine~I as follows.
Initially, $B = S \setminus \{s\}$.
We build the GUDRE data structure for $B$ in
$O(m + n\log^3 m + n\log n\log^2 m)$
time, where the $O(m)$ term is for constructing the GH data structure.
Since the algorithm performs $O(n)$ deletions and $O(n)$ GUDRE queries overall, the total running time spent in Subroutine~I is
$O(m + n \log n \log^3 m)$.

\paragraph{Subroutine II.}
For the second subroutine, observe that each $\sigma_b$ is essentially the (additively) \emph{weighted nearest neighbor} of $b$ in $A_2$, where each site $p \in A_2$ has weight $\dist[p]$, and the weighted distance from $p$ to a query point $q$ is defined as $\dist[p] + d(p,q)$. Thus, we require a weighted nearest neighbor data structure for $A_2$.

Furthermore, since at the end of each iteration, we will move all points of $B_a$ from $B$ to $A_2$ and move $a$ from $A_2$ to $A_1$, we need a dynamic data structure for $A_2$ that supports both insertions and deletions.
A general dynamic data structure supporting both insertions and deletions appears difficult to obtain. However, we make the observation that when we delete $a$ from $A_2$, $a$ is the site in $A_2$ with minimum weight.
Therefore, our data structure for Subroutine II needs to support the following operations
(1) insert a site into $A_2$, (2) delete the site from $A_2$ with minimum weight (referred to as {\em delete-min} operations),
and (3) for a query point $b$, return the weighted nearest neighbor $\sigma_b$ of $b$ in $A_2$.

We will propose a new technique in Section~\ref{sec:priorityqueue} to handle the three operations by extending Bently's logarithmic method~\cite{ref:BentleyDe80,ref:BentleyDe79}. To this end, we need a static geodesic Voronoi diagram for $A_2$ that can be implicitly built in the sense that its construction time depends only logarithmically on $m$. To explain this, consider a simplified situation where we only consider the insertion and query operations on $A_2$ as above. Then, we could use Bently's logarithmic method together with a static geodesic Voronoi diagram for $A_2$. Since computing the geodesic Voronoi diagram of $A_v$ in $P$ requires $\Omega(m)$ time (note that the best algorithm can compute the diagram in $O(m + |A_v| \log |A_v|)$ time~\cite{ref:OhOp19}), rebuilding these diagrams $n$ times would lead to $\Omega(mn)$ total time, which is too slow. Therefore, we require an \emph{implicit} construction whose dependence on $m$ is only polylogarithmic. This is exactly the subproblem introduced in Section~\ref{sec:results}. In Section~\ref{sec:ls_gawvd}, we establish the following result.

\begin{theorem}\label{theo:nearneighbor}
Let $P$ be a simple polygon with $m$ vertices, and assume that a GH data structure for $P$ has been built in $O(m)$ time. Given a set $S$ of $n$ weighted sites in $P$, we can construct a weighted nearest neighbor data structure for $S$ in $O(n \log^3 m + n \log n \log^2 m)$ time such that each query can be answered in $O(\log n \log^2 m)$ time.
\end{theorem}

With Theorem~\ref{theo:nearneighbor}, the technique in Section~\ref{sec:priorityqueue} obtains the following result. 

\begin{theorem}\label{theo:deletemin}
Let $P$ be a simple polygon with $m$ vertices, and assume that a GH data structure for $P$ has been built in $O(m)$ time.
We can maintain a dynamic set $S$ of $n$ weighted points in $P$,  with  $O(\log n \log^3 m + \log^2 n \log^2 m)$ amortized insertion and delete-min time, so that each weighted nearest neighbor query on $S$ can be answered in $O(\log^2 n\log^2 m)$ worst-case time.
\end{theorem}

With Theorem~\ref{theo:deletemin}, the total running time of Subroutine~II over the entire algorithm is $O(m + n\log n \log^3 m + n\log^2 n\log^2 m)$, including the $O(m)$ time required to build the GH data structure for $P$.

\paragraph{Summary.}
Combining the above discussions, we obtain the following theorem.

\begin{theorem}
The SSSP problem in the weighted geodesic unit-disk graph in a simple polygon of $m$ vertices for a set of $n$ points can be solved in $O(m + n \log^2 n \log^2 m)$ time.
\end{theorem}
\begin{proof}
The total time of Subroutine~I is $O(m + n \log n \log^3 m)$ while the total time of Subroutine~II is 
$O(m + n \log n \log^3 m + n \log^2 n \log^2 m)$. 
We now show that $n \log^2 n \log^2 m=O(m + n \log^2 n \log^2 m)$. Indeed, if $m < n^2$, then $\log m = O(\log n)$, and thus $n \log^2 n \log^2 m = O(n \log^2 n \log^2 m)$. If $m \ge n^2$, then $\log m = \Omega(\log n)$ and $n \log^2 n \log^2 m = O(m)$. In either case, the bound $O(m + n \log^2 n \log^2 m)$ holds. Hence, the running time of the whole algorithm is $O(m + n \log^2 n \log^2 m)$, as claimed.
\end{proof}

\section{A dynamic data structure for priority-queue updates}
\label{sec:priorityqueue}

In this section, we prove Theorem~\ref{theo:deletemin} using the result of 
Theorem~\ref{theo:nearneighbor}. 
Our technique is quite general: it extends Bentley’s logarithmic method 
from supporting insertions only to supporting both insertions and 
delete-min operations, which we refer to as \emph{priority-queue updates}.
As in Bentley’s original framework, the technique applies more broadly 
to other decomposable searching problems. 
We first present the proof of Theorem~\ref{theo:deletemin} 
and then discuss how the technique can be generalized.

We follow the notation in Theorem~\ref{theo:deletemin}. We assume that each point $p$ of $S$ has a key, denoted by $key(p)$, which is equal to its weight $\dist[p]$ Subroutine~II; but in general, the key can be any number. 
We consider the following operations on $S$. 

\begin{description}
    \item[Query:] Given a query point $q$, compute its weighted nearest neighbor in $S$. 
    \item[Insertion:] Insert a new point into $S$. 
    \item[Delete-Min:] Find a point of $S$ with the smallest key and delete it from $S$. 
\end{description}

We begin with describing the data structure. 

\paragraph{Data structure description.}
Let $n = |S|$ and $\Psi(S)$ denote our data structure for $S$. 
We partition $S$ into subsets $S_0, S_1, \ldots, S_\clg$. For each $S_i$, $0\leq i \leq \clg$, we compute the data structure of Theorem~\ref{theo:nearneighbor} on $S_i$ and use $\calD(S_i)$ to denote it.
For each $S_i$, we also maintain a subset $C_i \subseteq S_i$. 

For ease of exposition, we assume that $|S|\geq 16$ (otherwise we could simply handle the operations by brute-force) and thus $\clg\geq 4$. For any $j$, define $S_{\ge j}=\bigcup_{i= j}^{\clg}S_i$. 
Our data structure maintains the following invariants:

\begin{enumerate}    
    \item For all $0\leq i \leq \clg$, $|S_i| \leq 2^i$.
    \item $|C_0| = 0$, $|C_1|=1$, and $2^{i - 2} \leq |C_i| \leq 2^{i - 1}$ for all $2\leq i\leq  \clg-1$. For $C_\clg$, we only require an upper bound $|C_{\clg}|\leq 2^{\clg-1}$.
    \item For any $0\leq i\leq \clg$, for all points $p \in C_i$ and all points $q \in S_{\ge i}\setminus C_i$, $key(p)\leq key(q)$. In other words, $C_i$ contains $|C_i|$ points of $S_{\ge i}$ with the smallest keys.
    \item For any $i > \clg$, $S_i = \emptyset$.
\end{enumerate}

Note that by our invariants $S_0$ contains at most one point and $C_1$ contains exactly one point. Furthermore, either the point in $C_1$ or the point in $S_0$ (if $S_0\neq \emptyset$) has the smallest key of $S$. Hence, we can find a point of $S$ with the smallest key in $O(1)$ time. 

\paragraph{Queries.}
The query algorithm is straightforward (refer to Algorithm~\ref{algo:A2_query} for the pseudocode). Given a query point $q$, our goal is to compute its weighted nearest neighbor $\beta_q(S)$ in $S$. For each $0\leq i \leq \clg$, we find $q$'s nearest neighbor in $S_i$ using the data structure $\calD(S_i)$. Finally, among the nearest neighbors for all $S_i$, the one whose weighted distance from $q$ is the smallest is $\beta_q(S)$. The query time is $O(\log^2 n\log^2 m)$ by Theorem~\ref{theo:nearneighbor}. 

\begin{algorithm}
    \caption{The query algorithm.} \label{algo:A2_query}
    \KwIn{$\Psi(S)$, and a point $q \in P$.}
    \KwOut{Weighted nearest neighbor of $q$ in $S$.}
    $Q \gets \emptyset$\;
    \For{$i =0, \ldots, \clg$}{
        $p_i \gets$ weighted nearest neighbor of $q$ in $S_i$\;
        add $p_i$ to $Q$\;        
    }
    $\beta_q(S) \gets$ weighted nearest neighbor of $q$ in $Q$\;
    \Return{$\beta_q(S)$\;}
\end{algorithm}

In the following, we describe the insertion algorithm in Section~\ref{sec:insert}. The algorithm for the delete-min operation is in Section~\ref{sec:deletemin}. We analyze the time complexities of the two update operations in Section~\ref{sec:time}. 

\subsection{Insertion}
\label{sec:insert}

To insert a point $b$ to $S$, our algorithm works as follows (refer to Algorithm~\ref{algo:A2_insert} for the pseudocode). 


We start with incrementing $n$ by one and then adding $b$ to $S_0$.
Next, for each $i$ starting from $0$, while $|S_i|> 2^i$, we perform the following {\em insertion-induced rebalance procedure}. 

We first move the points of $S_i\setminus C_i$ to $S_{i+1}$ and reset $S_i=C_i$. Then, we rebuild $\calD(S_i)$ by calling the construction algorithm of $\calD(S_i)$ on $S_i$. In addition, we update $C_{i+1}$ as follows. Let $c_{i+1}$ be the current size of $C_{i+1}$. 
We update $C_{i+1}$ by first emptying it and then adding to it $c_{i+1}$ points of $S_{i+1}$ with the smallest keys. As such, $C_{i+1}$ does not change its size although its content may be changed. A special case happens if $i=\clg -2$ and $\clg = \lceil\log (n-1)\rceil+1$. In this case, we know $|C_{i+1}| \leq 2^i$ by the second invariant before the insertion but we do not have a lower bound for it. In order to re-establish the second invariant after the insertion, since now $i+1=\clg -1$, we need to guarantee $2^{i-1}\leq |C_{i+1}|\leq 2^i$. To this end, we do the following. 
As $|S_i|> 2^i$ and $|C_i|\leq 2^{i-1}$, $S_i\setminus C_i$ has at least $2^{i-1}$ points. Since all points of $S_i\setminus C_i$ are moved to 
$S_{i+1}$, we have $|S_{i+1}|\geq 2^{i-1}$. 
Hence, if $c_{i+1}<2^{i-1}$, 
we first reset $c_{i+1}$ to $2^{i-1}$ and then perform the above update for $C_{i+1}$ using the new $c_{i+1}$; otherwise, we perform the update using the original $c_{i+1}$ as above. This finishes the rebalance procedure for the current iteration $i$. 

If $|S_i|\leq 2^i$, then we rebuild $\calD(S_i)$ and finish the insertion algorithm. Note that $|S_\clg|\leq n\leq 2^{\clg}$ and thus the relalance procedure can only happen in the $i$-th iterations for $i\leq \clg-1$. In addition, we make the following observation that will be used later. 
\begin{observation}\label{obser:insert}
For any $i$ with $0\leq i \leq \clg$ and $i\neq \clg-1$, $|C_i|$ does not change after the insertion. For $i=\clg-1$, if $\clg=\lceil \log (n-1)\rceil$, then $|C_i|$ also does not change; otherwise, $|C_i|$ either does not change or increases to $2^{i-2}$. 
\end{observation}


\begin{algorithm}
    \caption{Inserting a new point $b$ to $S$.} \label{algo:A2_insert}
    \KwIn{$\Psi(S)$, a point $b \notin S$.}
    \KwOut{$\Psi(S \cup \{b\})$.}
    $n \gets n + 1$\;
    $S_0 \gets S_0 \cup \{b\}$\;
    $i\gets 0$\;
    \While{$|S_i| > 2^i$}
    {            
        $S_{i + 1} \gets S_{i + 1} \cup (S_i \setminus C_i)$\label{ln:defineS}\;
        $S_i \gets C_i$\;
        Rebuild $\calD(S_i)$\label{line:insert_body_end}\;
        $c_{i+1}\gets |C_{i+1}|$\;
        \If{$i=\clg -2$ and $\clg = \lceil\log (n-1)\rceil+1$ and $c_{i+1}<2^{i-1}$}{$c_{i+1}\gets 2^{i-1}$\;}
        $C_{i + 1} \gets$ the set of $c_{i+1}$ points in $S_{i + 1}$ with the smallest keys\label{ln:Ci1}\;                
        $i\gets i+1$\;
    }
    Rebuild $\calD(S_i)$\label{line:insert_body_end_00}\;
\end{algorithm}


The following lemma shows that all invariants are maintained after the insertion.
\begin{lemma}
All algorithm invariants are maintained after the insertion.
\end{lemma}
\begin{proof}
We assume that the invariants hold before the insertion with $n-1$ (where $n$ represents the size of $S$ after the insertion).

\paragraph{The second invariant.}
We start with the second invariant. 
Consider any $i$ with $0\leq i\leq \clg$. 

If $i\leq \clg-2$ or if $i=\clg-1$ and $\clg=\lceil \log (n-1)\rceil$, then by Observation~\ref{obser:insert}, $|C_i|$ does not change due to the insertion. By the second invariant before the insertion, we have $2^{i-2}\leq |C_i|\leq 2^{i-1}$. 

If $i=\clg-1$ and $\clg\neq \lceil \log (n-1)\rceil$, then by Observation~\ref{obser:insert}, $|C_i|$ either does not change or increases to $2^{i-2}$. In the latter case, it is obviously true that $|C_i|$ satisfies the second invariant. In the first case, since $\clg\neq \lceil \log (n-1)\rceil$, we have $\clg=\lceil \log (n-1)\rceil+1$ and thus 
$i=\lceil \log (n-1)\rceil-2$. By the second invariant before the insertion, $2^{i-2}\leq |C_i|\leq 2^{i-1}$.

Finally, if $i=\clg$, then by Observation~\ref{obser:insert}, $|C_i|$ does not change. If $\clg= \lceil \log (n-1)\rceil$, then by the second invariant before the insertion, $|C_\clg|=|C_{\lceil \log (n-1)\rceil}|\leq 2^{\lceil \log (n-1)\rceil-1}=2^{\clg-1}$. Otherwise, $S_\clg$ is a new subset created by the insertion algorithm and the algorithm does not add any point to $C_\clg$ (to see this, when $i+1=\clg$ in the while loop, we have $c_{i+1}=0$ at Line~\ref{ln:Ci1}),                                                    and thus $|C_\clg|=0\leq 2^{\clg-1}$. 

This proves that the second invariant still holds.

\paragraph{The first invariant.}
For the first invariant, notice first that $|S_\clg|\leq n\leq 2^{\clg}$. For all $i =0, \ldots, \clg-1$, $|S_i|>2^i$ can only happen if the while loop has the $i$-th iteration, in which 
$S_i$ is updated to $C_i$. Therefore, it suffices to argue that $|C_i|\leq 2^i$. 
By the second invariant, we have $|C_i|\leq 2^{i-1}$. Therefore, the first invariant holds. 


\paragraph{The fourth invariant.}
Because the fourth invariant holds before the insertion, 
$S_{\clg + 1}$ is not empty only if $|S_i|>2^i$ for $i=\clg$. But that would mean $|S_i|>n$, which is not possible. Hence, $S_{\clg + 1}$ must be empty after the insertion. Thus, the fourth invariant holds. 

\paragraph{The third invariant.}
For the third invariant, it suffices to argue that right after the $i$-th iteration, $0\leq i\leq \clg-1$, the invariant holds for $C_{i+1}$ since afterwards neither $C_{i+1}$ nor $S_{\geq i+1}$ will change (recall that $C_0=\emptyset$ by the second invariant, so we do not need to argue the case for $C_0$). 
Consider two points $p\in C_{i+1}$ and $q\in S_{\geq i+1}\setminus C_{i+1}$. Our goal is to show that $key(p)\leq key(q)$. In the following, we consider the situation right after the $i$-th iteration. 

Let $C_{i+1}'$ denote the original $C_{i+1}$ before the insertion. Similarly, let $S_{i+1}'$ denote the original $S_{i+1}$ before the insertion. 
Note that $S_j$ for all $j\geq i+2$ is the same as before the insertion because we are considering the situation right after the $i$-th iteration. Hence, $S_{\geq i+2}$ is the same as before the insertion. We distinguish two cases depending on whether $q\in S_{i+1}$. 

\begin{itemize}
    \item 
If $q\in S_{i+1}$, then since $p\in C_{i+1}$, $q\not\in C_{i+1}$, and $C_{i+1}$ is the subset of $|C_{i+1}|$ points of $S_{i+1}$ with the smallest keys (Line~\ref{ln:Ci1}), $key(p)\leq key(q)$ must hold.

\item 
If $q\not\in S_{i+1}$, then $q\in S_{\geq i+2}$. This implies that $i\leq \clg-2\leq \lceil \log(n-1)\rceil-1$. 
By Observation~\ref{obser:insert}, $|C'_{i+1}|=|C_{i+1}|$.
We further have two subcases depending on whether $p\in C'_{i+1}$. 

\begin{itemize}
\item 
If $p\in C'_{i+1}$, then since $S_{\geq i+2}$ is the same as before the insertion, by the third invariant before the insertion, $key(p)\leq key(q)$ holds. 

\item 
If $p\not\in C'_{i+1}$, then since $S'_{i+1}\subseteq S_{i+1}$ by Line~\ref{ln:defineS} and $|C'_{i+1}|=|C_{i+1}|$, $key(p)\leq key(p')$, where $p'$ is a point of $C'_{i+1}$ with the largest key. Because $p'\in C'_{i+1}$ and $q\in S_{\geq i+2}$, by the third invariant before the insertion, $key(p')\leq key(q)$ holds. Therefore, we obtain $key(p)\leq key(q)$. 
\end{itemize}
\end{itemize}

This proves $key(p)\leq key(q)$. The third invariant thus follows.  
\end{proof}

\subsection{Delete-Min}
\label{sec:deletemin}

We now discuss the delete-min operation. 
Between the point of $C_1$ and the point of $S_0$ (if $S_0\neq\emptyset$), let $a$ denote the one with a smaller key. As discussed before, $a$ has the smallest key among all points of $S$. 
In the following, we discuss how to delete $a$ from $S$ (refer to Algorithm~\ref{algo:A2_delete} for the pseudocode). 

We first decrement $n$ by one. If $S_{\clg+1} \neq \emptyset$, then the decrement of $n$ decreases $\clg$ by one. To maintain the fourth invariant, we merge $S_{\clg+1}$ with $S_{\clg}$, and set both $S_{\clg+1}$ and $C_{\clg+1}$ to $\emptyset$. We then rebuild $\calD(S_{\clg})$. We refer to this as the {\em merging process}. After that, we proceed as follows. 

If $a\in S_0$, then we remove $a$ from $S_0$, which effectively reset $S_0$ to $\emptyset$; in this case, we are done with the deletion. Otherwise, $a\in C_1$ and we remove $a$ from both $C_1$ and $S_1$ (this resets $C_1$ to $\emptyset$). Next we run a while loop. 
Starting from $i=1$, while $i<\clg$ and $|C_i|<2^{i-2}$, we perform the following {\em deletion-induced reblance procedure}. 
We define a number $c_i$ depending on whether $i\leq \clg-2$ or $i=\clg-1$. 

\begin{itemize}
\item 
If $i\leq \clg-2$, then by the second invariant before the deletion, $|C_{i+1}|\geq 2^{i-1}$ holds. As $C_{i+1}\subseteq S_{i+1}$, we have $|S_{i+1}|\geq 2^{i-1}$. In this case, we let $c_i=2^{i-1}$. 

\item 
If $i= \clg -1$, then Lemma~\ref{lem:sizeunion} proves that $|S_i\cup S_{i+1}|\geq 2^{i-2}$. In this case, we let $c_i=2^{i-2}$.
\end{itemize}

Hence, in both cases it holds that $|S_i\cup S_{i+1}|\geq c_i$. We update $C_i$ by first emptying it and then adding $c_i$ points from $S_i\cup S_{i+1}$ with the smallest keys. 
Note that $C_i$ may contain points of $S_{i+1}$, meaning that these points should be moved from $S_{i+1}$ to $C_i$. Hence, we update $S_i=S_i\cup C_i$, $C_{i+1} = C_{i+1}\setminus C_i$, and $S_{i+1} = S_{i+1}\setminus C_i$. After the above update for $S_i$, $|S_i|-k_i$ is the number of points of $S_{i+1}$ that have been moved to $S_i$, where $k_i$ denotes the original size of $S_i$ before the update. We arbitrarily move the same number of points from $S_i\setminus C_i$ to $S_{i+1}$ if possible (i.e., if $|S_i\setminus C_i|\geq |S_i|-k_i$). If $|S_i\setminus C_i|< |S_i|-k_i$, we move all points of $S_i\setminus C_i$ to $S_{i+1}$, after which $S_i=C_i$. In either case, let $Q$ be the set of points of $S_i\setminus C_i$ that have been moved to $S_{i+1}$. We add to $C_{i+1}$ all points of $Q$ whose keys are smaller than the largest key of $C_{i+1}$ (note that if $C_{i+1}=\emptyset$, then we do not add any points to $C_{i+1}$ in this step). Finally, we rebuild $\calD(S_i)$. This finishes the rebalance procedure. 

After the while loop, we rebuild $\calD(S_i)$ for the first $i$ that fails the condition of the while loop. This finishes the deletion algorithm. 

\begin{algorithm}
    \caption{Deleting $a$ from $S$ with either $\{a\} = S_0$ or $\{a\} = C_1$.} \label{algo:A2_delete}
    \KwIn{$\Psi(S)$, point $a$ with either $\{a\} = S_0$ or $\{a\} = C_1$.}
    \KwOut{$\Psi(S \setminus \{a\})$.}
    $n \gets n - 1$\;
    \If{$S_{\clg+1}\neq \emptyset$}{    
        $S_\clg \gets S_\clg \cup S_{\clg + 1}$\label{line:delete_endcheck_start}\;
        $S_{\clg + 1} \gets \emptyset$\;
        $C_{\clg + 1} \gets \emptyset$\label{ln:Cempty}\;        
        Rebuild $\calD(S_\clg)$\label{line:delete_endcheck_end}\;
    }   
    \If{$\{a\} = S_0$}{
        $S_0 \gets \emptyset$\;
        \Return\;
    }
    \tcp{The following is for the case $\{a\} = C_1$}
    Remove $a$ from both $C_1$ and $S_1$\label{ln:setC1}\;
    $i\gets 1$\;
    \While{$i<\clg$ and $|C_i| < 2^{i - 2}$}{                
        \If{$i\leq \clg -  2$}{
            $c_i=2^{i - 1}$\;
        }
        \Else(\tcc*[h]{$i=\clg -1$}){        
            $c_i=2^{i-2}$\;
        }
            $k_i \gets |S_i|$\label{line:delete_body_start}\;
            $C_i \gets$ the set of $c_i$ points from $S_i \cup S_{i + 1}$ with minimum keys\label{line:addCj}\;
            $S_i \gets S_i \cup C_i$\label{line:addAj}\;
            $C_{i + 1} \gets C_{i + 1} \setminus C_i$ \label{line:removeCj1}\;
            $S_{i+1} \gets S_{i+1} \setminus C_i$ \label{line:removeAj1}\;
            $Q \gets$ an arbitrary set of $\min \{|S_i \setminus C_i|, |S_i| - k_i\}$ points in $S_i \setminus C_i$\label{line:min}\;
            $S_i \gets S_i \setminus Q$\label{line:removeAj}\;
            $S_{i + 1} \gets S_{i + 1} \cup Q$\label{line:addAj1}\;
            $Q' \gets$ the set of points in $Q$ whose keys are smaller than the largest key in $C_{i + 1}$\label{ln:largestkey}\;
            $C_{i + 1} \gets C_{i + 1} \cup Q'$\label{line:addCj1}\;
            Rebuild $\calD(S_i)$\label{line:delete_body_end}\;
            $i \gets i + 1$\;
        }
    Rebuild $\calD(S_i)$\label{line:delete_body_end_00}\; 
\end{algorithm}


The following observation will be used later.
\begin{observation}\label{obser:delete}
Right after the $i$-th iteration of the while loop, we have the following: (1) $|C_i|=2^{i-1}$ for any $1\leq i\leq \clg-2$ and  $|C_i|=2^{i-2}$ for $i=\clg-1$; (2) $|S_i|$ either does not change or is equal to $|C_i|$, and if $|S_{i+1}|$ is changed, then $|S_i|=|C_i|$; and (3) neither $|S_{i+1}|$ nor $|C_{i+1}|$ can become larger.  
\end{observation}
\begin{proof}
According to our algorithm, $|C_i|=c_i$ and $|S_i|-k_i$ points have been moved from $S_{i+1}$ to $S_i$. But at most that many points are moved from $S_i\setminus C_i$ to $S_{i+1}$ (resp., $C_{i+1}$). If exactly that many points are moved out of $S_i$, then $|S_i|$ is the same as in the beginning of the iteration; otherwise, $S_i\setminus C_i=\emptyset$ and thus $|S_i|=|C_i|$. In addition, if $|S_{i+1}|$ is changed, then less than $|S_i|-k$ points have been moved to $S_{i+1}$, implying that $S_i\setminus C_i=\emptyset$ and $|S_i|=|C_i|$. 

Since $|S_i|-k_i$ points have been moved out of $S_{i+1}$ while at most that many points are moved into $S_{i+1}$, it follows that $|S_{i+1}|$ cannot get larger. In what follows, we argue that $|C_{i+1}|$ does not get larger either. 

Indeed, if $C_{i+1}=\emptyset$ at Line~\ref{ln:largestkey}, then no point will be added to $C_{i+1}$ and thus $|C_{i+1}|=0$. Hence, in this case it is vacuously true that $|C_{i+1}|$ cannot get larger than before. We next assume that $C_{i+1}\neq \emptyset$. Let $p_{i+1}$ be a point of $C_{i+1}$ with the largest key at Line~\ref{ln:largestkey}. Let $p_i$ be a point of $C_i$ with the largest key after $C_i$ is updated in Line~\ref{line:addCj}. By the way we update $C_i$ at Line~\ref{line:addCj}, $key(p_i)\leq key(p_{i+1})$. If $Q'=\emptyset$, then no point from $S_i$ will be added to $C_{i+1}$ and thus $|C_{i+1}|$ does not get larger than before the deletion. 
We next consider the case where $Q'\neq \emptyset$. 

Let $p$ be a point from $Q'$. By definition, $key(p)< key(p_{i+1})$ and $key(p_i)\leq key(p)$. It follows that $key(p_i)<key(p_{i+1})$. Let $p'$ be any point of $S_{i+1}$ that is moved into $C_{i}$ in Line~\ref{line:addCj}. By definition, $key(p')\leq key(p_i)$. Thus, $key(p')<key(p_{i+1})$. Therefore, $p'$ must be in the original $C_{i+1}$ in the beginning of the $i$-th iteration. This implies that all points of $S_{i+1}$ that are moved to $S_i$ in Line~\ref{ln:largestkey} are originally from $C_{i+1}$ and thus $|S_i|-k$ points have been moved out of $C_{i+1}$. Because at most $|S_i|-k$ points can be moved into $S_{i+1}$, at most that many points can be moved into $C_{i+1}$ at Line~\ref{line:addCj1}. Hence, $|C_{i+1}|$ cannot get larger than before the deletion. 
\end{proof}

\begin{lemma}\label{lem:sizeunion}
If $i=\clg-1$ in the while loop, then $|S_i\cup S_{i+1}|\geq 2^{i-2}$. 
\end{lemma}
\begin{proof}
According to our decision algorithm, $|S_i|$ cannot be changed in the $j$-th iteration for any $j<i-1$, but it can be changed in the $(i-1)$-th iteration. 
We distinguish two cases depending on whether $|S_i|$ is changed in the $(i-1)$-the iteration. 

If $|S_i|$ is not changed, then by our second invariant before the deletion, $|S_i|\geq |C_i|\geq 2^{i-2}$. Hence, in this case it is true that $|S_i\cup S_{i+1}|\geq 2^{i-2}$. 

We now assume that $|S_i|$ is changed in the $(i-1)$-th iteration. Then, by Observation~\ref{obser:delete} (applying to $i-1$), $|S_{i-1}|=|C_{i-1}|=2^{i-2}$. In addition, for each $0\leq j\leq i-2$, due to the first invariant and Observation~\ref{obser:delete}, it holds that $|S_j|\leq 2^{j}$. Hence, at the beginning of the $i$-th iteration, we have 
\begin{equation}\label{equ:10}
    \begin{split}
        \sum_{j=0}^{i-1}|S_j|=\sum_{j=0}^{i-2}|S_j|+|S_{i-1}|\leq \sum_{j=0}^{i-2}2^j+2^{i-2}=2^{i-1}-1+2^{i-2}. 
    \end{split}
\end{equation}

Note that due to the merging process in the beginning of the deletion algorithm, $S_{i+2}=\emptyset$. Hence, $n=|S|=\sum_{j=0}^{i+1}|S_j|$. 

Assume to the contrary that $|S_i\cup S_{i+1}|< 2^{i-2}$. Then, with \eqref{equ:10}, we have $\sum_{j=0}^{i+1}|S_i|<2^{i-1}-1+2^{i-2}+2^{i-2}=2^i-1$. As $i=\clg-1\leq \lfloor\log n\rfloor$, $2^i\leq n$. Consequently, we obtain $\sum_{j=0}^{i+1}|S_i|<n-1$. But this contradicts $n=\sum_{j=0}^{i+1}|S_i|$. Therefore, $|S_i\cup S_{i+1}|\geq 2^{i-2}$ must hold. 
\end{proof}

\begin{lemma}\label{lem:deleteinvariant}
All algorithm invariants are maintained after the deletion.
\end{lemma}
\begin{proof}
We assume that the invariants hold before the deletion with $n+1$ (where $n$ represents the size of $S$ after the deletion).

First of all, the merging process resets $S_{\clg+1}=\emptyset$, which establishes the fourth invariant. The process also increases the size of $S_{\clg}$. Clearly, $|S_{\clg}|\leq n\leq 2^{\clg}$. Hence, the first invariant also holds. So do the other two invariants. 

If $\{a\}=S_0$, then the only change by the algorithm is to set $S_0$ to $\emptyset$. Therefore, all invariants hold in this case. In the following, we assume that $\{a\}=C_1$. In this case, $S_0$ does not change. 

\paragraph{The first invariant.}
For the first invariant, since $S_0$ does not change, it suffices to show that $|S_i|\leq 2^i$ for all $1\leq i\leq \clg$. For $i=\clg$, we simply have $|S_i|\leq n\leq 2^i$. In the following, we assume that $1\leq i\leq \clg-1$. 

Observe that $S_i$ will not change in any iteration of the while loop after the $i$-th iteration. 
Therefore, we argue that right after the $i$-th iteration, $|S_i| \leq 2^i$ holds. According to our algorithm, $|S_i|$ can only be possibly changed in either the $i$-th iteration or the $(i-1)$-th iteration. By Observation~\ref{obser:delete}, after either iteration, $|S_i|$ is $2^{i-1}$, $2^{i-2}$, or the same as before (and thus $|S_i|\leq 2^i$ by our first invariant). Therefore, after the $i$-th iteration, it holds that $|S_i|\leq 2^{i}$. Hence, the first invariant holds after the deletion. 

\paragraph{The second invariant.}
Our goal is to show that after the deletion, $|C_1|=1$, $2^{i-2}\leq |C_i|\leq 2^{i-1}$ for any $2\leq i \leq \clg-1$, and $|C_{\clg}|\leq 2^{\clg-1}$. Consider any $1\leq i\leq  \clg$. 

We first argue the upper bound $|C_i|\leq 2^{i-1}$. Indeed, throughout the algorithm, $|C_i|$ can only be changed in either the $(i-1)$-th iteration or the $i$-th iteration. If the while loop has the $i$-th iteration, then $|C_i|\leq 2^{i-1}$ by Observation~\ref{obser:delete} and thus we are done with the proof. Otherwise, if $|C_i|$ is changed in the $(i-1)$-th iteration, then $|C_i|$ cannot become larger than before by Observation~\ref{obser:delete} (applying to $i-1$). By the second invariant before the deletion, $|C_i|\leq 2^{i-1}$. Hence, $|C_i|\leq 2^{i-1}$ still holds after the deletion. Note that for $i=1$, this is $|C_1|\leq 1$. 

Next, we prove the lower bound $|C_i|\geq 2^{i-2}$ for $2\leq i\leq \clg-1$. If $|C_i|\geq 2^{i-2}$ before the $i$-th iteration, then the while loop will be over and thus $C_i$ will not be changed. Hence, in this case $|C_i|\geq 2^{i-2}$ obviously hold. Otherwise, the while loop will enter the $i$-th iteration, after which $|C_i|$ is either $2^{i-1}$ or $2^{i-2}$ by Observation~\ref{obser:delete}. Since $|C_i|$ cannot be changed after the $i$-th iteration, $|C_i|\geq 2^{i-2}$ holds at the end of the algorithm. 

For $i=1$, since $|C_1|=0$ after Line~\ref{ln:setC1}, the while loop will enter the first iteration. Recall that we have assumed that $|S|\geq 16$ and thus $\clg\geq 4$. Hence, the first iteration will set $|C_1|$ to $1$ by Observation~\ref{obser:delete}. 
After that $|C_1|$ will not be changed throughout the algorithm. Therefore, $|C_1|=1$ holds after the deletion. 

This proves the second invariant. 

\paragraph{The third invariant.}
To show that the third invariant still holds after the deletion, it suffices to show that after each $i$-th iteration, the third invariant still holds for both $C_i$ and $C_{i+1}$, because the iteration does not change any other $C_j$ with $j\notin \{i,i+1\}$ (and neither $S_{\ge i}$ nor $S_{\geq i+1}$ will be changed after the iteration). Our goal is to argue that after the $i$-th iteration, $C_i$ contains $|C_i|$ points of $S_{\geq i}$ with the smallest keys and likewise $C_{i+1}$ contains $|C_{i+1}|$ points of $S_{\geq i+1}$ with the smallest keys. We assume inductively that these are true right before the iteration. 

Line~\ref{line:addCj} updates $C_i$ to the set of $c_i$ points of $S_i\cup S_{i+1}$ with the smallest keys. We claim that $C_i$ contains $|C_i|$ points of $S_{\geq i}$ with the smallest keys. Indeed, if $i=\clg-1$, then $S_{\geq i}=S_i\cup S_{i+1}$, and thus the claim immediately follows. Now assume that $i\leq \clg-2$. Then, by the third invariant before the deletion, $|C_{i+1}|\geq 2^{i-1}$. As $|C_i|=2^{i-1}$ in this case, each point of $C_i$ is either from $S_i$ or from $C_{i+1}$. Due to the third invariant of $C_{i+1}$ and the definition of $C_i$, the claim follows. 
The claim statement holds through the iteration because neither $C_i$ nor $S_{\geq i}$ will change later in this iteration. 

For $C_{i+1}$, it is changed in Line~\ref{line:removeCj1} and Line~\ref{line:addCj1}. Line~\ref{line:removeCj1} only moves some points out of $C_{i+1}$ and the next line also removes these points from $S_{i+1}$. Hence, after Line~\ref{line:removeAj1}, $C_{i+1}$ still contains $|C_{i+1}|$ points of $S_{\geq i+1}$ with smallest keys. The next three lines move a set $Q$ of points from $S_i$ to $S_{i+1}$. The points of $Q$ whose keys smaller than the largest key of $C_{i+1}$ are then added to $C_{i+1}$ at Line~\ref{line:addCj1}. Therefore, after Line~\ref{line:addCj1}, $C_{i+1}$ still contains $|C_{i+1}|$ points of $S_{\geq i+1}$ with smallest keys. 

This proves that the third invariant still holds after the deletion.

\paragraph{The fourth invariant.}
Consider any $i>\lceil \log (n+1)\rceil$ (recall that $n=|S|$ after the deletion). By the fourth invariant before the deletion, $S_i=\emptyset$. Since the algorithm never adds any points to $S_i$, $S_i=\emptyset$ still holds after the deletion. 
Note that it is possible that $\clg+1=\lceil \log (n+1)\rceil$ and $S_{\clg+1}\neq \emptyset$. If this happens, the merging process in the beginning of the algorithm resets $S_{\clg+1}= \emptyset$. Therefore, the fourth invariant is proved. 
\end{proof}

\subsection{Time analysis}
\label{sec:time}

We now analyze the time complexities of the two updates (insertion and delete-min). 

We refer to the reblance procedures in the $i$-th iterations of the insertion and the deletion algorithms as {\em iteration-$i$ rebalances}. 
We start with the following lemma.

\begin{lemma} \label{lem:amortize_lemma}
    For any $i = 1, \ldots, \clg - 1$, there are $\Omega(2^i)$ updates between any four consecutive iteration-$i$ rebalances.
\end{lemma}
\begin{proof}
    Consider four consecutive iteration-$i$ rebalances. Note that at least one of the following two cases is true:
    (1) Among the four rebalances, there are two consecutive insertion-induced rebalances.
    (2) Among the four rebalances, there are two (not necessarily consecutive) deletion-induced rebalances.
We analyze the two cases in the following. 

\paragraph{The first case.}
    For the first case, let $R_1$ and $R_2$ be the two consecutive insertion-induced iteration-$i$ rebalances, where $R_1$ is temporally before $R_2$. 
    We will show that there must be $\Omega(2^i)$ insertions between $R_1$ and $R_2$.
    
    First, notice that just before any insertion-induced iteration-$i$ rebalance, we have $|S_i| > 2^i$.
    Together with the second invariant, just before any insertion-induced iteration-$i$ rebalance (e.g., $R_2$), it holds that
    \[\sum_{j=0}^i|S_j|  \geq \sum_{j=0}^{i-1}|C_j| + |S_i|  \geq 0 + 1 + \sum_{j=2}^{i-1} 2^{j - 2} + 2^i = 2^{i - 2} + 2^i.\]
    
    According to the insertion algorithm, if there is an insertion-induced iteration-$i$ rebalance,
    then there has been an insertion-induced iteration-$j$ rebalance for all $j < i$ during the same insertion.
    Furthermore, right after an insertion-induced iteration-$j$ rebalance, 
    we have $|S_j| = |C_j| \leq 2^{j - 1}$. 
    For $j = 0$, we have $|S_0| = |C_0| = 0$.
    This implies that right after an insertion-induced iteration-$i$ rebalance (e.g., $R_1$), it holds that 
    \[\sum_{j=0}^i|S_j| \leq 0 + \sum_{j=1}^i 2^{j - 1} = 2^i - 1.\]
    
    Because $R_1$ and $R_2$ are two consecutive insertion-induced iteration-$i$ rebalances (in particular, there are no deletion-induced iteration-$i$ rebalances in between), the only cause for the increase of $\sum_{j=0}^i|S_j|$ between $R_1$ and $R_2$ is due to insertions.
    Therefore, there must be at least $(2^{i - 2} + 2^i) - (2^i - 1) = 2^{i - 2} + 1 = \Omega(2^i)$ insertions between $R_1$ and $R_2$.

\paragraph{The second case.}
    For the second case, let $R_1$ and $R_2$ be the two deletion-induced iteration-$i$ rebalances,
    where $R_1$ is temporally before $R_2$. We will show that there must be $\Omega(2^i)$ deletions between $R_1$ and $R_2$.
    We first consider the case where $i\leq \clg -2$.  

    According to our deletion algorithm, if there is a deletion-induced iteration-$i$ rebalance,
    then there must also be a deletion-induced iteration-$j$ rebalance for every $j < i$ during the same deletion. By Observation~\ref{obser:delete}, $|C_j|=2^{j-1}$ for all $1 \leq j< i$ right before the $i$-th iteration. 

    Hence, right before a deletion-induced iteration-$i$ rebalance (e.g., $R_2$), we have
    \[\sum_{j=0}^{i}|C_j| =\sum_{j=0}^{i-1}|C_j| + |C_i|= 0 + \sum_{j=1}^{i - 1} 2^{j - 1} + |C_i|= 2^{i - 1} - 1 + |C_i| <  2^{i - 1} - 1 + 2^{i-2}.\]
    Note that the last inequality $|C_i|<2^{i-2}$ follows from the while loop condition right before entering a deletion-induced iteration-$i$ rebalance. 
    
    On the other hand, right after any deletion-induced iteration-$i$ rebalance (e.g., $R_1$), we have $|C_j| =  2^{j - 1}$ for all $1 \leq j \leq i$, and in particular, $|C_i|=2^{i-1}$.  
    Therefore, right after $R_1$, it holds that  
    \[\sum_{j=0}^i|C_j| = 0 + \sum_{j=1}^i 2^{j - 1}  = 2^i - 1.\]

    By Observation~\ref{obser:insert}, insertions cannot decrease the size of any $C_j$. Also, in each $j$-th iteration of the deletion algorithm, $|C_j|$ always increases while $C_{j+1}$ may lose points. However, according to the analysis in the proof of Observation~\ref{obser:delete}, all the points lost from $C_{j+1}$ go to $C_j$. The above together implies that the decrease of $\sum_{j=0}^i|C_i|$ from $R_1$ to $R_2$, which is at least $(2^i - 1) - (2^{i - 1} - 1+2^{i-2}) =  2^{i - 2}=  \Omega(2^i)$, is all due to deletions between $R_1$ and $R_2$. Hence, there must be $\Omega(2^i)$ deletions between $R_1$ and $R_2$. 

    We now consider the case where $i=\clg -1$. Let $R_1'$ denote the deletion-induced iteration-$(i-1)$ rebalance caused by the same deletion as $R_1$. By definition, $R_1'$ happened right before $R_1$. 
    Similarly, let $R_2'$ denote the deletion-induced iteration-$(i-1)$ rebalance caused by the same deletion as $R_2$. If we follow the same analysis as above (which applies to $R_1'$ and $R_2'$ since $i-1=\clg-2$), there are $\Omega(2^i)$ deletions between $R_1'$ and $R_2'$. As $R_j'$ and $R_j$ are caused by the same deletion for $j\in \{1,2\}$, the same deletions between $R_1'$ and $R_2'$ are also between $R_1$ and $R_2$. Hence,    
    there must be $\Omega(2^i)$ deletions between $R_1$ and $R_2$.
\end{proof}

With Lemma~\ref{lem:amortize_lemma}, we are now in a position to analyze the time complexities of the updates. 
We show that the amortized times of both updates are $O(\tau\cdot \log n)$, assuming that constructing a static data structure $\calD(S')$ takes $O(|S'|\cdot \tau)$ time for any subset $S'\subseteq S$. Note that $\tau=O(\log^3 m + \log n \log^2 m)$ by Theorem~\ref{theo:nearneighbor}. 


\paragraph{Insertion time.}
For the insertion algorithm, we first bound the iteration-$i$ rebalance time for any $i$. 
First of all, it is not difficult to see that the runtime of the rebalance procedure excluding the rebuild in Line~\ref{line:insert_body_end} can be bounded by $O(|S_i|+|S_{i+1}|)$. 
By the first invariant, $|S_{i+1}|\leq 2^{i+1}$. For $S_i$, although $|S_i|>2^i$ in the beginning of the iteration, $|S_i|$ is no greater than $2^{i+1}$. To see this, each point of $S_i$ in the beginning of the iteration is either from the old $S_i$ before the insertion, or from some $S_j$ for $j<i$, or the new inserted point $b$. By the first algorithm invariant before the insertion, $\sum_{j=0}^{i}|S_j|\leq \sum_{j=0}^i2^j=2^{i+1}-1$. Hence, $|S_i|\leq 2^{i+1}$ holds in the beginning of the $i$-th iteration. Therefore, the time of the rebalance procedure excluding the rebuild is $O(2^i)$. The rebuild operation takes $O(2^i\cdot \tau)$ time as $|S_i|\leq 2^{i+1}$. 
Hence, each iteration-$i$ reblance takes $O(2^i\cdot \tau)$ time in total. 

By Lemma~\ref{lem:amortize_lemma}, after the first three iteration-$i$ rebalances (of both insertions and deletions), we can amortize all future iteration-$i$ rebalances across $\Omega(2^i)$ updates and therefore the amortized time of each iteration-$i$ reblance is $O(\tau)$. 
For the first three iteration-$i$ rebalances, we can handle them in the following observation. 
\begin{observation}\label{obser:first3}
The amortized time of the first three iteration-$i$ rebalances (of both insertions and deletions) is $O(\tau)$. 
\end{observation}
\begin{proof}
If there is a fourth iteration-$i$ rebalance, then we charge all first three reblances to the fourth one. As the fourth one has amortized time $O(\tau)$, which has been proved above for insertions and will also be proved later for the deletions, the amortized time of the first three is also $O(\tau)$. If there is no fourth iteration-$i$ rebalance, then we call $i$ a {\em special index}. 

Now let $I$ be the set of all special indices. Let $i$ be the largest index of $I$. Then, by the first invariant, $\sum_{j\in I}|S_j|=O(2^i)$. Hence, the total time of the iteration-$j$ reblances for all $j\in I$ is $O(2^i\cdot \tau)$. Note that there must be at least $\Omega(2^i)$ updates in total. We charge the time $O(2^i\cdot \tau)$ to all these updates and thus obtain that the amortized time of each iteration-$j$ reblance for any $j\in I$ is $O(\tau)$. 
\end{proof}

We conclude that the amortized time of each iteration-$i$ rebalance is $O(\tau)$.
Because there are $O(\log n)$ iterations in the insertion algorithm, the amortized time of the while loop in the insertion algorithm is $O(\tau\cdot \log n)$.

In addition, for the rebuild procedure in Line~\ref{line:insert_body_end_00}, if the while loop does not have any iteration, then $i=0$ and $|S_0|\leq 1$ at Line~\ref{line:insert_body_end_00}, and thus the rebuild takes $O(1)$ time. Otherwise, we charge the time of the rebuild to the rebuild operation in the last iteration of the while loop (the rebuild time in Line~\ref{line:insert_body_end_00} cannot exceed the rebuild time in the last iteration of the while loop by a constant factor because $i$ differs by one in these two rebuilds). 

In summary, the amortized time of the insertion algorithm is $O(\tau\cdot \log n)$, which is $O(\log n \log^3 m + \log^2 n \log^2 m)$.

\paragraph{Deletion time.}
The deletion time anaysis is similar. We first bound the iteration-$i$ rebalance time for any $i$. The reblance time excluding the rebuild in Line~\ref{line:delete_body_end} is $O(|S_i|+|S_{i+1}|)$, which is $O(2^i)$ since $|S_i|\leq 2^{i}$ and $|S_{i+1}|\leq 2^{i+1}$ due to the first invariant and Observation~\ref{obser:delete}. The rebuild operation in Line~\ref{line:delete_body_end} takes $O(2^i\cdot \tau)$ time as $|S_i|\leq 2^i$. 

By Lemma~\ref{lem:amortize_lemma}, after the first three iteration-$i$ rebalances,
we can amortize all future rebalances across $\Omega(2^i)$ updates and thus the amortized time of each reblance excluding the first three is $O(\tau)$. 
The first three rebalances have been handled by Observation~\ref{obser:first3}. 
We conclude that the amortized time of the iteration-$i$ rebalance is $O(\tau)$. Because there are $O(\log n)$ iterations, the amortized time of the while loop in the deletion algorithm is $O(\tau \cdot \log n)$. Using the same analysis as in the insertion algorithm, the time of the rebuild in Line~\ref{line:delete_body_end_00} can be charged to the last iteration of the while loop.

It remains to analyze the time of the merging process in the beginning of the deletion algorithm. 
The time of the process is dominated by the rebuild in Line~\ref{line:delete_endcheck_end},
which takes $O(2^\clg\cdot \tau)$ time. Notice that the set $S_{\clg + 1}$ in the merging process was created by a unique rebalance procedure in the $\clg$-th iteration of the insertion, and the runtime of that rebalance is asymptotically the same as the rebuild in Line~\ref{line:delete_endcheck_end}. 
Hence, we charge the rebuild in Line~\ref{line:delete_endcheck_end} to that rebalance in the insertion. Note that each such insertion-induced rebalance is charged at most once because after $S_{\clg + 1}$ is merged into $S_{\clg}$, it can only be created by a new  insertion-induced rebalance in the future. As the amortized time of each insertion-induced rebalance is $O(\tau)$, the amortized time of the merging processing is also $O(\tau)$.

Putting it all together the amortized deletion time is $O(\tau\cdot\log n)$, which is $O(\log n \log^3 m + \log^2 n \log^2 m)$.

\paragraph{Summary.}
Combining the above discussions, we obtain Theorem~\ref{theo:deletemin}.

\subsection{Extending it to other decomposable searching problems}
\label{sec:extension}
It is straightforward to extend the above algorithm to other decomposable searching problems. 
Roughtly speaking, a problem is said to be \emph{decomposable} if a query on a set $S$ can be answered in constant time 
by combining the results of two queries on two disjoint subsets that partition $S$; 
see \cite{ref:BentleyDe80,ref:BentleyDe79} for a formal definition. 
Typical examples of decomposable problems include nearest neighbor searching and 
membership queries (i.e., given an element, determine whether it belongs to $S$).



We can obtain the following result. 

\begin{theorem}\label{theo:main}
Suppose there exists a static data structure $\calD(S)$ for a decomposable problem 
on a set $S$ of $n$ elements with complexity $(P(n), U(n), Q(n))$, where $P(n)$ denotes 
the preprocessing time, $U(n)$ the space usage, and $Q(n)$ the query time. 
Assume that both $P(n)$ and $U(n)$ are at least linear in $n$, and that $Q(n)$ is monotone increasing. 
Further assume that each element of $S$ is assigned a key (or weight). 
Then $\calD(S)$ can be transformed into a dynamic data structure that supports 
priority queue operations with the following performance guarantees: 
insertion time $O(P(n)\log n / n)$, delete-min time $O(P(n)\log n / n)$, 
and query time $O(Q(n)\log n)$. 
The total space usage remains $O(U(n))$.   
\end{theorem}
\begin{proof}    
These bounds are obtained by applying the algorithmic framework described above 
for weighted nearest neighbor searching. 
The structure partitions $S$ into subsets 
$S_0, S_1, \ldots, S_{\clg}$ and maintains, for each $0 \le i \le \clg$, 
a subset $C_i \subseteq S_i$. 
Since the sizes of the subsets $S_i$ grow geometrically, 
the total space of the dynamic structure is $O(U(n))$. 
A query is answered by querying each subset $S_i$, 
which yields a total query time of $O(Q(n)\log n)$. 
The insertion and delete-min operations can be analyzed in the same manner as above simply by changing the parameter $\tau$ accordingly.
\end{proof}


Since many searching problems are decomposable, Theorem~\ref{theo:main} has the potential to find many applications. 
In particular, it is most useful in settings where updates consist solely of insertions and delete-min operations, and where fully dynamic data structures have previously been used only to support such restricted updates. Below we outline several representative potential applications.


\paragraph{Proximity problems.}
Static proximity problems such as nearest neighbor searching, bichromatic closest pair~\cite{ref:EppsteinDy95,ref:ChanDy20}, and related variants are decomposable. In scenarios where points are processed in increasing order of a key (e.g., weight, distance, or time), only insertions and delete-min operations are required. Our transformation could yield dynamic data structures with improved amortized bounds compared to those obtained via fully dynamic proximity structures.

\paragraph{Range searching.}
Range searching problems are classical decomposable problems with efficient static solutions based on range trees~\cite{ref:BentleyDe79,ref:BentleyDe80}, partition trees~\cite{ref:ChanOp12,ref:MatousekEf92,ref:MatousekRa93,ref:WangA25}, or related structures~\cite{ref:AgarwalSi17}. In applications where objects are activated over time and removed in increasing key order, our technique could provide dynamic support for priority-queue updates while preserving near-optimal space usage. This may be particularly relevant in sweepline and incremental geometric algorithms.

\paragraph{Parametric and incremental geometric optimization.}
Some geometric optimization problems process events in sorted order of a parameter (e.g., distance threshold or weight). In such settings, if elements are inserted as the parameter increases, and only the smallest pending event is removed at each step, then our technique could offer a systematic way to avoid fully dynamic data structures in these frameworks by exploiting the decomposability of the underlying static queries.

\medskip
Overall, as discussed in~\cite{ref:BentleyDe79,ref:BentleyDe80}, a wide range of geometric and combinatorial searching problems are decomposable. Consequently, our priority-queue transformation provides a general tool for converting static solutions into dynamic ones in settings where only insertions and delete-min updates are required. We expect further applications to emerge in many other decomposable problems.

\section{Deletion-only geodesic unit-disk range emptiness (GUDRE) queries}
\label{sec:line_seperated_disk_emptiness}

In this section, we prove Theorem~\ref{theo:rangeempty} for the deletion-only GUDRE query problem.
The problem can be reduced to a {\em diagonal-separated} problem by using the balanced polygon decompositions (BPD) of simple polygons as described in Section~\ref{sec:pre}.



\paragraph{Diagonal-separated deletion-only GUDRE queries.}
Consider a diagonal $d$ that partitions $P$ into two subpolygons $P_L$ and $P_R$. In the diagonal-separated setting, we assume that all points of $S$ lie in $P_L$, and every query point $q$ lies in $P_R$. We will prove the following result.

\begin{lemma}\label{lem:diaseprangeempty}
Let $P$ be a simple polygon with $m$ vertices, and assume that a GH data structure has been built for $P$. Given a diagonal $d$ that partitions $P$ into two subpolygons $P_L$ and $P_R$, and a set $S$ of $n$ points contained in $P_L$, we can construct a data structure for $S$ in $O(n\log^2 m + n\log n\log m)$ time that supports the following operations:
\begin{enumerate}
    \item \emph{Delete:} delete a point from $S$ in $O(\log n \log^2 m)$ amortized time.
    \item \emph{GUDRE query:} for any query point $q \in P_R$, determine whether there exists a point $p \in S$ with $d(p,q)\le 1$, and if so, return such a point as a witness, in $O(\log n \log m)$ time.
\end{enumerate}
\end{lemma}

\paragraph{Proving Theorem~\ref{theo:rangeempty}.}
Before proving Lemma~\ref{lem:diaseprangeempty}, we first use it to prove Theorem~\ref{theo:rangeempty}.

In the preprocessing, we compute the BPD-tree $T_P$ for $P$ in $O(m)$ time~\cite{ref:GuibasOp89}. In addition, we build a point location data structure on the triangulation of $P$ formed by the triangles of the leaves of $T_P$ in $O(m)$ time so that given a point $q\in P$, the triangle that contains $q$ can be computed in $O(\log m)$ time~\cite{ref:KirkpatrickOp83,ref:EdelsbrunnerOp86}. Finally, we construct the GH data structure for $P$ in $O(m)$ time. This finishes our preprocessing, which takes $O(m)$ time in total.

Given a set $S$ of $n$ points in $P$, for each node $v\in T_P$, let $S_v$ be the set of points of $S$ in $P_v$. We compute $S_v$ for all nodes $v\in T_P$ in $O(n\log m)$ time, as follows.
First, using the point location data structure, for each point $p\in S$, we find the leaf $v_p$ of $T_P$ whose triangle contains $p$ in $O(\log m)$ time. Then, for each node $v$ in the path from $v_p$ to the root of $T_P$, we add $p$ to $S_v$. As the height of $T_P$ is $O(\log m)$, the total time is $O(n\log m)$ for computing $S_v$ for all $v\in T_P$.

In the above algorithm, we also keep a list $L$ of nodes $v$ with $S_v\neq \emptyset$. As each point of $S$ can appear in $S_v$  for $O(\log m)$ nodes $v\in T_P$, we have $L=O(n\log m)$ and $\sum_{v\in L}|S_v|=O(n\log m)$. The reason we use $L$ is to avoid spending $\Omega(m)$ time on constructing the data structure for $S$.

For each node $v\in L$, if $v$ is a leaf, then $P_v$ is a triangle. In this case, we construct a Euclidean case unit-disk range emptiness query data structure $\calD_v$ for $S_v$ in $O(|S_v|\log |S_v|)$ time~\cite{ref:EfratGe01,ref:WangCo26}. Doing this for all leaves in $L$ takes $O(n\log n)$ time in total.

If $v$ is not a leaf, then let $u$ and $w$ be the left and right children of $v$, respectively. We construct a data structure of Lemma~\ref{lem:diaseprangeempty} for $S_u$ in $P_u$ with respect to the diagonal $d_v$ for query points in $P_w$, so the points of $S_u$ and the query points are separated by $d_v$; let $\calD_u(w)$ denote the data structure. This takes $O(|S_u|\log^2 m+|S_u|\log n\log m)$ time by Lemma~\ref{lem:diaseprangeempty}.
Symmetrically, we construct a data structure $\calD_w(u)$ of Lemma~\ref{lem:diaseprangeempty} for $S_w$ in $P_w$ with respect to the diagonal $d_v$ for query points in $P_u$, which takes $O(|S_w|\log^2 m+|S_w|\log n\log m)$ time.
Doing the above for all internal nodes $v$ of $L$ takes $O(n\log^3m+n\log n\log^2 m)$ time as $\sum_{v\in L}|S_v|=O(n\log m)$.

This finishes the construction of our data structure for $S$, which takes $O(n\log^3m+n\log n\log^2 m)$ time in total.

Given a query point $q\in P$, our goal is to determine whether $S$ has a point $p$ with $d(p,q)\leq 1$, and if so, return such a point.

We first find the leaf $v_q$ of $T_P$ whose triangle contains $q$. This takes $O(\log m)$ time using the point location data structure for the leaf triangles of $T_P$. Then, using the Euclidean data structure $\calD_{v_q}$, we can determine in $O(\log n)$ time whether $S_{v_q}$ has a point $p$ with $d(p,q)=|\overline{pq}|\leq 1$~\cite{ref:EfratGe01,ref:WangCo26}. If so, we return such a point $p$ and halt our query algorithm. Otherwise, for each ancestor $v$ of $v_p$ in $T_P$, we do the following. Let $u$ and $w$ be the left and right children of $v$, respectively. Note that either $u$ or $w$ is an ancestor of $v_q$. We assume that $w$ is an ancestor of $v_q$ and the other case can be handled similarly. Using the data structure $\calD_u(w)$, we determine whether $S_u$ has a point $p$ with $d(q,p)\leq 1$. If so, we return such a point $p$ and halt the query algorithm. This can be done in $O(\log n\log m)$ time by Lemma~\ref{lem:diaseprangeempty}. As the height of $T_P$ is $O(\log m)$, the query algorithm takes $O(\log n\log^2 m)$ time in total.

To delete a point $p\in S$, by a point location, we first find the leaf $v_p$ whose triangle contains $p$. We delete $p$ from the Euclidean data structure $\calD_{v_p}$, which takes $O(\log n)$ amortized time~\cite{ref:EfratGe01,ref:WangCo26}. Then, for each ancestor $v$ of $v_p$ in $T_P$, let $u$ and $w$ be the left and right children of $v$, respectively.
We assume that $w$ is an ancestor of $v_q$ and the other case can be handled similarly. We delete $p$ from the data structure $\calD_w(u)$, which takes $O(\log n\log^2 m)$ amortized time by Lemma~\ref{lem:diaseprangeempty}. As the height of $T_P$ is $O(\log m)$, the deletion takes $O(\log n\log^3 m)$ amortized time.

This proves Theorem~\ref{theo:rangeempty}.

\paragraph{Proving Lemma~\ref{lem:diaseprangeempty}.}
In the following, we prove Lemma~\ref{lem:diaseprangeempty}. We follow the notation in the lemma, e.g., $P$, $S$, $d$, $P_L$, $P_R$, $n$, $m$. Without loss of generality, we assume that $d$ is vertical and $P_R$ is locally to the right of $d$. Let $z^*_0$ and $z^*_1$ be the lower and upper endpoints of $d$, respectively.


\subsection{Overview}

For each point $p\in P$, we define $D_p$ as the {\em geodesic unit-disk} centered at $p$, which consists of all points $q\in P$ with $d(p,q)\leq 1$. For any $S'\subseteq S$, define $d(S',p)=\min_{s\in S'}d(s,p)$.
Let $\beta_{S'}(p)$ denote the nearest neighbor of $p$ in $S'$ (if the nearest neighbor is not unique, then let $\beta_{S'}(p)$ refer to an arbitrary one).

We define $\partial P_R$ to be $\partial P\cap P_R$. Hence, although $d$ is a segment on the boundary of $P_R$, $\partial P_R$ does not contain the interior points of $d$.

In what follows, we give an overview of our data structure for Lemma~\ref{lem:diaseprangeempty}.

For any query point $q\in P_R$, we need to know whether $D_q\cap S=\emptyset$. Observe that $D_q \cap S \neq \emptyset$ if and only if $q \in \cup_{p \in S} D_p$. Based on this observation, our strategy is to maintain an implicit representation of $\calD_R(S)$, where $\calD_R(S)$ is the portion of $\cup_{s \in S} D_s$ in $P_R$. To this end, we propose a concept {\em right envelope} defined below.

\paragraph{Right envelope.}
Roughly speaking, the {\em right envelope} $\Xi(S)$ of $\calD_R(S)$ is the boundary of $\calD_R(S)$ excluding $d$. Specifically, $\Xi(S)$ is union of $\Xi\ca(S)$ and $\Xi\sa(S)$, with
$\Xi\ca(S) = \{p \in P_R : d(S, p) = 1\}$ and $\Xi\sa(S) = \{p \in \partial P_R : d(S, p) < 1\}$.
In general, $\Xi\ca(S)$ is a collection of circular arcs each belonging to $\partial D_s$ for some point $s\in S$, while $\Xi\sa(S)$ is a collection of segments on $\partial P_R$; see Figure~\ref{fig:XiS}.

Note that $\calD_R(S)$ may have multiple connected components. The following lemma characterizes the shape of each component. The lemma proof requires additional observations and concepts. As we intend to give a quick overview here, we defer the proof to Section~\ref{sec:order}.
\begin{lemma}\label{lem:boundary}
For each connected component $\calD$ of $\calD_R(S)$, $\calD\cap d$ consists of a single segment of $d$ and $\calD$ is simply connected (i.e., there are no holes in $\calD$), meaning that $\partial\calD$ is composed of a segment of $d$ and a simple curve connecting the two endpoints of $\calD\cap d$.
\end{lemma}

By Lemma~\ref{lem:boundary}, each connected component of $\Xi(S)$ is a simple curve bounding a connected component of $\calD_R(S)$, and we call it a {\em component curve of $\Xi(S)$}. The curve connects two points of $d$, and we call them the {\em upper} and {\em lower} endpoints of the curve, respectively.

For any point $s \in S$, define $\xi_s(S)$ as the set of points $p\in \Xi(S)$ such that $d(s,p)\leq d(s',p)$ for any other point $s'\in S$. In other words, if we consider the geodesic Voronoi diagram $\vd(S)$ of $S$, then $\xi_s(S)$ consists of all points $p\in \Xi(S)$ in the Voronoi region of $s$; see Figure~\ref{fig:xisS}.
Note that it is possible that $\xi_s(S)=\emptyset$. We define $I(S) = \{s \in S : \xi_s(S) \neq \emptyset\}$.
One property is we will show later is that if $\xi_s(S)\neq \emptyset$, then $\xi_s(S)$ is a single connected portion of $\Xi(S)$.

For any $S'\subseteq S$, we define the above notations on $S'$ analogously, e.g., $\Xi(S')$, $\xi_s(S')$ for $s\in S'$. In the rest of this section, we will introduce other notations and related properties on $S$, with the understanding that the same notations and properties are also applicable to any $S'\subseteq S$ unless otherwise stated.

\begin{figure}
    \centering
    \begin{subfigure}[b]{.45\textwidth}
        \centering
        \includegraphics[width=\linewidth]{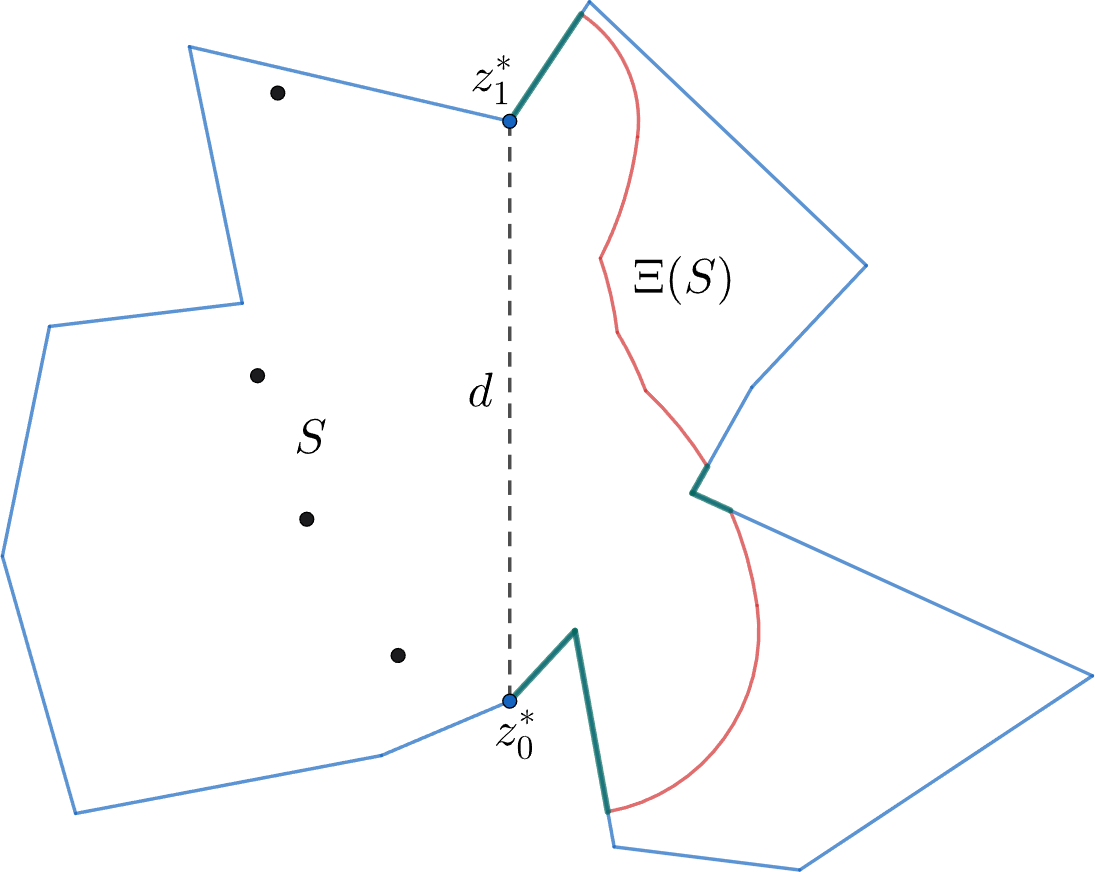}
        \subcaption{$\Xi\ca(S)$ is red and $\Xi\sa(S)$ is green.}
        \label{fig:XiS}
    \end{subfigure}
    \hfill
    \begin{subfigure}[b]{.45\textwidth}
        \centering
        \includegraphics[width=\linewidth]{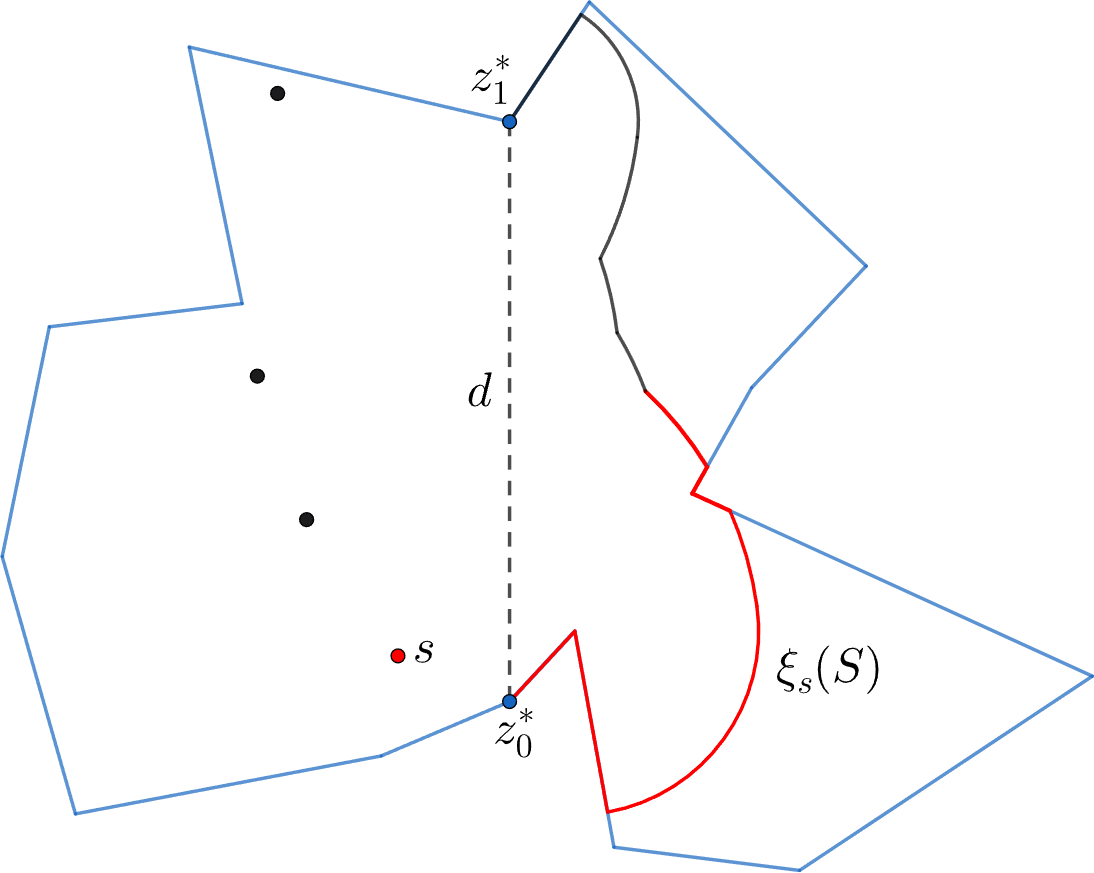}
        \subcaption{$\xi_s(S)$ is colored with red.}
        \label{fig:xisS}
    \end{subfigure}

    \caption{Illustrations of $\Xi(S)$. The region bounded by $\Xi(S)$ and $d$ is $\calD_R(S)$.}
\end{figure}

\paragraph{General position assumption.}
For the ease of exposition, we make a general position assumption to rule out some degenerate cases for $\Xi(S)$.
Specifically, we assume that: (1) no point in $P_R$ has the geodesic distance equal to $1$ from three points in $S$
and (2) there are no points $s\in S$ such that $D_s\cap d$ consists of exactly one point. This guarantees that for any $S'\subseteq S$ and any $s\in S'$, $\xi_s(S')$ cannot be a single point. In addition, we assume that no vertex of $P$ is equidistant to two points in $S$. This assures that the bisector of any two points of $S$ does not contain a vertex of $P$.
Finally, we assume that the query point $q$ is not on $d$ (the case $q\in d$ can be solved in a similar but simpler way).
These assumptions can be lifted without affecting the results, e.g., by slight perturbation to the locations of the sites or the vertices of $P$.

\paragraph{Answering queries using the right envelope.}

Intuitively, for any point $q\in P_R$, $q\in \calD_R(S)$ if and only if $q$ is ``to the left of'' $\Xi(S)$. However, the question is how to interpret ``left.''

In the Euclidean case where $P$ is the entire plane and $d$ becomes a vertical line, $\Xi(S)$ is a $y$-monotone curve. Due to this property, there is a trivial way to determine whether $q$ is to the left of $\Xi(S)$: First find by binary search the arc of $\Xi(S)$ intersecting the horizontal line through $q$, and then determine whether $q$ is to the left of the arc~\cite{ref:EfratGe01,ref:WangCo26}.

In our problem, however, $\Xi(S)$ is no longer $y$-monotone. Intuitively, this is because $P$ is not ``large enough'' and $\partial P$ changes the ``shape'' of $\Xi(S)$. Thus, we need a new approach. At a very high level, our strategy is to develop some concepts analogous to those in the Euclidean case. Specifically, we will introduce an order $\prec_R$ for any two points of $P_R$; this order corresponds to the $y$-coordinate order in the Euclidean case. We will also define a {\em traversal order} $\precxi$ for any two points of $\Xi(S)$; in the Euclidean case, this would be the order of traversing on $\Xi(S)$ from bottom to top as $\Xi(S)$ is $y$-monotone. In addition, we will define an order $\prec_L$ for all points of $S$ based on their geodesic distances to $z^*_0$, the lower endpoint of $d$. We will show that these three orders are consistent with each other, which will be instrumental in our data structure design.

\paragraph{Outline.}
The remainder of this section is organized as follows. Section~\ref{sec:order} introduces the above three orders and prove several properties. Section~\ref{sec:dsdescription} describes our data structure and Section~\ref{sec:rangeemptyquery} gives our query algorithm. We discuss in Section~\ref{sec:dsconstruct} how to construct our data structure. Finally, the algorithm for handling deletions is presented in Section~\ref{sec:deletion}.

\subsection{The three orders $\boldsymbol{\prec_R}$, $\boldsymbol{\precxi}$, and $\boldsymbol{\prec_L}$}
\label{sec:order}

In this section, we introduce the three orders and discuss their properties. In particular, we will show that they are consistent.

\paragraph{The order $\boldsymbol{\prec_R}$ for $\boldsymbol{P_R}$.}
For each point $p\in P_R$, we define $z_p = \argmin_{z \in d} d(z, p)$.
We will also use the notation $\pi(d, p)$ to mean $\pi(z_p, p)$.

For two points $p, q \in P_R$, if $p\in \pi(s,q)$ or $q\in \pi(s,p)$, then we say that the two points are not {\em $\prec_R$-comparable}; otherwise, we define $p \prec_R q$ if one of the following two cases holds:
\begin{itemize}
    \item $z_p$ lies below $z_q$; see Figure~\ref{fig:LTR1}.
    \item $z_p = z_q$ and the three paths $\pi(c, z_p)$, $\pi(c, p)$, $\pi(c, q)$ are ordered counterclockwise around $c$, where $c$ is the junction vertex of $\pi(z_p, p)$ and $\pi(z_p, q)$; see Figure~\ref{fig:LTR2}.
\end{itemize}


\begin{figure}
    \centering
    \begin{subfigure}[b]{.45\textwidth}
        \centering
        \includegraphics[width=\linewidth]{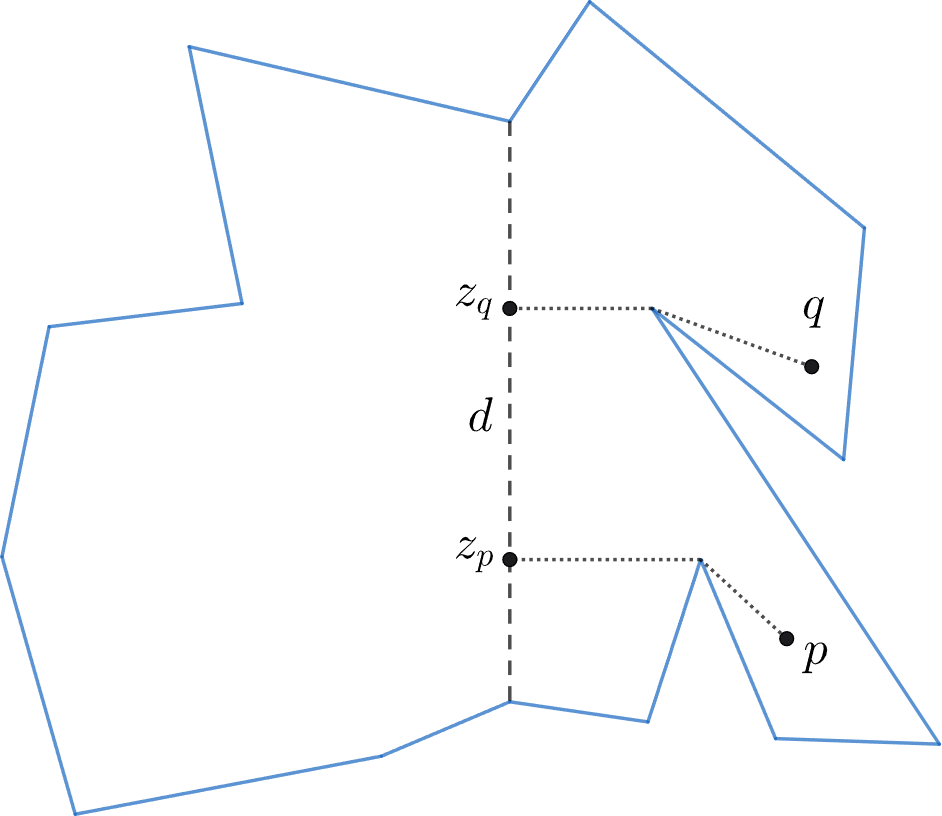}
        \subcaption{$p \prec_R p$ due to $z_p$ lying below $z_q$.}
        \label{fig:LTR1}
    \end{subfigure}
    \hfill
    \begin{subfigure}[b]{.45\textwidth}
        \centering
        \includegraphics[width=\linewidth]{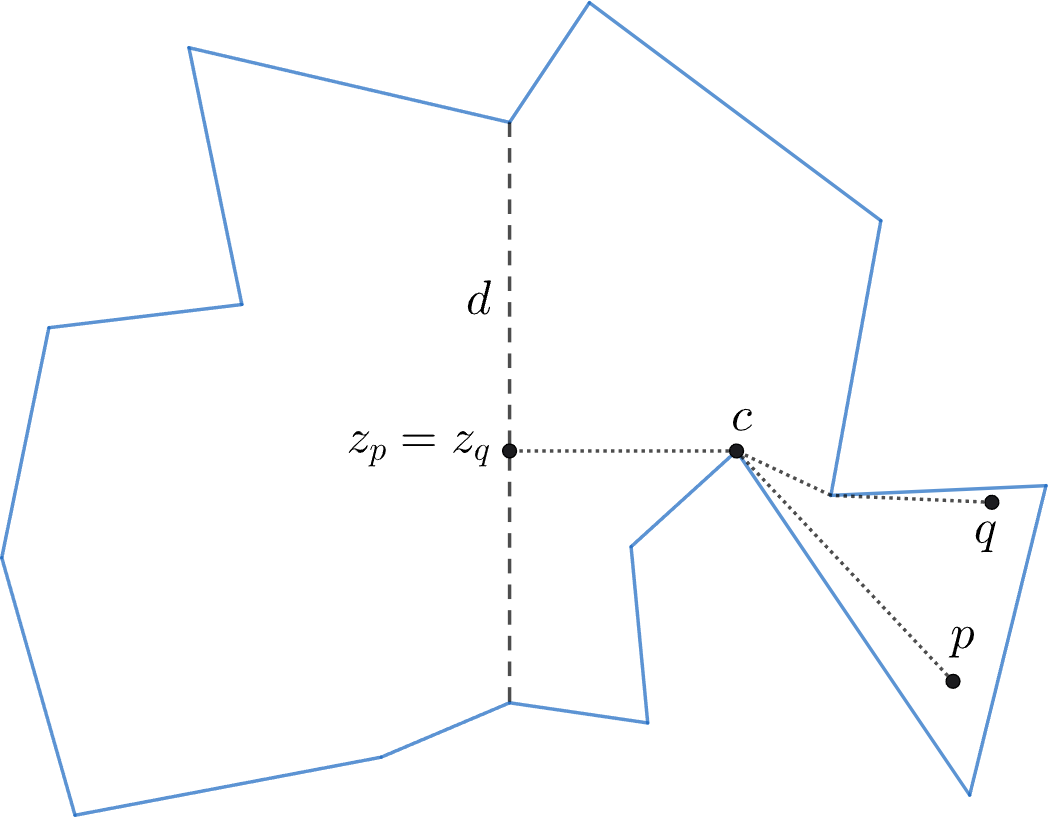}
        \subcaption{$p \prec_R q$ due to the clockwise order around $c$.}
        \label{fig:LTR2}
    \end{subfigure}
    \caption{Illustrations of $p \prec_R p$.}
\end{figure}

\begin{lemma} \label{lem:Risorder}
    $\prec_R$ is a partial order on $P_R$.
\end{lemma}
\begin{proof}
Consider three points $p_1,p_2,p_3\in P_R$ such that $p_1\prec_R p_2$ and $p_2\prec_R p_3$. In the following, we prove that $p_1\prec_R p_3$, which leads to the lemma.

If $z_{p_1}$ lies below $z_{p_3}$, then it immediately follows that $p_1\prec_R p_3$. Otherwise, since $p_1\prec_R p_2$ and $p_2\prec_R p_3$, it must be that $z_{p_1}=z_{p_2}=z_{p_3}$. Let $z=z_{p_1}$.

Let $T_z$ be a tree rooted at $z$ and formed by the three paths $\pi(s,p_1)$, $\pi(s,p_2)$, and $\pi(s,p_3)$. Consider the pre-order traversal list $L$ of the nodes of $T_z$ such that at each node $v\in T_z$, we visit their children counterclockwise starting from the one right after the parent of $v$. In particular, if $v$ is the root $z$, then we start visiting the first node counterclockwise from $\overline{zz^*_0}$. Observe that the order of $p_1,p_2,p_3$ following $\prec_R$ by our definition is consistent with their order in $L$. Since $p_1\prec_R p_2$ and $p_2\prec_R p_3$, $p_1$ is before $p_2$ and $p_2$ is before $p_3$ in $L$. Hence, $p_1$ is before $p_3$ in $L$. We thus conclude that $p_1\prec_R p_3$.
\end{proof}

\paragraph{The traversal order $\boldsymbol{\precxi}$ for $\boldsymbol{\Xi(S)}$.}
As discussed above, $\Xi(S)$ in general consists of multiple component curves, and each curve has an upper and lower endpoints on $d$.

For any $p, q \in \Xi(S)$, if $p$ and $q$ are in the same component curve $\gamma$ of $\Xi(S)$,
we define $p \precxi q$ if,
in the traversal of $\gamma$ from its lower endpoint to its upper endpoint, $p$ is before $q$.
If $p$ and $q$ are in different component curves, then the upper endpoint of one curve $\gamma_1$ must lie below the lower endpoint of the other $\gamma_2$ (i.e., $\gamma_1$ is ``lower'' than $\gamma_2$).
We define $p \precxi q$ if $p\in \gamma_1$ and $q\in \gamma_2$.

It is not difficult to see that $\precxi$ is a total order on the points of $\Xi(S)$.


\paragraph{The order $\boldsymbol{\prec_L}$ for $\boldsymbol{S}$.}
Due to our general position assumption, for any two points $s,t\in S$, $d(s,z^*_0)\neq d(t,z^*_0)$. If $d(s,z^*_0)<d(t,z^*_0)$,  we define $s\prec_L t$; otherwise, $t\prec_L s$. Thus, $\prec_L$ is a total order on $S$.
We write $s\preceq_L t$ if either $s=t$ or $s\prec_L t$.

\paragraph{Properties of the three orders.}
We prove some properties of the three orders, which eventually lead to the conclusion that they are consistent.

We start with the following lemma, which will be frequently used.

\begin{restatable}{lemma}{extdist} \label{lem:ext_dist}
    For any two points $p, q \in P_R$ with $p \in \pi(d, q)\setminus\{q\}$, $d(s, p) < d(s, q)$ holds for any point $s\in S$, and more generally, $d(S',p)<d(S',q)$ holds for any subset $S'\subseteq S$.
\end{restatable}
\begin{proof}
    According to Pollack, Sharir, and Rote~\cite[Lemma 1]{ref:PollackCo89},
    for any three points $a, b, c \in P$ and any point $a' \in \pi(b, c) \setminus \{b, c\}$,
    $d(a, a') < \max\{d(a, b), d(a, c)\}$.
    Applying this with $a = s$, $b = z_q$, $c = q$, and $a' = p$,
    we obtain $d(s, p) < \max \{d(s, z_q), d(s, q)\}$.
    In the following, we argue that $d(s,z_q)\leq d(s,q)$, which will lead to $d(s,p)<d(s,q)$.

    Let $z=d\cap \pi(s,q)$; see Figure~\ref{fig:disincrease}.
    We claim that $d(z,z_q)\leq d(z,q)$. Note that the claim immediately leads to $d(s,z_q)\leq d(s,q)$. Indeed,
    \[d(s,z_q)\os 1\leq d(s,z)+d(z,z_q)\os 2\leq d(s,z)+d(z,q)\os 3=d(s,q),\]
    where (1) is due to triangle inequality, (2) is due to the claim, and (3) holds because $z\in \pi(s,q)$.

\begin{figure}[t]
\begin{minipage}[t]{\linewidth}
\begin{center}
\includegraphics[totalheight=1.5in]{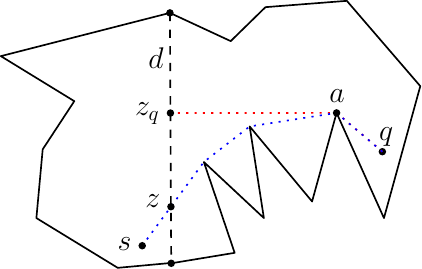}
\caption{Illustrating the case where $z_p$ is in the interior of $d$.}
\label{fig:disincrease}
\end{center}
\end{minipage}
\end{figure}

\begin{figure}[t]
\begin{minipage}[t]{\linewidth}
\begin{center}
\includegraphics[totalheight=1.5in]{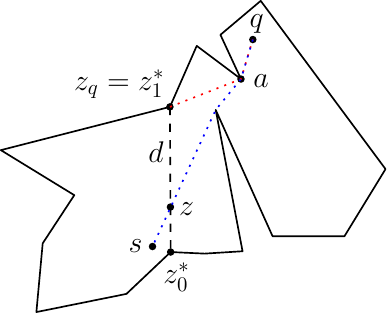}
\caption{Illustrating the case where $z_p$ is the upper endpoint of $d$.}
\label{fig:disincrease10}
\end{center}
\end{minipage}
\end{figure}

    It remains to prove the claim. First of all, since both $z$ and $z_q$ are on $d$, $\pi(z,z_q)$ is the line segment $\overline{zz_q}$. Let $a$ be the vertex of $\pi(z_q,q)$ adjacent to $z_q$. Depending on whether $z_q$ is in the interior of $d$, there are two cases.

    \begin{itemize}
    \item

    If $z_q$ is in the interior of $d$, then $\overline{az_q}$ must be perpendicular to $d$; see Figure~\ref{fig:disincrease}. According to the analysis in the proof of \cite[Lemma 1]{ref:PollackCo89}, if we move a point $b$ from $z_q$ to $q$ along $\pi(z_q,q)$, $d(z,b)$ will be strictly increasing. Hence, we obtain the claim $d(z,z_q)<d(z,q)$.

\item
    If $z_q$ is an endpoint of $d$, without loss of generality, we assume that $z_q$ is the upper endpoint $z^*_1$ of $d$; see Figure~\ref{fig:disincrease10}.
    Then, the angle $\angle z^*_0z^*_1a$ counterclockwise from $\overline{z^*_1z^*_0}$ to $\overline{z^*_1a}$ around $z^*_1$ must be at least $90^{\circ}$. Again, following the analysis in the proof of \cite[Lemma 1]{ref:PollackCo89}, if we move a point $b$ from $z_q$ to $q$ along $\pi(z_q,q)$, $d(z,b)$ will be strictly increasing. Hence, we obtain the claim $d(z,z_q)<d(z,q)$.
    \end{itemize}

    The above proves that $d(s,p)<d(s,q)$. Given any $S'\subseteq S$, we now argue $d(S',p)<d(S',q)$. Let $s=\beta_{S'}(q)$, i.e., the nearest neighbor of $q$ in $S'$. According to the above proof, we have $d(s,p)<d(s,q)$. Consequently, $d(S',p)\leq d(s,p)<d(s,q)=d(S',q)$.
\end{proof}

Lemma~\ref{lem:consistentRT} shows that the orders $\prec_R$ and $\precxi$ are consistent.
\begin{lemma} \label{lem:consistentRT}
    For any two points $p, q \in \Xi(S)$, if $p \prec_R q$, then ${p \precxi q}$.
\end{lemma}
\begin{proof}
Assuming that $p \prec_R q$, in the following we prove $p \precxi q$.

Let $s = \beta_p(S)$ and $t = \beta_q(S)$. Since $p,q\in \Xi(S)$, $d(s,p)\leq 1$ and $d(t,q)\leq 1$.
By Lemma~\ref{lem:ext_dist}, $d(s,p')<d(s,p)$ for all points $p'\in \pi(z_p,p)\setminus\{p\}$. This implies that $\pi(z_p, p) \subseteq D_s$. Thus, $\pi(z_p,p)$ is in the same component $\calD_p$ of $\calD_R$ containing $p$. Similarly, $\pi(z_q, q) \subseteq D_t$ and $\pi(z_q,q)$ is in the same component $\calD_q$ of $\calD_R$ containing $q$. Depending on whether $\calD_p=\calD_q$, there are two cases.

\begin{itemize}
    \item
If $\calD_p\neq \calD_q$, then the intersections $d\cap \calD_p$ and $d\cap \calD_q$ are disjoint. Since $\pi(z_p,p)\subseteq \calD_p$ and $\pi(z_q,q)\subseteq \calD_q$, we get that $z_p\neq z_q$. As  $p \prec_R q$, by definition, $z_p$ must lie below $z_q$. Hence, the intersection $d\cap \calD_p$ lies below $d\cap \calD_q$. According to our definition of $\precxi$, we obtain that $p\precxi q$.

    \begin{figure}
        \centering
        \begin{subfigure}[b]{.45\textwidth}
            \centering
            \includegraphics[totalheight=1.7in]{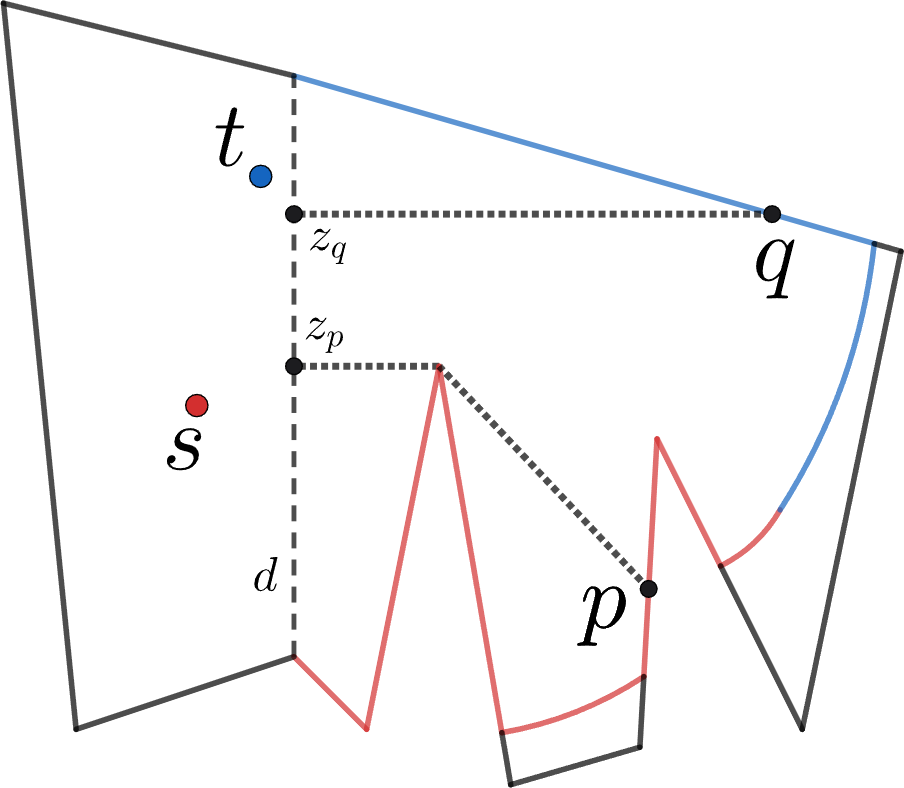}
            \subcaption{Case 1: $p \prec_R p$ due to $p^*$ lying below $q^*$.}
            \label{fig:plrqc1}
        \end{subfigure}
        \hfill
        \begin{subfigure}[b]{.45\textwidth}
            \centering
            \includegraphics[totalheight=1.7in]{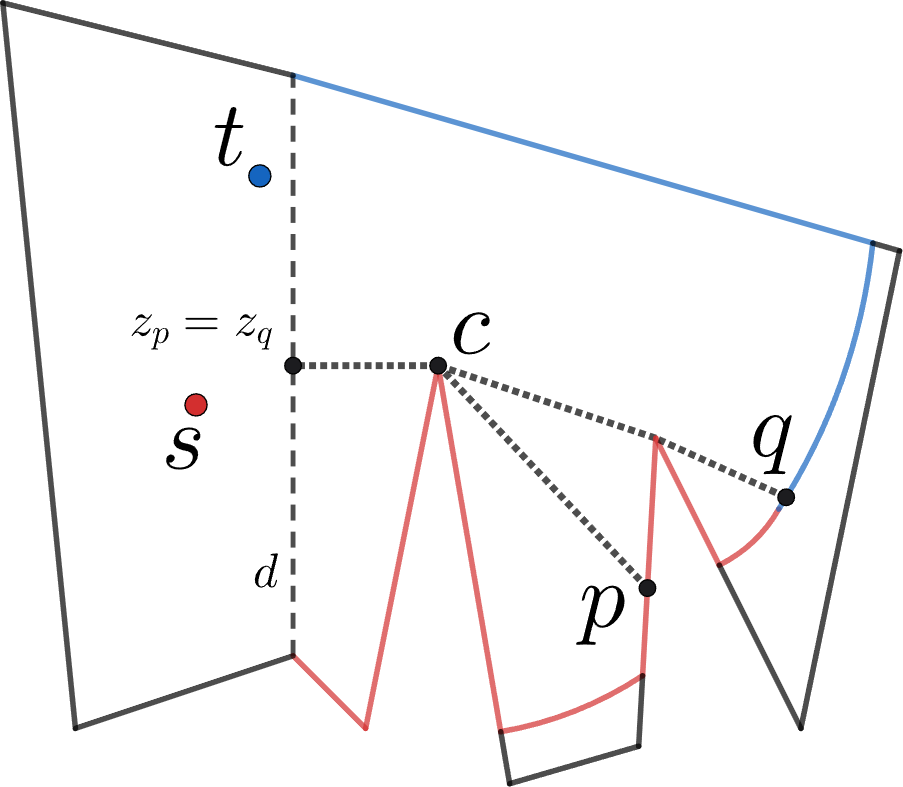}
            \subcaption{Case 2: $p \prec_R q$ due to the clockwise order around the point $c$.}
            \label{fig:plrqc2}
        \end{subfigure}

        \caption{Illustrations for the proof of Lemma~\ref{lem:consistentRT}.}
        \label{fig:consistent}
    \end{figure}

\item
If $\calD_p=\calD_q$, then let $Q=\calD_p$.
See Figure~\ref{fig:consistent} for illustrations of the proof for this case.
The path $\pi(z_p, p)$ divides $Q$ into two portions.
Let $Q_L$ and $Q_R$ be the portions of $Q$ left and right of $\pi(z_p, p)$, respectively.
Since $\pi(z_p,p)$ is the shortest path from $p$ to $d$ and $\pi(z_q,q)$ is the shortest path from $q$ to $d$, $\pi(z_p,p)$ and $\pi(z_q,q)$ cannot cross each other (since otherwise at least one of them could be shortened). This implies that $\pi(z_q,q)$ is either in $Q_L$ or in $Q_R$. Since $p\prec_R q$, by our definition of $\prec_R$, it must be that $\pi(z_q,q)\subseteq Q_L$. Hence, $q\in Q_L$. According to our definition of $\precxi$, we obtain that $p \precxi q$.
\end{itemize}

The lemma thus follows.
\end{proof}

Next, we will show that the order $\precxi$ is consistent with the order $\prec_L$ in Lemma~\ref{lem:consistentTL}. Before that, we present several properties that are mostly from the previous work. These properties will be helpful in proving Lemma~\ref{lem:consistentTL} and will also be used later in the paper.

Recall that $B(s,t)$ is the bisector of two points $s,t\in S$. The properties in the following lemma are from \cite{ref:AronovOn89,ref:AgarwalIm18}.

\begin{lemma} {\em \cite{ref:AgarwalIm18,ref:AronovOn89}} \label{lem:star_shaped}
\begin{enumerate}
    \item Each Voronoi region of $\vd(S)$ is star-shaped, i.e., for any point $p$ in the Voronoi region of a site $s\in S$, $\pi(s,p)$ is in the interior of the Voronoi region except that $p$ may be on the boundary of the region.
    \item For any $s, t \in S$ and $p\in P$, $\pi(s, p) \cap B(s,t)$ consists of at most one point.
    \item For any $s \in S$ and $p \in P$ such that $s = \beta_p(S)$, for all $q \in \pi(s, p) \setminus \{p\}$, $d(s, q) < d(S \setminus \{s\}, q)$.
    \item For any $s, t \in S$, $B(s,t)$ intersects $d$ at most once.
    \item The intersection of $d$ with each Voronoi region of $\vd(S)$ is connected (i.e., it is either $\emptyset$ or a single segment of $d$).
\end{enumerate}
\end{lemma}

We further have the following observation.

\begin{lemma} \label{lem:Lvor}
Suppose the Voronoi regions $\R(s)$ and $\R(t)$ of two points $s,t\in S$ both intersect $d$. Then, $d\cap \R(s)$ is below $d\cap \R(t)$ if and only if $s\prec_L t$.
In other words, the order of the Voronoi regions of $\vd(S)$ intersecting $d$ from bottom to top is consistent with the order $\prec_L$ of $S$.
\end{lemma}
\begin{proof}
Let $d_s=d\cap \R(s)$ and $d_t=d\cap \R(t)$. Suppose that $d_s$ is below $d_t$. Assume to the contrary that $t\prec_L s$. Then, $d(t,z^*_0)<d(s,z^*_0)$. Hence, the bisector $B(s,t)$ must intersect $d$ somewhere between $z^*_0$ and $d_s$. Since $d_s$ is below $d_t$, $B(s,t)$ must also intersect $d$ between $d_t$ and $d_s$. Therefore, $B(s,t)$ must intersect $d$ at least twice. But this contradicts with Lemma~\ref{lem:star_shaped}.

If  $d_t$ is below $d_s$, then following the same analysis as above, we can obtain $t\prec_L s$. The lemma thus follows.
\end{proof}

\begin{lemma} \label{lem:consistentTL}
For any two points $p, q \in \Xi(S)$, let $s = \beta_p(S)$ and $t = \beta_q(S)$. If $p \precxi q$, then $s \prec_L t$.
\end{lemma}
\begin{proof}
Let $p' = \pi(s, p) \cap d$ and $q' = \pi(t, q) \cap d$.

Since $s = \beta_p(S)$ and $\R(s)$ is star-shaped, $\pi(s,p)\subseteq \R(s)$. Hence, $p'\in \R(s)$. Similarly, $q'\in \R(t)$. To prove that $s\prec_L t$, if suffices to show that $p'$ is below $q'$.

Let $\calD_p$ be the connected component of $\calD_R(S)$ containing $p$. Define $\calD_q$ for $q$ similarly.
As $p\in \Xi(S)$, $d(s,p)\leq 1$ and thus $\pi(s,p)\subseteq D_s$. This implies that $\pi(p',p)\subseteq \calD_p$. Similarly, $\pi(q',q)\subseteq \calD_q$. Depending on whether $\calD_p=\calD_q$, there are two cases.

\begin{itemize}
    \item

If $\calD_p\neq \calD_q$, then since $p\precxi q$, by the definition of $\precxi$, $\calD_p\cap d$ is below $\calD_q\cap d$. As $\pi(p',p)\subseteq \calD_p$ and $\pi(q',q)\subseteq \calD_q$, we have $p'\in \calD_p\cap d$ and $q'\in \calD_q\cap d$. Therefore, $p'$ must be below $q'$.

\item
If $\calD_p= \calD_q$, then both $\pi(p',p)$ and $\pi(q',q)$ are in $\calD_p$. Because $\pi(p',p)\subseteq \pi(s,p)\subseteq \R(s)$ and $\pi(q',q)\subseteq \pi(t,q)\subseteq \R(t)$, $\pi(p',p)$ and $\pi(q',q)$ do not intersect each other. As $p\precxi q$ and both points are in $\calD_p$, by the definition of  $\precxi$, $p'$ must be below $q'$.

\end{itemize}
\end{proof}

Note that the above also implies that for any $s\in S$, if $\xi_s(S)\neq \emptyset$, then $\xi_s(S)$ is a single connected portion of $\Xi(S)$.

\paragraph{Proof of Lemma~\ref{lem:boundary}.}
We now prove Lemma~\ref{lem:boundary} by using Lemma~\ref{lem:ext_dist}.
\begin{proof}
We first argue that $\calD\cap d$ consists of a single segment of $d$. Assume to the contrary that $\calD\cap d$ has two disconnected segments $d_1$ and $d_2$ such that there is a point $z\in d$ between them that is not in $\calD$. Then, since $\calD$ is connected, if we shoot a horizontal ray $\rho$ rightwards from $z$, $\rho$ must hit a point $p\in \calD$ (since otherwise $d_1$ and $d_2$ wound not be connected in $\calD$). Because $p\in \calD$, we know that $d(S,p)\leq 1$. Note that the horizontal segment $\overline{zp}$ is $\pi(z,p)$. By Lemma~\ref{lem:ext_dist}, $d(S,p')\leq d(S,p)\leq 1$ holds for any $p'\in \overline{zp}$, meaning that $\overline{zp}\subseteq \calD$ and thus $z\in \calD$. But this contradicts with $z\not\in \calD$.

We now argue that $\calD$ does not have a hole. Assume to the contrary that it has a hole $O$ that contains a point $z\not\in \calD$. Consider the shortest path $\pi(d,z)$. We extend the last segment of $\pi(d,z)$ beyond $z$ until a point $p$ on the boundary of $O$. Note that $\pi(d,p)=\pi(d,z)\cup \overline{zp}$.
Since $O$ is a hole of $\calD$ and $p$ is on the boundary of $O$, $p$ is also on the boundary of $\calD$. Hence, $d(S,p)\leq 1$. By Lemma~\ref{lem:ext_dist}, for any point $p'\in \pi(d,p)$, $d(S,p')\leq d(S,p)\leq 1$, meaning that the entire path $\pi(d,p)$ is inside $\calD$. As $z\in \pi(d,p)$, we obtain that $z$ is in $\calD$, contradicting with $z\not\in \calD$.
\end{proof}


\subsection{The data structure $\boldsymbol{\Psi(S)}$}
\label{sec:dsdescription}

In this section, we describe our data structure for $S$, denoted by $\Psi(S)$.

The data structure $\Psi(S)$ is a complete binary tree whose leaves from left to right store the points of $S$ under the order $\prec_L$. For each node $v$, let $S_v$ be the subset of points $S$ stored in the leaves of the subtree rooted at $v$.
We wish to store an implicit representation of $\Xi(S_v)$ at $v$. Recall that $I(S_v)$ is the set of points $s\in S_v$ such that $\xi_s(S_v)\neq \emptyset$. In the following, unless otherwise stated, we assume that $I(S_v)$ is a list of its points sorted under $\prec_L$.

If $v$ is the root, then we store $I(S_v)$ at $v$ explicitly.
If we store $I(S_v)$ at $v$ explicitly for all nodes $v\in \Psi(S)$, then $\Psi(S)$ would take super-linear space.
We instead employ a standard space saving technique~\cite{ref:HershbergerAp92,ref:OvermarsMa81}, where we store a point $s \in S$ in the highest node $v$ such that $s \in I(S_v)$. Using this strategy, the total space of $\Psi(S)$ is $O(n)$. The details are given below.

For a node $v$ with left child $u$ and right child $w$,
define $r_u = \max_{\prec_L} I(S_v) \cap I(S_u)$, i.e., among all points in $I(S_v) \cap I(S_u)$, $r_u$ is the largest one following the order $\prec_L$. Note that $r_u$ may not always exist, e.g., if $v$ does not have a left child or if $I(S_v) \cap I(S_u) = \emptyset$; we let $r_u=null$ if $r_u$ does not exist. Symmetrically, define $r_w = \min_{\prec_L} I(S_v) \cap I(S_w)$.
Lemma~\ref{lem:concatenate} establishes that $I(S_u) \cap I(S_v)$ is a prefix of $I(S_u)$ and $I(S_w) \cap I(S_v)$ is a suffix of $I(S_w)$. The proof of Lemma~\ref{lem:concatenate} relies on the following lemma, which will also be needed in other parts of the paper.

\begin{figure}[t]
\begin{minipage}[t]{\linewidth}
\begin{center}
\includegraphics[totalheight=1.5in]{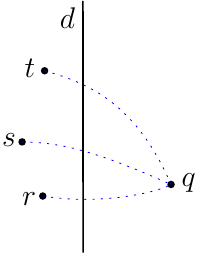}
\caption{Illustrating Lemma~\ref{lem:consistent2v3}. The three blue dotted curves are $\pi(r,q)$, $\pi(s,q)$, and $\pi(t,q)$. If $d(r, q) \leq d(s, q)$ and $d(r, q) \leq 1$, then $d(r, q) < d(t, q)$.}
\label{fig:threepoints10}
\end{center}
\end{minipage}
\end{figure}

\begin{lemma} \label{lem:consistent2v3}
    Suppose $S' = \{r, s, t\}$ with $r \prec_L s \prec_L t$ (or $t \prec_L s \prec_L r$) and $s \in I(S')$. Then, for any point $q \in P_R$, if $d(r, q) \leq d(s, q)$ and $d(r, q) \leq 1$, then $d(r, q) < d(t, q)$; see Figure~\ref{fig:threepoints10}.
\end{lemma}
\begin{proof}
We only prove the case $r \prec_L s \prec_L t$ since the other case $t \prec_L s \prec_L r$ can be handled analogously.

Define $Q_r=P_R\cap D_r$. Define $Q_s$ and $Q_t$ similarly. Let $Q = Q_r \cup Q_s \cup Q_t$.
Note that since $d(r,q)\leq 1$, we have $q\in Q_r$.
If $Q_t\cap Q_r=\emptyset$, then since $q\in Q_r$, we obtain $q\not\in Q_t$ and thus $d(t,q)>1$, meaning that $d(r,q)<d(t,q)$.
In the following, we assume $Q_t\cap Q_r\neq \emptyset$.

Suppose for the sake of contradiction that $d(t, q) \leq d(r, q)$. Then, we have $d(t, q) \leq d(r, q) \leq d(s, q)$.
As $s\in I(S')$, we know that $\xi_s(S')\neq\emptyset$. Furthermore, by our general position assumption,
    $\xi_s(S')$ cannot contain exactly one point.
    Consider a point $p \in \xi_s(S')$ such that $p \neq q$. Let $p' = \pi(s, p) \cap d$. See Figure~\ref{fig:threepoints}.

Because $d(t, q) \leq d(r, q) \leq 1$, we have $\pi(t,q)\subseteq D_t$ and $\pi(r,q)\subseteq D_r$. Hence, both $\pi(t, q)$ and $\pi(r, q)$ must be in $Q$.

Since $p\in \xi_s(S')$, we have $s=\beta_p(S')$.
    By applying Lemmas~\ref{lem:star_shaped} and \ref{lem:Lvor} to $\{r, s\}$ with $d(r,q)\leq d(s,q)$,
    we know that $\pi(r, q)$ must intersect $d$ at a point $a$ below $p'$.
    Similarly, applying the two lemmas to $\{s, t\}$ with $d(t,q)\leq d(s,q)$, we know that $\pi(t, q)$ must intersect $d$ at a point $b$ above $p'$. See Figure~\ref{fig:threepoints}.

\begin{figure}[h]
\begin{minipage}[t]{\linewidth}
\begin{center}
\includegraphics[totalheight=1.5in]{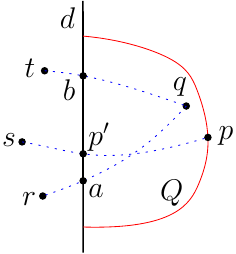}
\caption{Illustrating the proof of Lemma~\ref{lem:consistent2v3}. The red curve is the boundary of $Q$. The three blue dotted curves are $\pi(r,q)$, $\pi(t,q)$, and $\pi(s,p)$.}
\label{fig:threepoints}
\end{center}
\end{minipage}
\end{figure}

We claim that $Q_s\cap Q'\neq \emptyset$, where $Q'=Q_r\cup Q_t$. Indeed, since $Q_r\cap Q_t\neq \emptyset$, $\calD_R(\{s,t\})=Q'$ has a single connected component. By Lemma~\ref{lem:boundary}, $Q'\cap d$ is a single segment, denoted by $d_{tr}$. Since $\pi(t,q)\subseteq D_t$ and $\pi(r,q)\subseteq D_r$, both $a$ and $b$ are on $d_{tr}$. Since $p'$ is above $a$ and below $b$, we have $p'\in d_{tr}$. Since $p'\in \pi(s,p)$ and $d(s,p)\leq 1$, it follow that $d(s,p')\leq 1$ and $p'\in D_s$. As $p'\in d_{tr}=Q'\cap d$, we conclude that $D_s\cap Q'\neq \emptyset$.


The claim implies that $Q$ is a single connected component. Notice that $\pi(p', p)$ divides $Q$ into two region $Q_L$ and $Q_R$.
    Let $Q_L$ be the region relatively left of $\pi(p', p)$
    and $Q_R$ the one right of $\pi(p', p)$.
As $a$ is below $p'$ and $b$ is above $p'$, we have $a \in Q_R$ and $b \in Q_L$. Since $p\neq q$, $q$ is either in $Q_L$ or in $Q_R$, implying that either $\pi(a, q)$ or $\pi(b, q)$ must intersect $\pi(p', p)$.

    If $\pi(a, q)$ intersects $\pi(p',p)$ (see Figure~\ref{fig:threepoints}), then $\pi(r,q)$ intersects $\pi(s,p)$ since $a\in \pi(r,q)$ and $p'\in \pi(s,p)$. Since $d(r,q)\leq d(s,q)$ and $d(r,p)\geq d(s,p)$ (as $p\in \xi_s(S')$), if we consider the geodesic Voronoi diagram of $\{r,s\}$, then $p$ is in the Voronoi region of $s$ and $q$ is in the Voronoi region of $r$. As the Voronoi regions are star-shaped by Lemma~\ref{lem:star_shaped}(1) and $p\neq q$, $\pi(r,q)$ cannot intersect $\pi(s,p)$. We thus obtain contradiction.

    If $\pi(b, q)$ intersects $\pi(p',p)$, then $\pi(t,q)$ intersects $\pi(s,p)$. Since $d(t,q)\leq d(s,q)$, following similar  analysis to the above, we can also obtain contradiction.
\end{proof}

\begin{lemma} \label{lem:concatenate}
    $I(S_u) \cap I(S_v)$ is a prefix of $I(S_u)$ and $I(S_w) \cap I(S_v)$ is a suffix of $I(S_w)$.
\end{lemma}

\begin{proof}
We only prove that $I(S_u) \cap I(S_v)$ is a prefix of $I(S_u)$. By a similar method, we can also prove that $I(S_w) \cap I(S_v)$ is a suffix of $I(S_w)$.

Suppose to the contrary that $I(S_u) \cap I(S_v)$ is not a prefix of $I(S_u)$.
    Then there must exist $s_1, s_2 \in I(S_u)$ such that $s_1 \prec_L s_2$, $s_1 \notin I(S_v)$, and $s_2 \in I(S_v)$.
    Because $s_1 \in I(S_u)$, we know that $\xi_{s_1}(S_u) \neq \emptyset$.
    Pick any $q \in \xi_{s_1}(S_u)$. Since $s_1,s_2\in S_u$, $d(s_1,q)\leq d(s_2,q)$ and $d(s_1,q)\leq 1$.
    Because $s_1 \notin I(S_v)$, we know that there must exist some $s_3 \in I(S_v)$ such that $d(s_3, q) < d(s_1, q)$.
    Because $s_1 = \beta_q(S_u)$, it must be the case that $s_3 \in I(S_w)$.
    Therefore, $s_1 \prec_L s_2 \prec_L s_3$.

    Applying Lemma~\ref{lem:consistent2v3} to $S'=\{s_1,s_2,s_3\}$ with $d(s_1,q)\leq d(s_2,q)$ and $d(s_1,q)\leq 1$ (note that $s_2\in I(S')$ since $s_2\in I(S_v)$ and $S'\subseteq S_v$), we get that $d(s_1, q) < d(s_3, q)$,
    contradicting $d(s_3, q) < d(s_1, q)$.
\end{proof}

By Lemma~\ref{lem:concatenate}, $I(S_v)$ is the concatenation of $I(S_u)_{\leq r_u}$ and $I(S_w)_{\geq r_w}$, where the notation $I(S_u)_{\leq r_u}$ represents the prefix of $I(S_u)$ up to and including $r_u$ (if $r_u=null$, then $I(S_u)_{\leq r_u}=\emptyset$) and $I(S_w)_{\geq r_w}$ represents the suffix of $I(S_w)$ from and including $r_w$.
Similarly, we define $I(S_u)_{> r_u}=I(S_u)\setminus I(S_u)_{\leq r_u}$ and $I(S_w)_{< r_w}=I(S_w)\setminus I(S_w)_{\geq r_w}$.

We store both $r_u$ and $r_w$ in a node field $X(v)$ at $v$.
As $I(S_v)$ is the concatenation of $I(S_u)_{\leq r_u}$ and $I(S_w)_{\geq r_w}$, we only store $I(S_u)_{>r_u}$ in a list $arcs(u)$ at $u$ and similarly store $I(S_w)_{<r_w}$ in $arcs(w)$ at $w$. We use a doubly-linked list for each $arcs(\cdot)$ following the order of $\prec_L$.
This means that in a descent from the root of $\Psi(S)$, we can do the following.
Let $v$ be the current node in our descent and assume $arcs(v) = I(S_v)$, which is true initially when $v$ is the root. Suppose $v$ has left child $u$ and right child $w$.
We can cut $arcs(v)$ at $r_u$ to obtain $I(S_u)_{\leq r_u}$ and $I(S_w)_{\geq r_w}$.
As $I(S_u)_{> r_u}$ is stored in $arcs(u)$,
we can concatenate $I(S_u)_{\leq r_u}$ and $I(S_u)_{> r_u}$ to obtain $I(S_u)$.
We can do likewise to obtain $I(S_w)$.
In an ascent, we can simply reverse this process to obtain $I(S_v)$ from $I(S_u)$ and $I(S_w)$ using $r_u$ and $r_w$.
More details will be given later.


For each point $s\in I(S)$, we define the {\em ends} $a_s(S)$ and $b_s(S)$ of $\xi_s(S)$
to be the minimum and maximum points, respectively, of $\xi_s(S)$ under $\precxi$.

By Lemma~\ref{lem:concatenate}, for all $s \in I(S_u)_{< r_u}$,
the ends of $\xi_s(S_u)$ are also the ends of $\xi_s(S_v)$.
For $r_u$, we have $a_{r_u}(S_u) = a_{r_u}(S_v)$,
but $b_{r_u}(S_u) \neq b_{r_u}(S_v)$ is possible.
Likewise, for all $s \in I(S_w)_{> r_w}$,
the ends of $\xi_s(S_w)$ are also the ends of $\xi_s(S_v)$.
For $r_w$, we have $b_{r_w}(S_w) = b_{r_w}(S_v)$ holds, but $a_{r_w}(S_w) \neq a_{r_w}(S_v)$ is possible.

Based on these, for each $s \in S$,
we store in a field $ends(s)$ the ends of $\xi_s(S_v)$
where $v$ is the highest node such that $s \in I(S_v)$.
Also, for each node $v$ with left child $u$ and right child $w$,
we store $b_{r_u}(S_u)$ and $a_{r_w}(S_w)$ in a node field $Y_1(v)$ and store $b_{r_u}(S_v)$ and $a_{r_w}(S_v)$ in a node field $Y_2(v)$ at $v$.
With $Y_1(v)$ and $Y_2(v)$, we can easily obtain $ends(s)$ for all $s\in I(S_v)$ at each node $v$ during any descent/ascent of $\Psi(S)$.

Algorithm~\ref{algo:descend} gives pseudocode for how we can descend one level in a descent from the root. That is, assuming that we have $\Psi(S_v)$ at $v$, i.e., $I(S_v)$ is stored in $arcs(v)$ with $ends(s)$ storing the ends of $\xi_s(S_v)$ for each $s\in I(S_v)$, Algorithm~\ref{algo:descend} returns $\Psi(S_{v'})$ for each child $v'$ of $v$. In particular, for each $v'$, $arcs(v') = I(S_{v'})$, and for all $s \in S_{v'}$, $ends(s)$ contains the ends of $\xi_s(S_{v'})$.
Note that Algorithm~\ref{algo:descend} takes $O(1)$ time as $arcs(\cdot)$ is stored in a doubly-linked list. Algorithm~\ref{algo:ascend} shows how to ascend to $v$ to obtain $\Psi(S_v)$ given $\Psi(S_{v'})$ of each child $v'$ of $v$, which also takes $O(1)$ time.
Note also that Algorithm~\ref{algo:descend} requires $X(v)$ and $Y_1(v)$ while Algorithm~\ref{algo:ascend} requires $X(v)$ and $Y_2(v)$.

\begin{algorithm}
    \caption{Descend from a node $v$ to its children} \label{algo:descend}
    \KwIn{$\Psi(S_v)$}
    \KwOut{$\Psi(S_{v'})$ for each child $v'$ of $v$}
    \If{$v$ has no children}{
        \Return\;
    }
    \ElseIf{$v$ has only one child $v'$}{
        $arcs(v') \gets arcs(v)$\;
        $arcs(v) \gets \emptyset$\;
    }
    \ElseIf{$v$ has left child $u$ and right child $w$}{
        $r_u, r_w \gets X(v)$\;
        \If{both $r_u$ and $r_w$ exist}{
            $I(S_u)_{\leq r_u}, I(S_w)_{\geq r_w} \gets$ split $arcs(v)$ at $r_u$\;
            $I(S_u)_{> r_u} \gets arcs(u)$\;
            $I(S_w)_{< r_w} \gets arcs(w)$\;
            $arcs(u) \gets I(S_u)_{\leq r_u} + I(S_u)_{> r_u}$\;
            $arcs(w) \gets I(S_w)_{< r_w} + I(S_w)_{\geq r_w}$\;
            $b_{r_u}(S_u), a_{r_w}(S_w) \gets Y_1(v)$\;
            $a_{r_u}(S_v), b_{r_u}(S_v) \gets ends(r_u)$\;
            $a_{r_w}(S_v), b_{r_w}(S_v) \gets ends(r_w)$\;
            $ends(r_u) \gets \{a_{r_u}(S_v), b_{r_u}(S_u)\}$\;
            $ends(r_w) \gets \{a_{r_w}(S_w), b_{r_w}(S_v)\}$\;
        }
        \ElseIf{$r_u$ exists but $r_w=null$}{
            $arcs(u) \gets arcs(v)$\;
            $arcs(v) \gets \emptyset$\;
        }
        \ElseIf{$r_w$ exists but $r_u=null$}{
            $arcs(w) \gets arcs(v)$\;
            $arcs(v) \gets \emptyset$\;
        }
    }
\end{algorithm}


\begin{algorithm}
    \caption{Ascend to a node $v$ from its children} \label{algo:ascend}
    \KwIn{$\Psi(S_{v'})$ for each child $v'$ of $v$
    }
    \KwOut{$\Psi(S_v)$}
    \If{$v$ has only one child $v'$}{
        $arcs(v) \gets arcs(v')$\;
        $arcs(v') \gets \emptyset$\;
    }
    \ElseIf{$v$ has left child $u$ and right child $w$}{
        $r_u, r_w \gets X(v)$\;
        \If{both $r_u$ and $r_w$ exist}{
            $I(S_u)_{\leq r_u}, I(S_u)_{> r_u} \gets$ split $arcs(u)$ at $r_u$\;
            $I(S_w)_{< r_w}, I(S_w)_{\geq r_w} \gets$ split $arcs(w)$ at $r_w$\;
            $arcs(v) \gets I(S_u)_{\leq r_u} + I(S_w)_{\geq r_w}$\;
            $arcs(u) \gets I(S_u)_{> r_u}$\;
            $arcs(w) \gets I(S_w)_{< r_w}$\;
            $a_{r_u}(S_u), b_{r_u}(S_u) \gets ends(r_u)$\;
            $a_{r_w}(S_w), b_{r_w}(S_w) \gets ends(r_w)$\;
            $b_{r_u}(S_v), a_{r_w}(S_v) \gets Y_2(v)$\;
            $Y_1(v) \gets \{b_{r_u}(S_u), a_{r_w}(S_w)\}$\;
            $ends(r_u) \gets \{a_{r_u}(S_u), b_{r_u}(S_v)\}$\;
            $ends(r_w) \gets \{a_{r_w}(S_v), b_{r_w}(S_w)\}$\;
        }
        \ElseIf{$r_u$ exists but $r_w=null$}{
            $arcs(v) \gets arcs(u)$\;
            $arcs(u) \gets \emptyset$\;
        }
        \ElseIf{$r_w$ exists but $r_u=null$}{
            $arcs(v) \gets arcs(w)$\;
            $arcs(w) \gets \emptyset$\;
        }
    }
\end{algorithm}

\subsection{Handling queries with $\boldsymbol{\Psi(S)}$}
\label{sec:rangeemptyquery}

We now discuss how to use the tree $\Psi(S)$ to answer geodesic unit-disk range emptiness queries: Given a query point $q \in P_R$,
determine if $D_q \cap S = \emptyset$, and if not, report a witness. We will present an algorithm that runs in $O(\log n\log m)$ time.

\paragraph{Query algorithm description.}
Starting from the root of $\Psi(S)$, for each node $v$, if $v$ is a leaf, we check whether $d(s_v,q)\leq 1$, where $s_v$ is the point of $S$ stored in $v$. If so, we return $s$; otherwise, we return null, meaning that $D_q \cap S = \emptyset$. If $v$ has a single child, then we simply descend to it. Otherwise, let $u$ and $w$ be the left and right children of $v$, respectively. We decide whether to descend to $u$ or $w$ depending on the comparison of $q$ with $b_{r_u}(S_v)$ (or $a_{r_w}(S_v)$) under $\prec_R$. Recall that $b_{r_u}(S_v)$ and $a_{r_w}(S_v)$ are stored in $Y_2(v)$. To simplify the notation, let $b=b_{r_u}(S_v)$ and $a=a_{r_w}(S_v)$.

If $r_u=null$, then we descend to $w$. If $r_w=null$, then we descend to $u$. In the following, we assume both $r_u$ and $r_w$ exist.

If $q \prec_R b$, then we descend to $u$. If $b\prec_R q$, then we descend to $w$. The remaining case is when $q$ and $b$ is not $\prec_R$-comparable, which happens when either $q\in \pi(d,b)$ or $b\in \pi(d,q)$.

\begin{figure}[t]
\begin{minipage}[t]{\linewidth}
\begin{center}
\includegraphics[totalheight=1.7in]{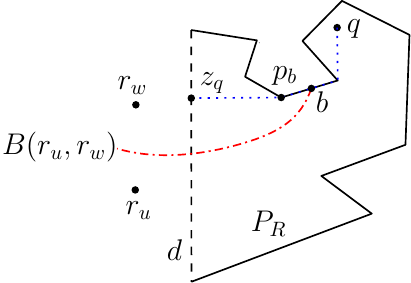}
\caption{Illustrating the case where $a=b\in \partial P_R$.}
\label{fig:query}
\end{center}
\end{minipage}
\end{figure}

First of all, if $q\in \pi(d,b)$, then by definition, $d(r_u,b)\leq 1$. Since $d(r_u,q)\leq d(r_u,b)$ by Lemma~\ref{lem:ext_dist}, we have $d(r_u,q)\leq 1$ and thus we simply return $r_u$. We next assume $b\in \pi(d,q)\setminus\{q\}$. If $b\not\in \partial P_R$ or $b\neq a$, then we return $null$. Otherwise, $b= a$ is on an edge $e$ of $\partial P_R$. As $b= a$, $b$ is on the bisector of $r_u$ and $r_w$. Due to our general position assumption, $b$ cannot be a vertex of $P$. This means that $b$ is in the interior of $e$. As $b$ is in the interior of $\pi(z_q,q)$, exactly one vertex $p_b$ of $e$ must be contained in $\pi(z_q,b)$; see Figure~\ref{fig:query}. We can find $p_b$ by checking whether $d(z_q,b)=d(z_q,p')+|\overline{p'b}|$ for each vertex $p'$ of $e$. Then, we check whether $P$ is locally on the right side of the directed segment $\overrightarrow{p_bb}$ from $p_b$ to $b$. If so, we descend to $u$; otherwise, we descend to $w$.


Algorithm~\ref{algo:DEDquery} gives the pseudocode of the query algorithm.

\begin{algorithm}
    \caption{Query$(v,q)$} \label{algo:DEDquery}
    \KwIn{A node $v$ of $\Psi(S)$ and a query point $q \in P_R$.}
    \KwOut{A point in $D_q \cap S_v$ or $null$ if $D_q \cap S_v = \emptyset$.}
    $u, w \gets$ left and right children of $v$, if they exist, respectively\;
    $b_{r_u}(S_v), a_{r_w}(S_v) \gets Y_2(v)$\;
    $b\gets b_{r_u}(S_v), a\gets a_{r_w}(S_v)$\;
    \If{$v$ is a leaf}{
        $\{s_v\} \gets arcs(v)$\;
        \If{$d(s_v, q) \leq 1$}{
            \Return $s_v$\;
        }
        \Else{
            \Return $null$\;
        }
    }
    \If{$u=null$ or $r_u=null$}{
        \Return Query$(w, q)$\label{line:nosu}\;
    }
    \If{$w=null$ or $r_w=null$}{
        \Return Query$(u, q)$\label{line:nosw}\;
    }
    \If{$q \prec_R b$}{
        \Return Query$(u, q)$\label{line:qlrx}\;
    }
    \If{$b \prec_R q$}{
        \Return Query$(w, q)$\label{line:xlrq}\;
    }
    \If{$q \in \pi(d, b)$}{
        \Return $r_u$ \label{line:abneq}\;
    }
    \If{$b \in \pi(d, q)\setminus\{q\}$}{
        \If{$b\not\in \partial P_R$ or $b\neq a$}{
            \Return $null$\;
        }
        \Else
        {
            compute the vertex $p_b$\;
            \If{$P$ is locally on the right of $\overline{p_bb}$}
            {
                \Return  Query$(u, q)$\;
            }
            \Else
            {
                \Return  Query$(w, q)$\;
            }
        }
    }
\end{algorithm}

\subsubsection{Query time analysis}

We now analyze the query time. All the discussions here are based on the assumption that we already have a GH data structrue for $P$.
We start with the following two lemmas.

\begin{lemma} \label{lem:projection}
Suppose a GH data structure has been built for $P$. Then, given any point $p\in P_R$, the point $z_p$ on $d$ can be computed in $O(\log m)$ time.
\end{lemma}
\begin{proof}
Let $c$ be the junction vertex of $\pi(p,z^*_0)$ and $\pi(p,z^*_1)$.
Consider the funnel $F$ bounded by $d$, $\pi(c,z^*_0)$ and $\pi(c,z^*_1)$.
There is a single vertex $v$ in the two sides of $F$ such that the horizontal line $\ell_v$ through $v$ is tangent to the side containing $v$. The intersection of $\ell_v$ with $d$ is $z_p$. Therefore, to compute $z_p$, it suffices to compute $v$. As the slopes of the tangents at the points of the two sides continuously change from $z^*_0$ to $c$ along $\pi(c,z^*_0)$ and then from $c$ to $z^*_1$ along $\pi(c,z^*_1)$, we can find $v$ by binary search on the two sides of $F$.

Specifically, using the GH data structure, we can compute $c$ in $O(\log m)$ time. The data structure can also obtain a binary search tree of height $O(\log m)$ representing $\pi(c,z^*_0)$ (resp., $\pi(c,z^*_1)$) in $O(\log m)$ time. Using the two trees, the vertex $v$ can be computed in $O(\log m)$ time by binary search (similar techniques have been used in the literature, e.g., \cite{ref:PollackCo89}).
\end{proof}

\begin{lemma}\label{lem:orderRalgo}
Suppose a GH data structure has been built for $P$.
    For any two points $p_1, p_2 \in P_R$,
    we can determine in $O(\log m)$ time if $p_1 \prec_R p_2$, $p_2 \prec_R p_1$, or that they are not $\prec_R$-comparable, and in the case that they are not $\prec_R$-comparable, we can determine whether $p_1 \in \pi(d, p_2)$ or $p_2 \in \pi(d, p_1)$.
\end{lemma}
\begin{proof}
We first compute $z_{p_1}$ and $z_{p_2}$ in $O(\log m)$ time by Lemma~\ref{lem:projection}.
If $z_{p_1}\neq z_{p_2}$, then either $p_1 \prec_R q_1$ or $p_2 \prec_R p_1$, which can be easily determined. In the following, we assume $z_{p_1}= z_{p_2}$. Let $z=z_{p_1}$.

    We can use the GH-data structure to determine in $O(\log m)$ time whether $p_1 \in \pi(z, p_2)$ or $p_2 \in \pi(z, p_1)$ by computing $d(z,p_1)$, $d(z,p_2)$, and $d(p_1,p_2)$.
    If this is the case, then $p_1$ and $p_2$ are not $\prec_R$-comparable.
    Otherwise, we can use the GH-data structure to find in $O(\log m)$ time the junction vertex $c$ of $\pi(z,p_1)$ and $\pi(z,p_2)$ as well as the two adjacent vertices of $c$ in each path. Using these vertices, we can determine whether $p_1 \prec_R p_2$ or $p_2 \prec_R p_1$ in additional $O(1)$ time.
\end{proof}

With the above two lemmas, we now analyze the time complexity of the query algorithm.

The algorithm traverse $\Psi(S)$ along a path from the root to a leaf.
Let $v$ be a node in the path. If $v$ is a leaf, we need to check whether $d(s_v,q)\leq 1$, which can be done in  $O(\log m)$ time using the GH data structure. If $v$ has a single child, we spend $O(1)$ time to descend to its child. It remains to consider the case where $v$ has both a left child $u$ and the right child $w$. Again, let $b=b_{r_u}(S_v)$ and $a=a_{r_w}(S_v)$.

By Lemma~\ref{lem:orderRalgo}, we can determine in $O(\log m)$ time whether $p\prec_R b$, $b\prec_R p$, or they are not $\prec_R$-comparable, and in the case where they are not $\prec_R$-comparable, whether $q\in \pi(d,b)$ or $b\in \pi(d,q)$. The remaining steps of the algorithm at $v$ take $O(\log m)$ in total.

In summary, the algorithm spends $O(\log m)$ time at each node $v$. As the height of $\Psi(S)$ is $O(\log n)$, the total time of the query algorithm is $O(\log n\log m)$.

\subsubsection{Correctness of the query algorithm}

We now show that the algorithm is correct. First of all, if the algorithm returns a point $s$, then according to the algorithm, $d(s,q)\leq 1$ holds and thus the algorithm is correct in this case. On the other hand, suppose that $D_q\cap S\neq\emptyset$. We next argue that the algorithm must return a point $s$ such that $s\in D_q$, i.e., $d(s,q)\leq 1$. We do so by induction.

Let $v$ be a node of $\Psi(S)$ that is considered by the query algorithm. We assume that $D_q\cap S_v\neq \emptyset$, which is true initially when $v$ is the root since $S=S_v$ and $D_q\cap S\neq \emptyset$. If $v$ is a leaf, then $S_v$ contains a single point $s$. As $D_q\cap S_v\neq \emptyset$, we have $s\in D_q$ and thus the algorithm is correct. Assume that $v$ is not a leaf.
It suffices to argue that if we descend to a child $v'$, then it must be $D_q\cap S_{v'}\neq\emptyset$. If $v$ has a single child $v'$, then $S_v=S_{v'}$ and thus it is obviously true that $D_q\cap S_{v'}\neq\emptyset$. It remains to argue the case where $v$ has a left child $u$ and a right child $w$. Again, let $b=b_{r_u}(S_v)$ and $a=a_{r_w}(S_v)$.

According to our algorithm, if $r_u=null$, then we descend to $w$, and if $r_w= null$, we descend to $u$. The correctness follows the following lemma.
\begin{lemma}
If $r_u=null$, then $D_q\cap S_w \neq\emptyset$. If $r_w=null$, then $D_q\cap S_u \neq\emptyset$.
\end{lemma}
\begin{proof}
If $r_u=null$, then by definition, $\Xi(S_w)=\Xi(S_v)$, which implies that $\calD_R(S_w)=\calD_R(S_v)$. Since $D_q\cap S_v\neq \emptyset$, $q\in \calD_R(S_v)$. Hence, $q\in \calD_R(S_w)$, meaning that $D_q\cap S_w \neq\emptyset$.

If $r_w=null$, a similar argument to the above can show that $D_q\cap S_u \neq\emptyset$.
\end{proof}

We now assume that both $r_u$ and $r_w$ exist. In this case, according to our algorithm, if $q\prec_R b$, then we descend to $u$, and if $b\prec_R q$, we descend to $w$. The following lemma justifies the correctness of this step.

\begin{lemma} \label{lem:invlr}
    If $q \prec_R b$, then $D_q \cap S_u\neq\emptyset$.
    If $b \prec_R q$, then $D_q \cap S_w\neq\emptyset$.
\end{lemma}
\begin{proof}
We first prove the case $q\prec_R b$.

As $D_q\cap S_v\neq \emptyset$, we have $q\in \calD_R(S_v)$. Depending on whether $q\in \Xi(S_v)$, there are two subcases.

\paragraph{The subcase $q\in \Xi(S_v)$.}
If $q\in \Xi(S_v)$, then we have
            \begin{equation} \label{eq:invlrab}
                q \prec_R b
                \os 1 \implies q \precxi b
                \os 2 \implies \beta_q(S_v) \prec_L r_u
                \os 3 \implies \beta_q(S_v) \in S_u,
            \end{equation}
            where (1) is due to Lemma~\ref{lem:consistentRT} by replacing $S$ with $S_v$,
            (2) is due to Lemma~\ref{lem:consistentTL}  by replacing $S$ with $S_v$ and $r_u=\beta_b(S_v)$,
            and (3) is due to the definition of $S_u$.

As $D_q\cap S_v\neq \emptyset$, we know that $\beta_q(S_v)\in D_q$. Since $\beta_q(S_v)\in S_u$, we obtain that $D_q\cap S_u\neq\emptyset$.

\paragraph{The subcase $\boldsymbol{q\not\in \Xi(S_v)}$.}
We now discuss the case $q\not\in \Xi(S_v)$. Recall our general position assumption that $q\not\in d$. Hence,  $z_q\neq q$. Since $q$ is inside $\calD_R(S_v)$, we extend the last segment of $\pi(z_q,q)$ beyond $q$ until a point $q'$ on the boundary of $\calD_R(S_v)$; see Figure~\ref{fig:querycorrect}. Recall that the boundary of $\calD_R(S_v)$ consists of $\Xi(S_v)$ and some segments of $d$. We claim that $q'$ cannot be on $d$. To see this, notice that $\pi(z_q,q')=\pi(z_q,q)\cup \overline{qq'}$. As $z_q\in d$, $q\not\in d$, and $q\in \pi(z_q,q')$, $q'$ cannot be on $d$ (otherwise the shortest path $\pi(z_q,q')$ is just $\overline{z_qq'}$). Therefore, $q'\in \Xi(S_v)$.

As $q\prec_R b$, we next argue that $q'\prec_R b$. Indeed, as $\pi(z_q,q)\cup\overline{qq'}$ is $\pi(z_q,q')$, we have $z_{q'}=z_q$. As $q\prec_R b$, if $z_q=z_{q'}$ is below $z_b$, then it immediately follows that $q'\prec_R b$. Otherwise, $z_q=z_b$; see Figure~\ref{fig:querycorrect}. Since $\pi(z_q,q')=\pi(z_q,q)\cup\overline{qq'}$, the junction vertex $c$ of $\pi(z_q,q)$ and $\pi(z_q,b)$ is also the junction vertex of $\pi(z_q,q')$ and $\pi(z_q,b)$. Since $q\prec_R b$, the counterclockwise ordering of $\pi(z_q,q)$ and $\pi(z_q,b)$ around $c$ is $\pi(z_q,c)$, $\pi(c,q)$, and $\pi(c,b)$. As $\pi(c,q')=\pi(c,q')\cup \overline{qq'}$, the counterclockwise ordering of $\pi(z_q,q')$ and $\pi(z_q,b)$ around $c$ is $\pi(z_q,c)$, $\pi(c,q')$, and $\pi(c,b)$. We thus obtain $q'\prec_R b$.

\begin{figure}[h]
\begin{minipage}[t]{\linewidth}
\begin{center}
\includegraphics[totalheight=1.7in]{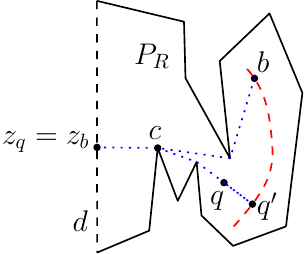}
\caption{Illustrating the case where $q\not\in \Xi(S_v)$ and $z_q=z_b$. The red curve is a portion of $\Xi(S_v)$.}
\label{fig:querycorrect}
\end{center}
\end{minipage}
\end{figure}


Replacing $q$ by $q'$ in Equation~\eqref{eq:invlrab} leads to $\beta_{q'}(S_v)\in S_u$. Let $s=\beta_{q'}(S_v)$.
Since $q'\in \calD_R(S_v)$, we have $d(s,q')\leq 1$. As $q\in \pi(d,q')$, by Lemma~\ref{lem:ext_dist}, $d(s,q)<d(s,q')\leq 1$. Hence, $s\in D_q$. Since $s\in S_u$, we obtain that $D_q\cap S_u\neq \emptyset$.

The above proves that if $q\prec_R b$, then $D_q \cap S_u\neq\emptyset$.

\paragraph{The case $\boldsymbol{b\prec_R q}$.}
We now prove that the other case $b\prec_R q$. If $a\prec_R q$, then following a symmetric argument to the above, we can prove $D_q\cap S_w\neq\emptyset$. In the following, we prove that $a\prec_R q$ must hold.

If $a=b$, then since $b\prec_R q$, it obviously holds that $a\prec_R q$. So we assume that $a\neq b$. Then, both $a$ and $b$ must be on $d$. By the definitions of $a$ and $b$ and Lemma~\ref{lem:Lvor}, $b$ is below $a$ and $\calD_R(S_v)$ does not intersect any point in the interior of $\overline{ab}$. Assume to the contrary that $a\prec_R q$ is not true. Then, since $b\prec_R q$, $z_q$ must be on $\overline{ab}$. By the definitions of $a$ and $b$, we have $d(S_v,a)=D(S_v,b)=1$. As $\calD_R(S_v)$ does not intersect any point in the interior of $\overline{ab}$, $d(S_v,p)\geq 1$ for any $p\in \overline{ab}$, and in particular, $d(S_v,z_q)\geq 1$. As $q\not\in d$, we have $z_q\in \pi(d,q)\setminus \{q\}$.
By Lemma~\ref{lem:ext_dist}, $d(S_v,q)>d(S_v,z_q)\geq 1$. We thus obtain $D_q\cap S_v=\emptyset$, contradicting $D_q\cap S_v\neq \emptyset$.

The lemma thus follows.
\end{proof}

We next discuss the case where $q$ and $b$ are not $\prec_R$-comparable.
If $q\in \pi(d,b)$, then we return $r_u$ and the correctness has already been justified.
We next assume that $b\in \pi(d,q)\setminus\{q\}$. We first have the following lemma.
\begin{lemma}
If $b\in \pi(d,q)\setminus\{q\}$, then $b=a\in \partial P_R$.
\end{lemma}
\begin{proof}
We first prove that $b=a$. Assume to the contrary that $a\neq b$. Then, by definition, $b$ is on $d$ and $d(S_v,b)=1$. As $b\in \pi(d,q)\setminus\{q\}$, by Lemma~\ref{lem:ext_dist}, we have $d(S_v,b)<d(S_v,q)$. Hence, $d(S_v,q)>1$, implying that $D_q\cap S_v=\emptyset$. But this contradicts $D_q\cap S_v\neq \emptyset$.

Similarly, we can argue that $b\in \partial P_R$. Assume to the contrary this is not true. Then, we also have $d(S_v,b)=1$. Following the same analysis as above, we can obtain contradiction.
\end{proof}

By the above lemma, it suffices to consider the case where $b=a$ and $b$ is on an edge $e$ of $\partial P_R$. According to our algorithm, if $P$ is on the right side of $\overrightarrow{p_bb}$, then we descend to $u$; otherwise, we descend to $w$. The following lemma justifies the correctness.

\begin{lemma}
If $P$ is on the right side of $\overrightarrow{p_bb}$, then $D_q\cap S_u\neq\emptyset$; otherwise, $D_q\cap S_w\neq\emptyset$.
\end{lemma}
\begin{proof}
We only prove the ``right side'' case since the other case can be handled analogously.

Let $\calD$ be the connected component of $\calD_R(S_v)$ containing $q$. We claim that $\pi(d,q)\subseteq \calD$. To see this, let $s=\beta_q(S_v)$. As $D_q\cap S_v\neq \emptyset$, $s\in D_q$ and thus $q\in D_s$ and $d(s,q)\leq 1$. Hence, $D_q\cap P_R\subseteq \calD$.
By Lemma~\ref{lem:ext_dist}, $d(s,p)\leq d(s,q)\leq 1$ holds for any $p\in \pi(d,q)$. Hence, $\pi(d,q)\subseteq D_q\subseteq \calD$. As $b\in \pi(d,q)$, $b\in \calD$.
Let $C$ denote the portion of $\partial \calD$ excluding the points on $d$. By Lemma~\ref{lem:boundary}, $C$ is a curve connecting two points $c_1$ and $c_2$ of $d$. We assume that $c_1$ lies above $c_2$.

Let $z=z_b$, which is also $z_q$ since $b\in \pi(d,q)$. As $b\in e$ and $b\in \calD$, $b$ must be on $C$. The point $b$ divides $C$ into two sub-curves, with one $C_1$ consists of all points after $b$ under the traversal order $\precxi$ and the other one $C_2=C\setminus C_1$. See Figure~\ref{fig:querycorrect20}.
Since $\pi(z,b)\subseteq \calD$, both $z$ and $b$ are on the boundary of $\calD$. The path $\pi(z,b)$ partitions $\calD$ into two portions $\calD_1$ and $\calD_2$, with $\calD_1$ bounded by $\pi(z,b)$, $C_1$, and $\overline{c_1z}$, and $\calD_2$ bounded by $\pi(z,b)$, $C_2$, and $\overline{c_2z}$.

\begin{figure}[h]
\begin{minipage}[t]{\linewidth}
\begin{center}
\includegraphics[totalheight=1.9in]{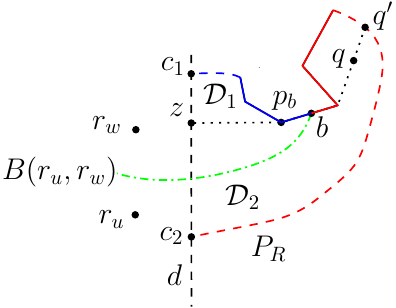}
\caption{The blue segments and the blue dashed curve constitute $C_1$, while the red segments and the red dashed curve constitute $C_2$. The blue and red segments are in $\partial P_R$.}
\label{fig:querycorrect20}
\end{center}
\end{minipage}
\end{figure}

As $q\in \calD$, if we extend the last edge of $\pi(z,q)$ beyond $q$, it will hit a point $q'\in \partial \calD$; see Figure~\ref{fig:querycorrect20}. Hence, $\pi(z,q')=\pi(z,q)\cup \overline{qq'}$. As $z\in d$ and $q\not\in d$, $q'$ cannot be on $d$. Hence, $q'\in C$. We claim that $q'$ must be on $C_2$. To see this, since $P$ is on the right side of $\overrightarrow{p_bb}$, $\pi(s,q')$ will enter $\calD_2$ right after $b$. Because $C_1$ is bounding $\calD_1$, to have $q'\in C_1$, $\pi(b,q')\setminus \{b\}$ must intersect $\pi(z,b)$. But this is not possible since $\pi(z,b)$ and $\pi(b,q')\setminus \{b\}$ are two subpaths of $\pi(z,q')$.

Since $q'\in C_2$ (and $q'\neq b$), by the definition of $C_2$, $q'\precxi b$. By Lemma~\ref{lem:consistentTL}, $\beta_{q'}(S_v)\prec_L \beta_{b}(S_v)=r_u$. As $r_u\in S_u$, we have $\beta_{q'}(S_v)\in S_u$. Let $s=\beta_{q'}(S_v)$. As $q\in \pi(d,q')$, by Lemma~\ref{lem:ext_dist}, $d(s,q)\leq d(s,q')$. Since $q'\in C_2$, $d(s,q')\leq 1$. Therefore, $d(s,q)\leq 1$, and thus $s\in D_q$. Since $s\in S_u$, we conclude that $D_q\cap S_u\neq \emptyset$.

\end{proof}

\subsection{Constructing $\boldsymbol{\Psi(S)}$}
\label{sec:dsconstruct}


Before describing the construction algorithm for $\Psi(S)$, we first give an algorithm to compute the ends in the following lemma. The algorithm will be needed in our construction algorithm for $\Psi(S)$ and also in our algorithm for handling deletions in Section~\ref{sec:deletion}. Again, we assume that a GH data structure has been computed for $P$.

\begin{lemma} \label{lem:ends_compute}
For any $S' \subseteq S$ and any $s \in I(S')$, if we know the predecessor and successor of $s$ in $I(S')$ under the order $\prec_L$, then the ends $a_s(S')$ and $b_s(S')$ of $\xi_s(S')$ can be computed in $O(\log^2 m)$ time.
\end{lemma}
\begin{proof}
    We will show how to compute $b_s(S')$ using the successor $t$ of $s$ in $I(S')$.
    Computing $a_s(S')$ can be done similarly using the predecessor of $s$.

    Recall that $\Xi(S')$ is a collection of disjoint curves whose endpoints are on $d$.
    Let $C$ be the curve containing $b_s(S')$. Note that there are two cases for $b_s(S')$: (1) it is the upper endpoint of $C$; (2) it is in the interior of $C$.
    In the first case, $b_s(S')$ is the upper endpoint of $d\cap D_s$. In the second case, by Lemma~\ref{lem:consistentTL}, $a_t(S') = b_s(S')= B(s,t) \cap \Xi(\{s, t\})$.

    Let $d_s=d\cap D_s$ and $z'$ the upper endpoint of $d_s$. Assuming that $z'$ has been computed, we can determine whether $b_s(S')$ belongs to the first or second case as follows. If $t$ does not exist, then obviously $b_s(S')$ belongs to the first case. If $t$ exists, then since $d(t, z') = 1$ is not possible due to our general position assumption, $d(t, z') > 1$ if and only if $b_s(S')$ belongs to the first case. As such, once $z'$ is known, we can determine whether it is the first or second case in $O(\log m)$ time by computing $d(t,z')$ using the GH data structure.

    In what follows, we first discuss how to compute $z'$ and then present an algorithm to compute $B(s,t) \cap \Xi(\{s, t\})$ in the second case.

    \paragraph{Computing the point $\boldsymbol{z'}$.}
    First of all, if $d(s, z^*_1) \leq 1$, we immediately obtain $z' = z^*_1$. Otherwise, we do the following. Let $c$ be the junction vertex of $\pi(s,z^*_0)$ and $\pi(s,z^*_1)$. We consider the funnel $F_s(d)$, formed by $\pi(c,z^*_0)$, $\pi(c,z^*_1)$, and $d$. Define $z^* = \argmin_{z \in d} d(s, z)$ (see Figure~\ref{fig:endscompute}), which can be computed in $O(\log m)$ time in a similar way to the algorithm of Lemma~\ref{lem:projection}.

    \begin{figure}
        \centering
        \includegraphics[width=.45\textwidth]{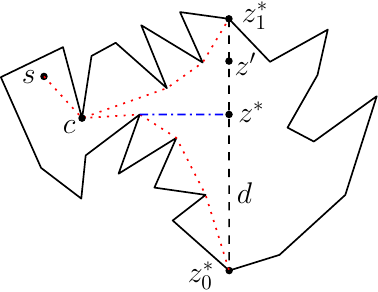}
        \caption{Illustrating the proof of Lemma~\ref{lem:ends_compute}.}
        \label{fig:endscompute}
    \end{figure}

    Note that $z'\in \overline{z^*z^*_1}$. In addition, if we move a point $p\in \overline{z^*z^*_1}$ from $z^*$ to $z^*_1$, $d(s,p)$ is strictly increasing~\cite{ref:PollackCo89}. Hence, we can do binary search on $\pi(c, z^*_0)$ and $\pi(c, z^*_1)$ to compute $z'$. Using the GH data structure, we can obtain a binary search tree representing $\pi(c, z^*_0)$ (resp., $\pi(c, z^*_1)$). With these binary search trees, $z'$ can be computed in $O(\log m)$ time (similar methods have been used before~\cite{ref:AgarwalIm18,ref:OhVo20}).

    \paragraph{Computing the intersection $\boldsymbol{B(s,t) \cap \Xi(\{s, t\})}$.}
    We now discuss how to compute $B(s,t) \cap \Xi(\{s, t\})$ in the second case. Let $B_R(s,t)=B(s,t)\cap P_R$. Hence, $b_s(S')=B_R(s,t)\cap \Xi(\{s, t\})$ and $B_R(s,t)\neq \emptyset$. To simplify the notation, let $b^*=b_s(S')$.

    Since $s,t\in P_L$, it is known that $B_R(s,t)$ is an $x$-monotone curve connecting a point $p_1 \in d$ with a point $p_2 \in \partial P_R$ and no interior point of the curve intersects $\partial P$~\cite{ref:AgarwalIm18}. Furthermore, if we move a point $p$ from $p_1$ to $p_2$ along $B_R(s,t)$, then $d(s,p)$ is strictly increasing~\cite{ref:AronovOn89}. We will use this property to perform binary search on $B_R(s,t)$ to search $b^*$.

    To compute $b^*=B_R(s,t) \cap \Xi(\{s, t\})$, we first compute $p_1$ and $p_2$, which can be done in $O(\log^2 m)$ time~\cite{ref:AgarwalIm18} (alternatively, we could also use Lemmas~\ref{lem:algointerdiagonal} and \ref{lem:bisectorintersection}). If $d(s, p_2) \leq 1$, then $p_2 \in \Xi(\{s, t\})$ and thus $b^* = p_2$.
    Otherwise, $b^*$ is the point on $B_R(s,t)$ such that $d(s, b^*) = 1$ and we can find $b^*$ by binary searh as follows.

    Using the GH data structure, it is possible to obtain a representation of $B_R(s,t)$ so that the $i$-th vertex of $B_R(s,t)$ from $p_1$ can be computed in $O(\log m)$ time~\cite{ref:AgarwalIm18}.
    To find $b^*$, we can do binary search on $B_R(s,t)$ by accessing the $i$-th vertex in each iteration to find the predecessor and successor of $b^*$ in $B_R(s,t)$. More specifically, in each iteration, let $v$ be the vertex of $B_R(s,t)$ examined in the iteration. We compare $d(s,v)$ with $1$. If $d(s,v)=1$, then we simply return $b^*=v$. If $d(s,v)<1$, then the predecessor of $b^*$ is $v$ or after $v$. If $d(s,v)>1$, then the predecessor is before $v$.
    Each iteration takes $O(\log m)$ time. As $B_R(s,t)$ has $O(m)$ vertices, there are $O(\log m)$ iterations in the binary search. Hence, the predecessor can be found in $O(\log^2 m)$ time. The successor of $b^*$ can be determined similarly. Note that the edge of $B_R(s,t)$ connecting the predecessor and successor of $b^*$ contains $b^*$. Using the edge we can compute $b^*$ in additional $O(1)$ time.

    Alternatively, we could compute $b^*$ use a similar approach to the algorithm for computing $B_{ij}\cap B_{jk}$ in the proof of Lemma~\ref{lem:bisectorintersection}, as follows. Let $c_s$ be the junction vertex of $\pi(s,p_1)$ and $\pi(s,p_2)$, and $\gamma_s$ concatenation of $\pi(c_s,p_1)$ and $\pi(c_s,p_2)$. Similarly, we define $c_t$ and $\gamma_t$ for $t$. We do binary search on the vertices of $\gamma_s$. In each iteration, given a vertex $p_s\in \gamma_s$, we can compute in $O(\log m)$ time $f(p_s)\in \gamma_t$  and a point $p$ that is the intersection of $B_R(s,t)$ and the tangent ray at $p_s$. By comparing $d(s,p)$ with $1$, we can decide which way to search on $\gamma_s$. Since each iteration takes $O(\log m)$ time, $b^*$ can be computed in $O(\log^2 m)$ time.

    In summary, in the second case $b_s(S')$ can be computed in $O(\log^2 m)$ time.
\end{proof}



\subsubsection{Construction algorithm of $\boldsymbol{\Psi(S)}$}

We now describe the construction algorithm for $\Psi(S)$.
We construct $\Psi(S)$ recursively. For each node $v$, we need to construct the node fields $arcs(v)$, $X(v)$, $Y_1(v)$, and $Y_2(v)$. Algorithm~\ref{algo:build} gives the full pseudocode.

Initially when $v$ is a leaf, $S_v$ contains a single point $s$. Then, we set $arcs(v)=\{s\}$, and set $X(v)$, $Y_1(v)$, and $Y_2(v)$ all to $\emptyset$. In addition, we compute $ends(s) = \{a_s(\{s\}), b_s(\{s\})\}$ by Lemma~\ref{lem:ends_compute}.

In general, suppose $v$ has a left child $u$ and a right child $w$ such that $||S_u| - |S_w|| \leq 1$ and the subtrees $\Psi(S_u)$ and $\Psi(S_w)$ at both $u$ and $w$ have been constructed. In particular, for each point $s\in S_u$, $ends(s) = \{a_s(S_u), b_s(S_u)\}$ are available, and for each $s\in S_w$, $ends(s) = \{a_s(S_w), b_s(S_w)\}$ are also available. We  construct $\Psi(S_v)$ as follows.

 We first compute $r_u$ and $r_w$, whose algorithm will be given later. We then set $X(v)=\{r_u,r_w\}$. Next we split $arcs(u)$ at $r_u$ and split $arcs(w)$ at $r_w$, and merge the portion of $arcs(u)$ before and including $r_u$ with the portion of $arcs(w)$ after and including $r_w$ and set $arcs(v)$ to the merged list. The remaining portion of $arcs(u)$ is the new $arcs(u)$, and the remaining portion of $arcs(w)$ is the new $arcs(w)$. In addition, we set $Y_1(v) = \{b_{r_u}(S_u), a_{r_w}(S_w)\}$ (these two ends are already available), and compute $b_{r_u}(S_v)$ and $a_{r_w}(S_v)$ using Lemma~\ref{lem:ends_compute} (note that we can access the predecessors and successors of $r_u$ and $r_w$ from $arcs(v)$), and add them to $Y_2(v)$. Note that for each point $s$ in $arcs(v)$ before $r_u$, we have $a_s(S_v)=a_s(S_u)$ and $b_s(S_v)=b_s(S_u)$, and therefore, we do not need to reset $ends(s)$. Similarly, for each point $s$ in $arcs(v)$ after $r_w$, we have $a_s(S_v)=a_s(S_w)$ and $b_s(S_v)=b_s(S_w)$, and therefore, we do not need to reset $ends(s)$. But we reset $ends(r_u)$ to $\{a_{r_u}(S_u), b_{r_u}(S_v)\}$ and reset $ends(r_w)$ to $\{a_{r_w}(S_v), b_{r_w}(S_w)\}$.

\begin{algorithm}
    \caption{Build($S$)} \label{algo:build}
    \KwIn{$S$}
    \KwOut{$\Psi(S)$}
    Create a node $v$\;
    $S_v\gets S$\;
    \If{$|S_v| = 1$}{
        Let $s$ be the only point in $S_v$\;
        $arcs(v) \gets \{s\}$\;
        Compute $a_s(\{s\})$ and $b_s(\{s\})$\;
        $ends(s) \gets \{a_s(\{s\}), b_s(\{s\})\}$\;
    }
    \If{$|S_v| \geq 2$}{
        $S_u, S_w \gets$ split $S_v$ in equal halves where points of $S_u$ are before $S_w$ under $\prec_L$\;
        $u \gets$ Build$(S_u)$\;
        $w \gets$ Build$(S_w)$\;
        Set $u$ and $w$ as the left and right children of $v$, respectively\;
        Compute $r_u$ and $r_w$\;
        $X(v) \gets \{r_u, r_w\}$\;
        \If{both $r_u$ and $r_w$ exist}{
            $I(S_u)_{\leq r_u}, I(S_u)_{> r_u} \gets$ split $arcs(u)$ at $r_u$\;
            $I(S_w)_{< r_w}, I(S_w)_{\geq r_w} \gets$ split $arcs(w)$ at $r_w$\;
            $arcs(v) \gets I(S_u)_{\leq r_u} + I(S_w)_{\geq r_w}$\;
            $arcs(u) \gets I(S_u)_{> r_u}$\;
            $arcs(w) \gets I(S_w)_{< r_w}$\;
            $a_{r_u}(S_u), b_{r_u}(S_u) \gets ends(r_u)$\;
            $a_{r_w}(S_w), b_{r_w}(S_w) \gets ends(r_w)$\;
            Compute $b_{r_u}(S_v)$ and $a_{r_w}(S_v)$\;
            $Y_1(v) \gets \{b_{r_u}(S_u), a_{r_w}(S_w)\}$\;
            $Y_2(v) \gets \{b_{r_u}(S_v), a_{r_w}(S_v)\}$\;
            $ends(r_u) \gets \{a_{r_u}(S_u), b_{r_u}(S_v)\}$\;
            $ends(r_w) \gets \{a_{r_w}(S_v), b_{r_w}(S_w)\}$\;
        }
        \ElseIf{$r_u$ exists but $r_w=null$}{
            $arcs(v) \gets arcs(u)$\;
            $arcs(u) \gets \emptyset$\;
        }
        \ElseIf{$r_w$ exists but $r_u=null$}{
            $arcs(v) \gets arcs(w)$\;
            $arcs(w) \gets \emptyset$\;
        }
    }
    \Return{$v$}\;
\end{algorithm}


For the time analysis, let $T(n)$ be the time that it takes to compute $\Psi(S_v)$ with $S_v=S$.
Using the selection algorithm, we can compute $S_u$ and $S_w$ using $O(n)$ $\prec_L$-comparisons. Each such comparison can be resoved in $O(\log m)$ time by two geodesic distance queries using the GH-data structure.
For computing $r_u$ and $r_w$, we will give an $O(n \log m)$-time algorithm.
By Lemma~\ref{lem:ends_compute}, we can compute $b_{r_u}(S_v)$ and $a_{r_w}(S_v)$ in $O(\log^2 m)$ time.
The rest of the algorithm for $v$ takes $O(n)$ time.
Therefore, we have the following recurrence $T(n) = 2 T(n / 2) + O(n \log m) + O(\log^2 m)$, which solves to $T(n) = O(n \log n \log m + n \log^2 m)$.

In summary, assuming that a GH data structure has been constructed for $P$, $\Psi(S)$ can be built in $O(n \log n \log m + n \log^2 m)$ time.

\subsubsection{Algorithm for computing $\boldsymbol{r_u}$ and $\boldsymbol{r_w}$}
We now present an algorithm to compute $r_u$ and $r_w$ in $O(|S_v|\log m)$ time.
We will do this by scanning $I(S_u)$ and $I(S_w)$ (which are stored in $arcs(u)$ and $arcs(w)$, respectively). Recall that $I(S_u)$ is a list of points ordered by $\prec_L$. Let $first(I(S_u))$ and $last(I(S_u))$ represent the first and last points of $I(S_u)$, respectively. For each point $p\in I(S_u)$, we define $prev(p)$ and $next(p)$ as its previous and next points in $I(S_u)$, respectively. We define the same notations for $I(S_w)$.

Our algorithm starts with $s = start(I(S_u))$ and $t = start(I(S_w))$ and scan the two lists $I(S_u)$ and $I(S_w)$.
For the sake of brevity, we let $a_s=a_s(S_u)$, $b_s=b_s(S_u)$, $a_t=a_t(S_w)$, and $b_t=b_t(S_w)$.
Algorithm~\ref{algo:iuiw} gives the pseudocode.

\begin{algorithm}
    \caption{Compute $r_u$ and $r_w$} \label{algo:iuiw}
    \KwIn{$I(S_u)$ and $I(S_w)$, stored in $arcs(u)$ and $arcs(w)$, respectively.}
    \KwOut{$r_u$ and $r_w$.}
    $s \gets first(I(S_u))$   \tcp*[l]{Let $a_s, b_s$ denote the endpoints of $\xi_s(S_u)$.}
    $t \gets first(I(S_w))$ \tcp*[l]{Let $a_t, b_t$ denote the endpoints of $\xi_t(S_w)$.}
    \While{True}{
        \If{$d(s,b_t)<d(t,b_t)$}
        {
            \If{$t\neq last(I(S_w))$}
            {
                $t \gets next(t)$\label{ln:advancet}\;
            }
            \Else{
                $r_w\gets null$ \;
                $r_u\gets last(I(S_u))$\;
                \Return{$r_u, r_w$}\;
            }
        }
        \Else{
            \If{$d(s,a_s)\leq d(t,a_s)$}
            {
                \If{$s\neq last(I(S_u))$}
                {
                    $s \gets next(s)$\label{ln:10}\;
                }
                \Else
                {
                    $r_u\gets s$ \;
                    $r_w\gets t$\;
                    \Return{$r_u, r_w$}\;
                }
            }
            \Else
            {\label{ln:20}
                \If{$s=first(I(S_u))$}
                {
                    $r_u\gets null$\;
                    $r_w\gets first(I(S_w))$\;
                }
                \Else{\label{ln:30}
                    $r_u\gets prev(s)$\;
                    \While{$d(prev(s),b_t)<d(t,b_t)$\label{ln:while}}
                    {
                        $t\gets next(t)$\;
                    }
                    $r_w\gets t$\;
                }
                \Return{$r_u, r_w$}\;\label{ln:final}
            }
        }
    }
\end{algorithm}


We describe the algorithm and in the meanwhile argue its correctness. The algorithm maintains the {\em invariants} that $I(S_u)_{< s} \subseteq I(S_v)$ and $I(S_w)_{< t} \cap I(S_v) = \emptyset$. Each iteration (the while loop in the pseudocode) either advances one of $s$ and $t$ (with $s \gets next(s)$ and $t \gets next(t)$), or find $r_u$ and $r_w$ and halt the algorithm.

In the beginning of each iteration, we first check if $d(s,b_t)<d(t,b_t)$. We have the following lemma.

\begin{lemma} \label{lem:tnotin}
    If $d(s, b_t) < d(t, b_t)$, then $t \notin I(S_v)$.
\end{lemma}

\begin{proof}
    Assume to the contrary that $t\in I(S_v)$.
    Then there must be some point $p \in \xi_t(S_v)$ such that $d(t, p) < d(s, p)$.
    As $d(s, b_t) < d(t, b_t)$ and $\xi_t(\{t\})$ is continuous,
    there must be some point $q \in \xi_t(\{t\})$ such that $d(s, q) = d(t, q)$ and $p \prec_{\Xi} q \precxi b_t$ where $\precxi$ is taken with respect to $\Xi(\{t\})=\xi_t(\{t\})$; see Figure~\ref{fig:computeuw}.

    \begin{figure}[h]
\begin{minipage}[t]{\linewidth}
\begin{center}
\includegraphics[totalheight=1.4in]{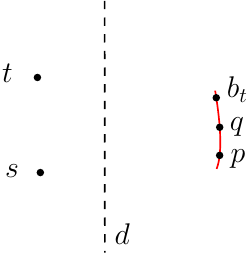}
\caption{Illustrating the proof of Lemma~\ref{lem:tnotin}. The red curve shows a portion of $\Xi(\{t\})$.}
\label{fig:computeuw}
\end{center}
\end{minipage}
\end{figure}


    Let $S' = \{s, t\}$. Note that $s = \beta_q(S')$ and $t = \beta_p(S')$.
    Applying Lemma~\ref{lem:consistentTL} to $S'$
    with $t = \beta_p(S')$ and $s = \beta_{q}(S')$, we obtain $t \prec_L s$ since $p \prec_{\Xi} q$.
    However, because $s \in S_u$ and $t \in S_w$, we have $s \prec_L t$. This is a contradiction.
\end{proof}

If $d(s,b_t)<d(t,b_t)$, then $t\not\in I(S_v)$ by Lemma~\ref{lem:tnotin}. If $t\neq last(I(S_w))$, then we advance $t$ (Line~\ref{ln:advancet}), which preserves the algorithm invariants since $t\not\in I(S_v)$. If $t= last(I(S_w))$, then by our algorithm invariants and due to $t\not\in I(S_v)$, $I(S_w)\cap I(S_v)=\emptyset$ and thus we set $r_w=null$. As $r_w=null$, $I(S_v)=I(S_u)$ and thus we simply set $r_u=last(I(S_u))$, and then terminate the algorithm by returning $r_u$ and $r_w$.

If $d(s,b_t)\geq d(t,b_t)$, we further check if $d(s,a_s)\leq d(t,a_s)$. We have the following lemma.

\begin{lemma} \label{lem:sisin}
    If $d(s, b_t)\geq d(t, b_t)$ and $d(s, a_s) \leq d(t, a_s)$, then $s \in I(S_v)$.
\end{lemma}
\begin{proof}
    Suppose for the sake of contradiction that $s \not\in I(S_v)$.
    Then it must be that $d(S_v, a_s) < d(s, a_s) \leq 1$. Depending on whether $a_s \in \partial P$, there are two cases.

\begin{itemize}
    \item
    If $a_s \in \partial P$, 
    then since $d(S_v, a_s) < 1$, we have $a_s \in \Xi(S_v)$ by definition. Let $t^* = \beta_{a_s}(S_v)$. Hence, $d(t^*,a_s)=d(S_v, a_s) < d(s, a_s)$.
    Because $a_s \in \Xi(S_v)$ and $t^* = \beta_{a_s}(S_v)$, we know that $a_s \in \xi_{t^*}(S_v)$ and thus $t^* \in I(S_v)$.
    Additionally, because $d(t^*, a_s) = d(S_v, a_s) < d(s, a_s) = d(S_u, a_s)$,
    it must be that $t^* \in S_w$.
    Therefore, $t^* \in I(S_v) \cap S_w = I(S_v) \cap I(S_w)$. By the algorithm invariants, $t \prec_L t^*$ (note that $t\neq t^*$ since $d(s, a_s) \leq d(t, a_s)$ while $d(t^*, a_s)<d(s,a_s)$).

    We next derive a contradiction using Lemma~\ref{lem:consistent2v3}.
    Let $S' = \{s, t, t^*\}$. Note that $s\prec_L t\prec_L t^*$. Since both $t$ and $t^*$ are in $S_w$, by the definition of $b_t$, we have $d(t,b_t)\leq d(t^*,b_t)$. As $d(s, b_t)\geq d(t, b_t)$, we obtain that $t = \beta_{b_t}(S')$. Since $b_t\in \Xi(S_w)$ by definition, we have $b_t\in \xi_t(S_w)$ and thus $b_t\in \xi_t(\{t\})$. As $t = \beta_{b_t}(S')$, it follows  that $b_t\in \Xi_t(S')$.  Hence, $t \in I(S')$.
    Because $d(s, a_s) \leq d(t, a_s)$ and $d(s, a_s) \leq 1$,
    Lemma~\ref{lem:consistent2v3} says that $d(s, a_s) < d(t^*, a_s)$.
    However, this contradicts $d(t^*,a_s)< d(s, a_s)$.

\item

    If $a_s \notin \partial P$, then $d(S_u, a_s) = d(s, a_s) = 1$.
    Because $d(S_v, a_s) < d(s, a_s) = 1$, if we extend the last edge of $\pi(d, a_s)$ beyond $a_s$, the extension will hit $\Xi(S_v)$ at a point $p$.
    Let $t^* = \beta_{p}(S_v)$.
    Because $p \in \Xi(S_v)$ and $t^* = \beta_{p}(S_v)$, we know that $p \in \xi_{t^*}(S_v)$ and thus $t^* \in I(S_v)$.
    By Lemma~\ref{lem:ext_dist}, we get that $d(t^*, a_s) < d(t^*, p) \leq 1$.
    Because $d(S_u, a_s) = d(s, a_s) = 1$,
    it must be that $t^* \in S_w$.
    As in the above case, $t^* \in I(S_v) \cap S_w = I(S_v) \cap I(S_w)$. By the algorithm invariants, $t \prec_L t^*$ (note that $t\neq t^*$ since $d(t^*, a_s) < d(t^*, p) \leq 1 = d(s, a_s) \leq d(t, a_s)$).

    We next derive a contradiction using Lemma~\ref{lem:consistent2v3}. Let $S' = \{s, t, t^*\}$. Note that $s\prec_L t\prec_L t^*$. Since both $t$ and $t^*$ are in $S_w$, by the definition of $b_t$, we have $d(t,b_t)\leq d(t^*,b_t)$. As $d(s, b_t)\geq d(t, b_t)$, we obtain that $t = \beta_{b_t}(S')$. As argued above, it follows that $t \in I(S')$. Because $d(s, a_s) \leq d(t, a_s)$ and $d(s, a_s) \leq 1$, Lemma~\ref{lem:consistent2v3} says that $d(s, a_s) < d(t^*, a_s)$.  However, this contradicts $d(t^*, a_s) < d(s, a_s)$.
\end{itemize}

The lemma thus follows.
\end{proof}

By the above lemma, if $d(s,a_s)\leq d(t,a_s)$, then we advance $s$ if $s\neq last(I(S_u))$, which preserves the algorithm invariants as $s\in I(S_v)$. If $s=last(I(S_u))$, then by the algorithm invariants, $I(S_u)\subseteq I(S_v)$. Thus, we set $r_u=s$, which is $last(I(S_u))$. The next lemma proves that $r_w$ is the current $t$, and thus we set $r_w=t$ and terminate the algorithm.

\begin{lemma}
 If $d(s, b_t)\geq d(t, b_t)$, $d(s, a_s) \leq d(t, a_s)$, and $s=last(I(S_u))$, then $t \in I(S_v)$.
\end{lemma}
\begin{proof}
To prove the lemma, it suffices to argue that $d(t,b_t)\leq d(s',b_t)$ for any $s'\in I(S_u)$. This is obviously true if $s$ is the only point of $I(S_u)$ since $d(s, b_t)\geq d(t, b_t)$. We thus assume that $|I(S_u)|>1$. Let $s'$ be any point of $I(S_u)$ other than $s$. As $s=last(I(s_u))$, we have $s'\prec_L s\prec_L t$. We next prove $d(t,b_t)\leq d(s',b_t)$ by applying Lemma~\ref{lem:consistent2v3} to $S'=\{s',s,t\}$.

First of all, we know that $s\in I(S_v)$ by Lemma~\ref{lem:sisin}. Hence, $s\in I(S')$. Since $d(t, b_t)\leq d(s, b_t)$ and  $d(t, b_t)\leq 1$,  Lemma~\ref{lem:sisin} says that $d(t,b_t)< d(s',b_t)$. The lemma thus follows.
\end{proof}

Next, we consider the case $d(s,a_s)> d(t,a_s)$ (Line~\ref{ln:20}). We have the following lemma.

\begin{lemma} \label{lem:snotin}
    If $d(s,a_s)> d(t,a_s)$, then $s \notin I(S_u)$.
\end{lemma}
\begin{proof}
    Symmetric to Lemma~\ref{lem:tnotin}.
\end{proof}

By the above lemma and our algorithm invariants, $r_u=prev(s)$ if $prev(s)$ exists, i.e., if $s\neq first(I(S_u))$. If $s=first(I(S_u))$, then we know that $r_u=null$. In this case, we have $S(I_v)=S(I_w)$ and $r_w=first(I(S_w))$, and thus we can terminate the algorithm.

If $s\neq first(I(S_u))$ (Line~\ref{ln:30}), then we set $r_u=prev(s)$. Next, as long as $d(prev(s),b_t)<d(t,b_t)$, by Lemma~\ref{lem:tnotin} (replacing $s$ with $prev(s)$), $t\not\in I(S_v)$ and thus we advance $t$. We claim that $t$ cannot reach $last(I(S_w))$ with $d(prev(s),b_t)<d(t,b_t)$, since otherwise $I(S_w)\cap I(S_v)=\emptyset$ and thus $I(S_v)=I(S_u)$, contradicting $s\not\in I(S_v)$. Hence, the while loop in Line~\ref{ln:while} will stop at a $t$ with $d(prev(s),b_t)\geq d(t,b_t)$. The following lemma proves that $t\in I(S_v)$, implying that this $t$ is $r_w$ due to the algorithm invariants, and thus we can set $r_w=t$ and terminate the algorithm (Line~\ref{ln:final}).

\begin{lemma} \label{lem:tisin}
If $d(s, b_t)\geq d(t, b_t)$, $d(s,a_s)>d(t,a_s)$, and $d(prev(s), b_t)\geq d(t, b_t)$, then
$t \in I(S_v)$.
\end{lemma}
\begin{proof}
    Suppose for the sake of contradiction that $t \not\in I(S_v)$.
    Then it must be that $d(S_v, b_t) < d(t, b_t) \leq 1$.
    Depending on whether $b_t \in \partial P$, there are two cases.

    \begin{itemize}
    \item
    If $b_t \in \partial P$, then since $d(S_v, b_t) < 1$, we have $b_t \in \Xi(S_v)$ by definition.
    Let $s^* = \beta_{b_t}(S_v)$. Hence, $d(s^*,b_t)=d(S_v, b_t) < d(t, b_t)$.
    Because $b_t \in \Xi(S_v)$ and $s^* = \beta_{b_t}(S_v)$, we know that $b_t \in \xi_{s^*}(S_v)$ and thus $s^* \in I(S_v)$.
    Additionally, because $d(s^*, b_t) = d(S_v, b_t) < d(t, b_t) = d(S_w, b_t)$,
    it must be that $s^* \in S_u$.
    Therefore, $s^* \in I(S_v) \cap S_u = I(S_v) \cap I(S_u)$.
    Because $r_u = prev(s)$, we have $s^* \prec_L s$.
     Let $s' = prev(s)$.
    Because $d(s^*, b_t) < d(t, b_t)$ and $d(s', b_t)\geq d(t, b_t)$, we have $s^* \neq s'$.
    Hence, it follows that $s^* \prec_L s' \prec_L t$.

    We next derive contradiction by applying Lemma~\ref{lem:consistent2v3} to $S' = \{s^*, s', t\}$.
    Indeed, as $s' \in I(S_v)$, $s'$ must also be in $I(S')$.
    Since $d(s', b_t)\geq d(t, b_t)$ and $d(t, b_t) \leq 1$, applying Lemma~\ref{lem:consistent2v3} leads to
    $d(t, b_t) < d(s^*, b_t)$. But this contradicts $d(s^*,b_t)< d(t, b_t)$.

\item
    If $b_t \notin \partial P$, then $d(S_w, b_t) = d(t, b_t) = 1$.
    Because $d(S_v, b_t) < d(t, b_t) = 1$, if we extend the last edge of $\pi(d, b_t)$ beyond $b_t$, we will hit $\Xi(S_v)$ at point $p$.
    Let $s^* = \beta_{p}(S_v)$.
    Because $p \in \Xi(S_v)$ and $s^* = \beta_{p}(S_v)$, we know that $p \in \xi_{s^*}(S_v)$ and thus $s^* \in I(S_v)$.
    By Lemma~\ref{lem:ext_dist}, we get that $d(s^*, b_t) < d(s^*, p) \leq 1$.
    Because $d(S_w, b_t) = d(t, b_t) = 1$, it must be that $s^* \in S_u$.
    Hence, we have that $s^* \in I(S_v) \cap S_u = I(S_v) \cap I(S_u)$.
    Because $r_u = prev(s)$, we have $s^* \prec_L s$.
    Let $s' = prev(s)$.
    Because $d(s^*, b_t) < d(s^*, p) \leq 1 = d(t, b_t)$ and $d(s', b_t)\geq d(t, b_t)$, we have $s^* \neq s'$.
    Hence, it follows that $s^* \prec_L s' \prec_L t$.

    We will derive contradiction by applying Lemma~\ref{lem:consistent2v3} to $S' = \{s^*, s', t\}$.
    Indeed, as $s' \in I(S_v)$, $s'$ must also be in $I(S')$.
    Since $d(s', b_t)\geq d(t, b_t)$ and $d(t, b_t) \leq 1$, applying Lemma~\ref{lem:consistent2v3} leads to
    $d(t, b_t) < d(s^*, b_t)$. But this contradicts $d(s^*, b_t) < d(t, b_t)$.
    \end{itemize}

The lemma thus follows.
\end{proof}

\paragraph{The time analysis.} The algorithm has $O(|S_v|)$ iterations and each iteration either advances one of $s$ and $t$, or computes $r_u$ and $r_w$. Each iteration takes $O(\log m)$ time for computing the involved geodesic distances using the GH data structure. Hence, the total time of the algorithm is $O(|S_v|\log m)$.

\subsection{Deletions}
\label{sec:deletion}

To delete a point $s \in S$, we need to delete the leaf of $\Psi(S)$ storing $s$.
We also need to update the fields of certain nodes of $\Psi(S)$ to reflect this change.
Notice that the nodes of $\Psi(S)$ that are affected are those in path from the root to the leaf containing $s$.
Starting from the root, suppose we are at a node $v$ with left child $u$ and right child $w$, and $s \in S_w$ (the case $s \in S_u$ is symmetric). We do the following at $v$.

We first call the descend algorithm at $v$ (Algorithm~\ref{algo:descend}), which gives us $\Xi(S_u)$ and $\Xi(S_w)$.
Then, we recursively delete $s$ from $S_w$ (i.e., call the deletion algorithm at $w$). If $s\neq r_w$, then $r_u$ and $r_w$ do not change after deleting $s$ from $S_v$, so we have $I(S_v \setminus \{s\}) = I(S_u)_{\leq r_u} + I(S_w \setminus \{s\})_{\geq r_w}$,
$b_{r_u}(S_v \setminus \{s\}) = b_{r_u}(S_v)$, and $a_{r_w}(S_v \setminus \{s\}) = a_{r_w}(S_v)$, i.e., $Y_2(v)$ does not change.
This means that even though it may be that $I(S_w \setminus \{s\}) \neq I(S_w)$,
we can still ascend as in Algorithm~\ref{algo:ascend} without issue.

The difficulty is in handling the case where $s = r_w$. In this case, we must compute the new $r_u$ and $r_w$. If we do so using Algorithm~\ref{algo:iuiw} without any modification, it would take $O(|S_v| \log m)$ time, making the overall deletion cost $\Omega(n)$ time. Therefore, we must do something smarter.

Let $r_u'$ and $r_w'$ refer to the new $r_u$ and $r_w$ after deleting $s$,
and let $r_u$ and $r_w$ refer to the old ones before the deletion.
Recall that the two invariants of Algorithm~\ref{algo:iuiw} are,
for current points $s \in I(S_u)$ and $t \in I(S_w)$ in the scanning process,
that $I(S_u)_{< s} \subseteq I(S_v)$ and $I(S_w)_{< t} \cap I(S_v) = \emptyset$.
Notice that $I(S_u)_{< r_u} \subseteq I(S_v \setminus \{s\})$.
This means that we can safely start the scan with $s = r_u$ rather than $first(I(S_u))$.
If $r_w \neq first(I(S_w))$, let $t_0 = prev(r_w)$ in $I(S_w)$.
Because $t_0 \in I(S_w)$ and $t_0 \notin I(S_v)$,
it must be that $t_0 \notin I(S_u \cup \{t_0\})$,
and thus we have $t_0 \notin I(S_v \setminus \{s\})$.
This means that we can safely start the scan from $t = t_0$ because
$I(S_w \setminus \{s\})_{< t_0} \cap I(S_v \setminus \{s\}) = \emptyset$.
If $r_w = first(I(S_w))$, then we can start with $t_0 = next(r_w)$.

Let $k_u$ be the number of points of $I(S_u)$ between $r_u$ and $r_u'$
and $k_v$ the number of points of $I(S_w)$ between $t_0$ and $r_w'$.
Then by starting Algorithm~\ref{algo:iuiw} with $s = r_u$ and $t = t_0$,
we can compute $r_u'$ and $r_w'$ in $O((1 + k_u + k_w) \log m)$ time.
Notice that the points between $r_u$ and $r_u'$ (except $r_u$ itself) were moved up from $u$ to $v$ (i.e., there were in $I(S_u)$ but not in $I(S_v)$, but they are now in $I(S_v)$).
Because future deletions cannot move a point down and the height of $\Psi(S)$ is $O(\log n)$,
each point can only be moved up $O(\log n)$ times.
Therefore, the sum of $k_u$ across all deletions is $O(n \log n)$.
Similarly, the points between $t_0$ and $r_w'$ (except $t_0$ and possibly $r_w'$) were moved up to $w$ from descendants of $w$.
By the same reasoning, the sum of $k_w$ across all deletions is $O(n \log n)$ as well.
Therefore, the computation of $r_u'$ and $r_w'$ across all levels of $\Psi(S)$ takes amortized $O(\log n \log m)$ time per deletion.

For the rest of the time analysis,
notice that the remaining costly operation is computing
$b_{r_u'}(S_v \setminus \{s\})$ and $a_{r_w'}(S_v \setminus \{s\})$ for each node $v$ in which $s \in X(v)$.
By Lemma~\ref{lem:ends_compute}, this takes $O(\log^2 m)$ time per node.
Because the height of $\Psi(S)$ is $O(\log n)$,
the time spent on them throughout the whole deletion is $O(\log n \log^2 m)$.

In summary, each deletion takes amortized $O(\log n \log m + \log n \log^2 m) = O(\log n \log^2 m)$ time.
Algorithm~\ref{algo:delete} gives the psuedocode.

\begin{algorithm}
    \caption{Delete($v,s$)} \label{algo:delete}
    \KwIn{$\Psi(S_v)$ and point $s \in S_v$ to be deleted}
    \KwOut{$\Psi(S_v \setminus \{s\})$}
    Descend from $v$ (Algorithm~\ref{algo:descend})\;
    \If{$v$ is a leaf}{
        Mark $v$ as deleted\;
    }
    \If{$v$ has only one child $v'$}{
        Delete$(v', s)$\;
        \If{$v'$ is marked as deleted}{
            Mark $v$ as deleted\;
        }
    }
    \ElseIf{$v$ has left child $u$ and right child $w$}{
        \If{$s \in S_u$}{
            Delete$(u, s)$\;
            \If{$u$ is marked as deleted}{
                Remove $u$ from being a child\;
            }
        }
        \Else{
            Delete$(w, s)$\;
            \If{$w$ is marked as deleted}{
                Remove $w$ from being a child\;
            }
        }
    }

    \If{$v$ has left child $u$, right child $w$, and $s \in X(v)$}{
        Compute $r_u'$ and $r_w'$\;
        $X(v) \gets \{r_u', r_w'\}$\;
        \If{$r_u'$ and $r_w'$ exist}{
            Compute $b_{r_u'}(S_v \setminus \{s\})$ and $a_{r_w'}(S_v \setminus \{s\})$, and add them to $Y_2(v)$\;
        }
    }
    Ascend to $v$ (Algorithm~\ref{algo:ascend})\;
\end{algorithm}

\section{Implicit geodesic additively-weighted Voronoi diagram}
\label{sec:ls_gawvd}

In this section, we prove Theorem~\ref{theo:nearneighbor}. Let $S$ be a set of $n$ weighted points in $P$. Our goal is to construct a data structure for $S$ to answer (additively-weighted geodesic) nearest neighbor queries: Given a query point $q$, find its nearest neighbor in $S$.

The problem can be reduced to a {\em diagonal-separated} problem by using the balanced polygon decompositions (BPD) of simple polygons as described in Section~\ref{sec:pre}.


\paragraph{Diagonal-separated nearest neighbor queries.}
Consider a diagonal $d$ that partitions $P$ into two subpolygons $P_L$ and $P_R$. In the diagonal-separated problem, we assume that all points of $S$ are in $P_L$ and every query point $q$ is in $P_R$. We will prove the following result.

\begin{lemma}\label{lem:diasepneighborquery}
Let $P$ be a simple polygon of $m$ vertices, and assume that a GH data structure has been constructed for $P$. Given a diagonal $d$ that partitions $P$ into two subpolygons $P_L$ and $P_R$, and a set $S$ of $n$ weighted points in $P_L$, we can construct a weighted nearest neighbor data structure for $S$ in $O(n\log^2 m+n\log n\log m)$ time such that for any query point in $P_R$, its nearest neighbor in $S$ can be computed in $O(\log n\log m)$ time.
\end{lemma}

\paragraph{Proving Theorem~\ref{theo:nearneighbor}.}
Before proving Lemma~\ref{lem:diasepneighborquery}, we first prove Theorem~\ref{theo:nearneighbor} using the lemma.

In the preprocessing, we compute the BPD-tree $T_P$ for $P$ in $O(m)$ time~\cite{ref:GuibasOp89}. In addition, we build a point location data structure on the triangulation of $P$ formed by the triangles of the leaves of $T_P$ in $O(m)$ time so that given a point $q\in P$, the triangle that contains $q$ can be computed in $O(\log m)$ time~\cite{ref:KirkpatrickOp83,ref:EdelsbrunnerOp86}. Finally, we  construct the GH data structure for $P$ in $O(m)$ time. This finishes our preprocessing, which takes $O(m)$ time.

Given a set $S$ of $n$ weighted points in $P$, for each node $v\in T_P$, let $S_v$ be the set of sites of $S$ in $P_v$. We compute $S_v$ for all nodes $v\in T_P$ in $O(n\log n + n\log m)$ time, as follows.
First, using the point location data structure, for each point $p\in S$, we find the leaf $v_p$ of $T_P$ whose triangle contains $p$ in $O(\log m)$ time. Then, for each node $v$ in the path from $v_p$ to the root of $T_P$, we add $p$ to $S_v$. As the height of $T_P$ is $O(\log m)$, the total time is $O(n\log m)$ for computing $S_v$ for all $v\in T_P$.

In the above algorithm, we also keep a list $L$ of nodes $v$ with $S_v\neq \emptyset$. As each point of $S$ can appear in $S_v$  for $O(\log m)$ nodes $v\in T_P$, we have $L=O(n\log m)$ and $\sum_{v\in L}|S_v|=O(n\log m)$. The reason we use $L$ is to avoid spending $\Omega(m)$ time on constructing the data structure for $S$.

For each node $v\in L$, if $v$ is a leaf, then $P_v$ is a triangle. In this case, we construct a weighted Euclidean Voronoi diagram $\vd_v$ for $S_v$ inside $P_v$, which can be done in $O(|S_v|\log |S_v|)$ time~\cite{ref:FortuneA87}. We further build a point location data structure for $\vd_v$ in $O(|S_v|)$ time~\cite{ref:KirkpatrickOp83,ref:EdelsbrunnerOp86} since $|\vd_v|=O(|S_v|)$~\cite{ref:FortuneA87,ref:SharirIn85}. Doing this for all leaves of $T_P$ takes $O(n\log n)$ time in total.

If $v$ is not a leaf, then let $u$ and $w$ be the left and right children of $v$, respectively. We construct a nearest neighbor data structure of Lemma~\ref{lem:diasepneighborquery} for $S_u$ in $P_u$ with respect to the diagonal $d_v$ for query points in $P_w$, so the points of $S_u$ and query points are in two different sides of $d_v$; let $\calD_u(w)$ denote the data structure. This takes $O(|S_u|\log^2 m+|S_u|\log n\log m)$ time by Lemma~\ref{lem:diasepneighborquery}.
Symmetrically, we construct a nearest neighbor data structure $\calD_w(u)$ of Lemma~\ref{lem:diasepneighborquery} for $S_w$ in $P_w$ with respect to the diagonal $d_v$ for query points in $P_u$, which takes $O(|S_w|\log^2 m+|S_w|\log n\log m)$ time.
Doing the above for all internal nodes $v$ of $T_P$ takes $O(n\log^3m + n\log n\log^2 m)$ time as $\sum_{v\in L}|S_v|=O(n\log m)$.

This finishes the construction of our nearest neighbor data structure for $S$, which takes $O(n\log^3 m+n\log n\log^2 m)$ time in total.

Given a query point $q\in P$, our goal is to find a nearest neighbor of $q$ in $S$, denoted by $\beta_q(S)$. We first find the leaf $v_q$ of $T_P$ whose triangle contains $q$. This takes $O(\log m)$ time using the point location data structure for the leaf triangles of $T_P$. Then, using the Voronoi diagram $\vd_{v_q}$, we find the nearest neighbor of $q$ in $S_{v_q}$ in $O(\log n)$ time and consider it a candidate for $\beta_q(S)$. Next, for each ancestor $v$ of $v_p$ in $T_P$, we do the following. Let $u$ and $w$ be the left and right children of $v$, respectively. Note that either $u$ or $w$ is an ancestor of $v_q$. We assume that $w$ is an ancestor of $v_q$ and the other case can be handled similarly. Using the data structure $\calD_u(w)$, we find the nearest neighbor of $q$ in $S_u$, which can be done in $O(\log n\log m)$ time by Lemma~\ref{lem:diasepneighborquery} and we consider it a candidate for $\beta_q(S)$. As the height of $T_P$ is $O(\log m)$, the above finds $O(\log m)$ candidates in $O(\log n\log^2 m)$ time. Finally, among all candidates, we find the one nearest to $q$ and return it as $\beta_q(S)$. To this end, for each candidate $p$, we need to compute the geodesic distance $d(p,q)$, which can be done in $O(\log m)$ time using the GH data structure.

In summary, a nearest neighbor of $q$ in $S$ can be found in a total of $O(\log n\log^2 m)$ time. This proves Theorem~\ref{theo:nearneighbor}.

\paragraph{Proving Lemma~\ref{lem:diasepneighborquery}.}
In the remainder of this section, we prove Lemma~\ref{lem:diasepneighborquery}.

Agarwal, Arge, and Staals~\cite{ref:AgarwalIm18} developed a related structure for the {\em unweighted} diagonal-separated case (their construction time is $O(n\log n + n\log^2 m)$, and their query procedure is essentially identical to ours). Their approach relies on the algorithm for Hamiltonian abstract Voronoi diagrams~\cite{ref:KleinHa94}.
A natural idea is to attempt to extend their method to the {\em weighted} setting. However, several obstacles arise. For example, a key property used in their algorithm is that the bisector of two sites intersects the diagonal $d$ at most once. This holds in the unweighted case because bisectors are (geodesically) “line-like,” but it fails in the weighted setting: even in the Euclidean plane, the bisector of two additively weighted sites is a hyperbola, which may intersect a line twice.
A second crucial property used in~\cite{ref:AgarwalIm18} is that, for any three sites $s,t,r$, the bisector of $s$ and $t$ intersects the bisector of $t$ and $r$ at most once. This again is true in the unweighted case but false for additively weighted sites, where two weighted bisectors may intersect twice. These difficulties make it challenging to adapt their Hamiltonian-diagram–based techniques.
For this reason, we develop a different approach.



In the following, we first discuss several properties and operations involving bisectors of weighted points in $P$ (Section~\ref{sec:bis_data_struct}). They will be essential for our algorithm of Lemma~\ref{lem:diasepneighborquery}.

\subsection{Bisectors -- properties and operations}
\label{sec:bis_data_struct}

Let $S$ be a set of weighted points in $P$ (we sometimes refer to these points as {\em sites}). For each site $s\in S$, let $w(s)$ denote the weight of $s$, and define $d_s(p)=w(s)+d(s,p)$ as the {\em weighted distance} from $s$ to $p$ for any point $p\in P$.

For two sites $s,t\in S$, {\em their bisector} $B(s,t)$ is defined as the set of points $p\in P$ such that $d_s(p)=d_t(p)$, i.e., $B(s,t)$ consists of all points of $P$ equidistant from $s$ and $t$. Note that if $d_s(t) < w(t)$ or $d_t(s) < w(s)$, then $B(s,t) = \emptyset$.


We present several properties and algorithms regarding bisectors of weighted points of $S$.
Similar results for the bisectors of unweighted points have been studied in the literature, e.g., \cite{ref:AronovOn89,ref:AronovTh93,ref:AgarwalIm18,ref:WangAn23DCG}. We will try to extend them to the weighted setting. Some of these extensions can be done by following the unweighted methods in a straightforward manner, while others present more challenging and new approaches are required.

As usually in the unweighted case~\cite{ref:AronovOn89,ref:AronovTh93,ref:AgarwalIm18,ref:WangAn23DCG}, we make a general position assumption that no vertex of $P$ is equidistant (with respect to the weighted distance) from two sites of $S$. Under the assumption, the bisector of any two sites of $S$ does not contain a vertex of $P$.



Suppose $p_t$ is the anchor of $t$ in the shortest path $\pi(s,t)$ for two points $s,t\in P$. Recall that $\pi(s,t)$ is oriented from $s$ to $t$.
In our discussion, we often need to extend the last edge of $\pi(s,t)$ (i.e., $\overline{p_tt}$) beyond $t$ along the direction from $p_t$ to $t$ until a point $p\in \partial P$, we refer to $\overline{tp}$ as the {\em last-edge extension} of $\pi(s,t)$. Similarly, we define the {\em first-edge extension} of $\pi(s,t)$ by extending the first edge of the path beyond $s$.

We next introduce some properties regarding bisectors. We start with the following lemma, whose proof is a straightforward extension of the unweighted case~\cite[Lemma 3.22]{ref:AronovOn89}.

\begin{lemma} {\em (Aronov \cite[Lemma 3.22]{ref:AronovOn89})}\label{lem:struct2}
    If $B(s,t)$ is not empty, then the followings hold.
    \begin{enumerate}
    \item $B(s,t)$ is a smooth curve connecting two points in $\partial P$ and does not intersect $\partial P$ in its interior.
    \item $B(s,t)$ is the union of $O(m)$ hyperbolic arcs (where a line segment is considered a special hyperbolic arc), and each endpoint of these arcs is the intersection of $B(s,t)$ with the last-edge extension of $\pi(s,p)$ or $\pi(t,p)$ for a vertex $p$ of $P$.
    \item The tangent to $B(s,t)$ at any point $p$ bisects the angle formed by the two segments connecting $p$ and its anchors in $\pi(s,p)$ and $\pi(t,p)$, respectively.
    \end{enumerate}
\end{lemma}
\begin{proof}
We use exactly the same proof as \cite[Lemma 3.22]{ref:AronovOn89} but replacing the geodesic distance $d(s,p)$ (resp., $d(t,p)$) for any $p\in P$ used in the proof by $d_s(p)$ (resp., $d_t(p)$), e.g., replacing $d(s,x)$ by $d_s(x)$ and replacing $d(s,\hat{s})$ by $d_s(\hat{s})$ in the proof. We omit the details.
\end{proof}


By a straightforward extension of \cite[Lemma 3.28]{ref:AronovOn89}, we have the following observation.


\begin{observation}\label{lem:star}
For all $p \in P$, $\pi(s, p)\cap B(s,t)$ consists of at most one point.
\end{observation}


By the above observation, $B(s,t)$ intersects $\pi(s,t)$ at exactly one point, which partitions $B(s,t)$ into two {\em half-bisectors}.

Consider a diagonal $d$ that partitions $P$ into two subpolygons $P_L$ and $P_R$. Without loss of generality, we assume that $d$ is vertical and $P_L$ lies locally on the left of $d$. We have the following lemma.

\begin{lemma}\label{lem:xmonotone}
For two weighted sites $s,t\in S\cap P_L$, for each half-bisector of $B(s,t)$, its portion inside $P_R$ is $x$-monotone, and its  intersection with $d$ consists of at most one point. Hence, $B(s,t)$ intersects $d$ at most twice.
\end{lemma}
\begin{proof}
Consider a half-bisector $B_1$ of $B(s,t)$ and let $B^*_1$ be its portion inside $P_R$. If $B^*_1=\emptyset$, then the lemma is obviously true. In the following, we assume $B^*_1\neq \emptyset$. To prove the lemma, it suffices to show that $B^*_1$ is $x$-monotone, since it immediately leads to that $B^*_1$ intersects $d$ at most once.

Assume to the contrary that $B^*_1$ is not $x$-monotone. Then, $B^*_1$ has a point $p$ whose tangent line $\ell_p$ to $B^*_1$ is vertical. If we move along $\ell_p$ from $p$ upwards (resp., downwards), let $p_1$ (resp., $p_2$) be the first point of $\partial P$ encountered by the traversal.

\begin{figure}[t]
\begin{minipage}[t]{\linewidth}
\begin{center}
\includegraphics[totalheight=1.3in]{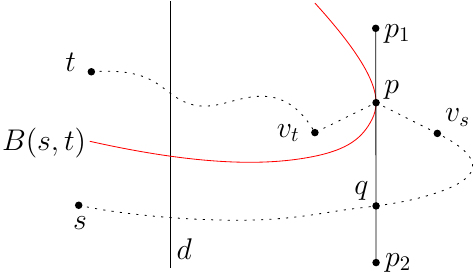}
\caption{Illustrating the proof of Lemma~\ref{lem:xmonotone}. The two dotted curves represent shortest paths $\pi(s,p)$ and $\pi(t,p)$.}
\label{fig:monotone}
\end{center}
\end{minipage}
\vspace*{-0.15in}
\end{figure}

Let $v_s$ and $v_t$ be the anchors of $p$ in $\pi(s,p)$ and $\pi(t,p)$, respectively. By Lemma~\ref{lem:struct2}(c), $\ell_p$ bisects the angle $\angle v_spv_t$. Note that $\angle v_spv_t$ cannot be zero degree since otherwise we would have $v_s=v_t$, and thus $d_s(v_s)=d_t(v_t)$, contradicting our general position assumption. Therefore, one of $v_s$ and $v_t$ must be strictly to the right of $\ell_p$, and without loss of generality we assume it is $v_s$ (see Figure~\ref{fig:monotone}).
This means that $\pi(s,v_s)$ must cross $\overline{p_1p_2}$ at a point $q$. As both $p$ and $q$ are on $\overline{p_1p_2}\subseteq P$, we have $\pi(p,q)=\overline{pq}$. On other other hand, since both $p$ and $q$ are in $\pi(s,p)$, the portion $\pi(t,p)$ of $\pi(s,p)$ between $p$ and $q$ is the shortest path $\pi(p,q)$, which is apparently not $\overline{pq}$ as it contains $v_s$ that is strictly to the right of $\ell_p$. We thus obtain a contradiction.
\end{proof}


The following lemma will be used frequently in our algorithm.

\begin{lemma}\label{lem:algointerdiagonal}
Assuming that a GH data structure has been computed for $P$, given any two sites $s,t\in S$, the intersection of $d$ with each half-bisector of $B(s,t)$ can be computed in $O(\log^2 m)$ time.
\end{lemma}
\begin{proof}
Let $v_s$ and $v_t$ be the anchors of $s$ and $t$ in $\pi(s,t)$, respectively. Let $t'$ denote the first point of $\partial P$ hit by extending the segment $\overline{v_tt}$ along the direction from $v_t$ to $t$. We define $s'$ similarly. It is easy to see that $\pi(s',t')=\overline{s's}\cup \pi(s,t)\cup \overline{tt'}$~\cite{ref:AronovOn89}. See Figure~\ref{fig:diagonalintersect}.

Note that for any point $p\in \overline{s's}$, $d(p,t)=d(p,s)+d(s,t)$. Since $B(s,t)\neq \emptyset$, we can obtain that $d_s(p)<d_t(p)$. This implies that $B(s,t)$ cannot intersect any point of $\overline{ss'}$. Similarly, $B(s,t)$ cannot intersect any point of $\overline{tt'}$. As $B(s,t)$ intersects $\pi(s,t)$ at exactly one point, $B(s,t)$ intersects $\pi(s',t')$ at exactly one point.

The path $\pi(s',t')$ partitions $P$ into two subpolygons, one bounded by $\pi(s',t')$ and the portion of $\partial P$ from $s'$ to $t'$ counterclockwise and the other bounded by $\pi(s',t')$ and the portion of $\partial P$ from $s'$ to $t'$ clockwise. Let $P_1$ and $P_2$ denote the first and second subpolygons, respectively (e.g., in Figure~\ref{fig:diagonalintersect}, $P_1$ is the gray portion). Since $B(s,t)$ intersects $\pi(s',t')$ at exactly one point, $P_1$ and $P_2$ each contains exactly one half-bisector of $B(s,t)$.

As both $s$ and $t$ are in $P_L$, $\pi(s,t)\subseteq P_L$. Since $\pi(s',t')$ cannot intersect $d$ twice (otherwise the path could be shortened along $d$), at most one of $\overline{ss'}$ and $\overline{tt'}$ can cross $d$.

If neither $\overline{ss'}$ nor $\overline{tt'}$ crosses $d$, then one of $P_1$ and $P_2$ must be in $P_L$ (e.g., in  Figure~\ref{fig:diagonalintersect}, $P_1\subseteq P_L$), and thus only one half-bisector of $B(s,t)$ can intersect $d$. Otherwise, without loss of generality, we assume that
$\overline{tt'}$ crosses $d$ at a point $z$. Then, $z$ divides $d$ into two subsegments that belong to $P_1$ and $P_2$, respectively. By Lemma~\ref{lem:xmonotone}, $B(s,t)$ intersects each subsegment of $d$ at most once.

\begin{figure}[t]
\begin{minipage}[t]{\linewidth}
\begin{center}
\includegraphics[totalheight=1.5in]{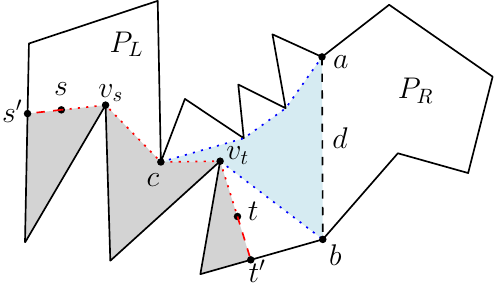}
\caption{The gray region is $P_1$. The light blue region is the funnel $F_s(d)$.}
\label{fig:diagonalintersect}
\end{center}
\end{minipage}
\vspace*{-0.15in}
\end{figure}

With the above discussion, we present an algorithm for the lemma as follows.

Using the GH data structure, we can compute the two anchors $v_s$ and $v_t$ in $O(\log m)$ time. Then, using the GH data structure, we determine in $O(\log m)$ time whether $d$ intersects $\overline{ss'}$ or $\overline{tt'}$ (and if so, compute the intersection) without explicitly computing $s'$ and $t'$ as follows\footnote{Note that we could construct a ray-shooting data structure for $P$ in the preprocessing so that $s'$ and $t'$  can be computed in $O(\log m)$ time using ray-shooting queries~\cite{ref:HershbergerA95,ref:ChazelleRa94}. However, we want to restrict our preprocessing work to the GH data structure only as this is required in the statement of Lemma~\ref{lem:diasepneighborquery}.}. Let $a$ and $b$ be the upper and lower endpoints of $d$, respectively.

We only discuss how to determine whether $d$ intersects $\overline{tt'}$ (and if so, compute the intersection), since the other case can be handled similarly.
Let $\overrightarrow{r_t}$ be the ray out of $t$ in the direction opposite $v_t$.
Let $c_t$ be the junction vertex of $\pi(t,a)$ and $\pi(t,b)$. Let $v_a$ be the anchor of $t$ in $\pi(a, t)$ and $v_b$ the anchor of $t$ in $\pi(b, t)$. Notice that if $t \neq c_t$, then $t$ is not ``visible'' to $d$ and thus $\overline{tt'}$ cannot intersect $d$.
Also note that $t \neq c_t$ if and only if $v_a = v_b$.
We can compute $c_t$, $v_a$, and $v_b$ in $O(\log m)$ time using the GH data structure.
If $v_a = v_b$, we know that $\overline{tt'}$ does not intersect $d$. Otherwise, observe that $\overline{tt'}$ intersects $d$ if and only if the ray $\overrightarrow{r_t}$ lies in the clockwise angle from
$\overrightarrow{t v_a}$ to $\overrightarrow{t v_b}$.
If  $\overrightarrow{r_t}$ lies in the angle, then the intersection $d\cap \overline{tt'}$ is
the intersection of $d$ with $\overrightarrow{r_t}$, which can be computed in $O(1)$ time. In summary, we can determine whether $d$ intersects $\overline{tt'}$ (and if so, compute the intersection) in $O(\log m)$ time.


Depending on whether $d$ intersects $\overline{ss'}\cup \overline{tt'}$, we proceed with two cases.

\paragraph{The non-intersection case.}
If neither $\overline{ss'}$ nor $\overline{tt'}$ intersects $d$, then there is at most one intersection between $B(s,t)$ and $d$, and we compute it in $O(\log^2 m)$ time as follows.

Let $c$ be the junction vertex of $\pi(s,a)$ and $\pi(s,b)$. Consider the funnel $F_s(d)$ formed by $\pi(c,a)$, $\pi(c,b)$, and $d$. Note that each side of $F_s(d)$ is a concave chain. Let $\gamma_s$ be the concatenation of the two sides of $F_s(d)$. Hence, $a$ and $b$ are the two endpoints of $\gamma_s$.
For any point $z\in d$, there is a single vertex $p_z\in \gamma_s$ such that $\overline{p_zz}$ is tangent to the chain of $\gamma_s$ containing $p_z$ and $p_z$ is the anchor of $z$ in $\pi(s,z)$; see Figure~\ref{fig:diagonalintersect40}. If we move $z$ along $d$ from $a$ to $b$, the point $p_z$ also moves from $a$ to $b$ along $\gamma_s$.

\begin{figure}[t]
\begin{minipage}[t]{\linewidth}
\begin{center}
\includegraphics[totalheight=1.5in]{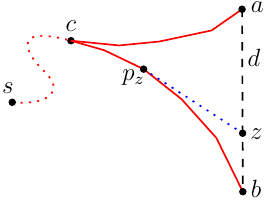}
\caption{Illustrating $p_z$. The red solid curves are the two sides of the funnel $F_s(d)$ while the red dotted curve connecting $s$ and $t$ is $\pi(s,c)$.}
\label{fig:diagonalintersect40}
\end{center}
\end{minipage}
\end{figure}

Let $z^*$ be the intersection of $B(s,t)$ and $d$. Our goal is to compute $z^*$. We first show how to compute its anchor $p_{z^*}$ in $\pi(s,z^*)$ in $O(\log^2 m)$ time. A similar algorithm was already given in \cite{ref:AgarwalIm18} for the unweighted case. Hence we briefly discuss the main idea below for completeness.

Due to our general position assumption that a bisector of two sites of $S$ cannot intersect a vertex of $P$, $z^*$ must be in the interior of $d$. Hence, either $d_s(a)<d_t(a)$ or $d_s(a)>d_t(a)$, which can be determined in $O(\log m)$ time using the GH data structure. Without loss of generality, we assume that $d_s(a)<d_t(a)$. Hence, $d_s(b)>d_t(b)$. This means that for all points $p\in \overline{az^*}$, $d_s(a)\leq d_t(a)$, and for all $p\in \overline{bz^*}$, $d_s(a)\geq d_t(a)$. We will use this property to perform binary search on the vertices of $\gamma_s$ to compute $p_{z^*}$.

Using the GH data structure, we can obtain in $O(\log m)$ time a tree structure representing $\pi(a,c)$ and a tree structure representing $\pi(c,b)$, so that we can do binary search on them.
In each iteration of the binary search, we have a vertex $p$, say in $\pi(c,a)$, and we also have $d(s,p)$ and thus $d_{s}(p)$ available from the GH data structure. Let $p'$ be the adjacent vertex of $p$ in $\gamma_s$ closer to $a$ (see Figure~\ref{fig:diagonalintersect50}); note that $p'$ can be obtained from the GH data structure. We extend the segment $\overline{pp'}$ from $p$ to $p'$ until a point $z\in d$. Then, we know that $d_s(z)=d_s(p)+|\overline{pz}|$. As $d_s(p)$ is already available, we can have $d_s(z)$ immediately. We compute $d(t,p)$ in $O(\log m)$ time using the GH data structure and thereby obtain $d_t(p)$. We compare $d_s(p)$ with $d_t(p)$. If they are equal, then $z$ is $z^*$ and we can terminate the entire algorithm. Otherwise, if
$d_s(z)<d_t(z)$, then according to the above discussion, $z^*$ is below $z$ and thus we continue to search $\gamma_s$ on the side  toward $b$ after $p$. If $d_s(z)>d_t(z)$, we search $\gamma_s$ on the other side of $p$. In this way, $p_{z^*}$ can be computed within $O(\log m)$ iterations. As each iteration takes $O(\log m)$ time, the total time for computing $p_{z^*}$ is $O(\log^2 m)$.

\begin{figure}[t]
\begin{minipage}[t]{\linewidth}
\begin{center}
\includegraphics[totalheight=1.5in]{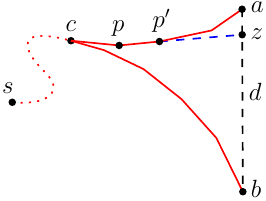}
\caption{The dashed blue segment is the extension of $\overline{pp'}$.}
\label{fig:diagonalintersect50}
\end{center}
\end{minipage}
\end{figure}

Using a similar method, we can compute the anchor $p_{z^*}'$ of $z^*$ in $\pi(t,z^*)$ in $O(\log^2m)$ time. With the two anchors $p_{z^*}$ and $p_{z^*}'$, $z^*$ can be computed in additional $O(1)$ time.

\paragraph{The intersection case.}
If one of $\overline{ss'}$ and $\overline{tt'}$ intersects $d$, then let $z$ be the intersection. As discussed above, $B(s,t)$ intersects $\overline{az}$ at most once. We can compute the intersection in $O(\log^2 m)$ time using the same algorithm as above but with respect to $\overline{az}$ (instead of with respect to $d$ as above). Similarly, $B(s,t)$ intersects $\overline{bz}$ at most once and we can compute the intersection in $O(\log^2 m)$ time using the same algorithm as above but with respect to $\overline{bz}$.
\end{proof}

\subsection{Overview}
We now in a position to prove Lemma~\ref{lem:diasepneighborquery}.
We follow the notation in the lemma, e.g., $P$, $S$, $d$, $P_L$, $P_R$, $n$, $m$. Without loss of generality, we assume that $d$ is vertical and $P_R$ is locally to the right of $d$. Let $a$ and $b$ be the upper and lower vertices of $d$, respectively.

For any $S'\subseteq S$, we let $\vd(S')$ denote the (weighted) geodesic Voronoi diagram of $S'$ in $P$. For any subset $P'\subset P$, let $\vd(S',P')$ denote the portion of $\vd(S')$ in $P'$. In particular, $\vd(S',P_R)$ is the portion of $\vd(S')$ in $P_R$ and $\vd(S',d)$ is the intersection of $d$ with $\vd(S')$. If a point $p$ is in a Voronoi region in $\vd(S')$ of a site $s\in S'$, we say that $p$ is {\em claimed} by $s$.

The following observation from the literature \cite{ref:AronovOn89,ref:HershbergerAn99} will be used later.
\begin{observation}{\em \cite{ref:AronovOn89,ref:HershbergerAn99}} \label{obser:star}
For any $S'\subseteq S$, the Voronoi region of each site $s\in S'$ in $\vd(S')$ is star-shaped, i.e., for any point $p$ in the Voronoi region, $\pi(s,p)$ is in the interior of the Voronoi region except that $p$ may be on the boundary of the region.
\end{observation}

As our query points are in $P_R$, we are interested in $\vd(S,P_R)$. Since the size of $\vd(S,P_R)$ is $\Omega(m)$ as $|P_R|=\Omega(m)$, we will construct $\vd(S,P_R)$ implicitly so that give a query point $q$ we can still efficiently locate site whose Voronoi region contains $q$ (and the site is the nearest neighbor of $q$ in $S$).

In the following, we first compute $\vd(S,d)$ in $O(n\log n\log m+n\log^2 m)$ time in Section~\ref{sec:vddnew} and then extend $\vd(S,d)$ to compute $\vd(S,P_R)$ implicitly in additional $O(n\log n+n\log^2 m)$ time in Section~\ref{sec:vdextend}. We show in Section~\ref{sec:point_location} that we can construct a point location data structure in additional $O(n)$ time so that a nearest neighbor query for any point $q\in P_R$ can be answered in $O(\log n\log m)$ time.
Note that all these algorithms are based on the assumption that a GH data structure for $P$ is already available.

\subsection{Computing $\boldsymbol{\vd(S,d)}$}
\label{sec:vddnew}

In the unweighted case, Agarwal, Arge, and Staals~\cite{ref:AgarwalIm18} gave an algorithm to compute $\vd(S, d)$ in $O(n \log n + n \log^2 m)$ time. Their algorithm relies on the fact that each Voronoi region in $\vd(S,d)$ is a single segment of $d$. However, this does not hold anymore in our weighted problem because a Voronoi region of $\vd(S,d)$ may have $\Omega(n)$ segments. Consequently, it appears difficult to extend their approach. Instead, we propose a new algorithm of runtime is $O(n \log n \log m + n \log^2 m)$.

We compute $\vd(S,d)$ incrementally. First, we sort the sites of $S$ by their weighted distance to the lower vertex $b$ of $d$.
This can be done in $O(n \log n + n \log m)$ time by using the GH data structure to compute the weighted distances and then sorting.
Let $s_1, s_2, \ldots, s_n$ be the sites of $S$ in increasing order by their weighted distances to $b$. Note that due to our general position assumption, no two sites have the same weighted distance to $b$. Hence, the above sorted order is unique.
Define $S_i = \{s_j \in S : j \leq i\}$.

Note that for any $S'\subseteq S$, $\vd(S',d)$ a set of points (these points are considered vertices of $\vd(S',d)$) partitioning $d$ into segments such that each segment is claimed by one site of $S'$ (i.e., it belongs to a Voronoi region of $\vd(S',d)$).

We use a balanced binary search tree $T$ to store the vertices of $\vd(S_i, d)$ in the order along $d$ from $b$ to $a$.
Initially, for $\vd(S_0, d)$, $T$ contains only $a$ and $b$.
Assuming that $\vd(S_{i-1},d)$ is already computed and stored in $T$, we compute $\vd(S_{i},d)$ by updating $T$, as follows.

First of all, although a Voronoi region in $\vd(S,d)$ may consist of multiple segments in general, we have the following lemma,
which makes the task of computing $\vd(S_{i},d)$ much easier.


\begin{lemma} \label{lem:one_seg}
    The Voronoi region of $s_i$ in $\vd(S_i, d)$ is either empty or a single segment of $d$.
\end{lemma}
\begin{proof}
Suppose for the sake of contradiction that this is not true. Then, $d$ must have two segments $e_1$ and $e_2$ claimed by $s_i$ in $\vd(S_i, d)$ such that between $e_1$ and $e_2$ there is a segment $e$ claimed by some site $s_j \in S_{i - 1}$.
Note that such $e$ and $s_j$ must exist or else $e_1$ and $e_2$ would not be disconnected segments.

Without loss of generality, we assume that $e_1$ is below $e_2$.
Since $d_{s_j}(b)<d_{s_i}(b)$, the bisector $B(s_i,s_j)$ must intersect $d$ at least three times: an odd number of times between $b$ and $e_1$, an odd number of times between $e_1$ and $e$, and an odd number of times between $e$ and $e_2$.
However, $B(s_i,s_j)$ can intersect $d$ at most twice by Lemma~\ref{lem:xmonotone}, a contradiction.
\end{proof}

Let $e$ be the Voronoi region of $s_i$ in $\vd(S_i, d)$. By the above lemma, either $e$ is $\emptyset$ or a single segment of $d$.
We first show how to compute $\vd(S_i,d)$ by assuming $e\neq\emptyset$. Let $v_1$ and $v_2$ be the lower and upper endpoints of $e$, respectively. Once $v_1$ and $v_2$ are known, to obtain $\vd(S_i, d)$, we can simply remove all vertices of $T$ between $v_1$ and $v_2$, and then insert $v_1$ and $v_2$ to $T$. This can be done in $O(\log |T|)$ time by updating $T$. Note that this also proves that $\vd(S_{i},d)$ has at most $2i$ vertices. Hence, $|T|\leq 2n$ throughout the algorithm.

We next show how to compute $v_1$ and $v_2$. We will only describe how to compute $v_1$ since computing $v_2$ can be done analogously. Let $v_0$ and $v_0'$ be the vertices of $T$ immediately below and above $v_1$, respectively.
Then, $\overline{v_0v_0'}$ is claimed by a site $s_j\in S_{i-1}$ in $\vd(S_{i-1}, d)$ and thus $v_1$ is an intersection between $B(s_i,s_j)$ and $d$. By Lemma~\ref{lem:xmonotone}, $B(s_i,s_j)$ intersects $d$ at most twice.
If $B(s_i,s_j)$ intersects $d$ exactly once, then the intersection is $v_1$. Otherwise, if only one intersection of $B(s_i,s_j)\cap d$ is on $\overline{v_0v_0'}$, then that intersection is $v_1$. If both intersections are on $\overline{v_0v_0'}$, then the lower one is $v_1$. Since we can compute $B(s_i,s_j)\cap d$ in $O(\log^2 m)$ time by Lemma~\ref{lem:algointerdiagonal}, once we know $v_0$ and $v_0'$, $v_1$ can be computed in $O(\log^2 m)$ time. In what follows, we discuss how to compute $v_0$, after which $v_0'$ can be easily obtained from $T$ as it is the vertex of $T$ immediately above $v_0$.

We will compute $v_0$ by searching in $T$. To this end, we have the following lemma.

\begin{lemma} \label{lem:vdd_bs}
For any vertex $v$ of $\vd(S_{i - 1}, d)$, if $e\neq\emptyset$, then we can determine whether $v$ is below or above $v_0$ in $O(\log m)$ time. In addition, if the algorithm returns $-1$, then $e=\emptyset$.
\end{lemma}
\begin{proof}
Let $s_j$ be a site of $\vd(S_{i-1},d)$ claiming $v$ (note that $s_j$ can be maintained along with $T$).
Using the GH data structure, we compute $d_{s_i}(v)$ and $d_{s_j}(v)$ in $O(\log m)$ time. If $d_{s_i}(v) \leq d_{s_j}(v)$, then $v \in e$. As $v_0$ is the lower endpoint of $e$, we know that $v$ is above $v_0$ (including the case $v=v_0$) and thus we are done in this case.

If $d_{s_i}(v) > d_{s_j}(v)$, then  $v \notin e$. Using the GH data structure, we compute $d_{s_i}(a)$ and $d_{s_j}(a)$.
If $d_{s_i}(a) < d_{s_j}(a)$, then we claim that $v$ is below $e$ and thus is below $v_0$. Indeed, because $d_{s_j}(b) < d_{s_i}(b)$,
the bisector $B(s_i, s_j)$ must intersect $d$ an even number of times between $b$ and $v$ and an odd number of times between $v$ and $a$.
As $B(s_i, s_j)$ can intersect $d$ at most twice by Lemma~\ref{lem:xmonotone}, it must intersect $d$ zero times between $b$ and $v$ and once between $v$ and $a$. Hence, $B(s_i, s_j)$ intersects $d$ exactly once. Since both $v$ and $b$ are closer to $s_j$ than to $s_i$ and $b$ is the lower vertex of $d$, $v$ must lie below $B(s_i, s_j)\cap d$ while $e$ must be above $B(s_i, s_j)\cap d$. The claim thus follows.


The remaining case is when $d_{s_j}(a) < d_{s_i}(a)$ (recall that $d_{s_j}(a) \neq d_{s_i}(a)$ due to our general position assumption).
In this case, following the above argument, $B(s_i,s_j)$ must intersect $d$ either zero times or twice.
    We start by checking if $d_{s_j}(s_i) \leq d_{s_i}(s_i)$ (note that $d_{s_i}(s_i)$ is just the weight of $s_i$).
    If so, then we know that the Voronoi region of $s_i$ in $\vd(S)$ is empty and thus $e$ is also empty. In this case we return $-1$. In the following, we assume $d_{s_j}(s_i) >  d_{s_i}(s_i)$.

If we extend the last edge of $\pi(s_j,s_i)$ beyond $s_i$, let $s_i'$ be the first point on $\partial P$ hit by the extension. As discussed in the proof of Lemma~\ref{lem:algointerdiagonal}, every point of $\overline{s_is_i'}$ is closer to $s_i$ than to $s_j$. According to the proof of Lemma~\ref{lem:algointerdiagonal}, we can do the following in $O(\log m)$ time: determine if $\overline{s_is_i'}$ intersects $d$ and if so compute the intersection. Depending on whether $\overline{s_is_i'}$ intersects $d$, there are two cases.

\begin{itemize}
    \item
If $\overline{s_is_i'}$ intersects $d$, let $z$ be the intersection. Let $e'$ be the set of all points $p\in d$ with $d_{s_i}(p)\leq d_{s_j}(p)$. Clearly, $e\subseteq e'$, $z\in e'$, and $v\not\in e'$. Since $\overline{s_is_i'}$ intersects $d$, according to the analysis of Lemma~\ref{lem:algointerdiagonal}, $e'$ is a single segment of $d$. Hence, if $z$ lies below (resp., above) $v$, then we know that $e'$ and thus $e$ must lie below (resp., above) $v$.

\item
If $\overline{s_is_i'}$ does not intersect $d$, then we claim that $e=\emptyset$ and thus we return $-1$. We now prove the claim. If we extend the last edge of $\pi(s_i,s_j)$ beyond $s_j$, let $s_j'$ the first point of $\partial P$ hit by the extension.

If $\overline{s_js_j'}$ does not intersect $d$, then according to the proof of Lemma~\ref{lem:algointerdiagonal}, $B(s_i,s_j)$ intersects $d$ at most once. As we know that $B(s_i,s_j)$ intersect $d$ either zero times or two times, we obtain that $B(s_i,s_j)$ does not intersect $d$. Since $d_{s_j}(b) < d_{s_i}(b)$, it follows that every point of $d$ is closer to $s_j$ than to $s_i$. Hence, $e=\emptyset$.

If $\overline{s_js_j'}$ intersects $d$, then let $z'$ be their intersection. If $B(s_i,s_j)$ does not intersect $d$, then as above we still have $e=\emptyset$. Otherwise, $B(s_i,s_j)$ intersects $d$ twice. Then, according to the proof of Lemma~\ref{lem:algointerdiagonal}, $z'$ is closer to $s_j$ than to $s_i$ and $z'$ is between the two intersections of $B(s_i,s_j)$ and $d$. Hence, $B(s_i,s_j)$ intersect $\overline{bz'}$ exactly once. As $z'$ is closer to $s_j$ than to $s_i$, i.e., $d_{s_j}(z')<d_{s_i}(z')$, we obtain that $d_{s_j}(b)>d_{s_i}(b)$. But this contradicts with the fact that $d_{s_j}(b)<d_{s_i}(b)$. Hence, the case in which $B(s_i,s_j)$ intersects $d$ twice cannot happen. The claim thus follows.
\end{itemize}

In summary, if $e$ exists, then we can determine whether $v$ is below or above $v_0$ in $O(\log m)$ time, and if the algorithm return $-1$, then $e=\emptyset$. The lemma thus follows.
\end{proof}

By the above lemma, $v_0$ can be computed in $O(\log m\log n)$ time using $T$ since the height of $T$ is $O(\log n)$. Therefore, if $e$ exists, then it can be computed in $O(\log n\log m)$ time. Consequently, we can compute $\vd(S_i,d)$ by updating $T$ in $O(\log n\log m)$ time.

The above computes $\vd(S_i,d)$ assuming that $e$ exists. Otherwise, $\vd(S_i,d)=\vd(S_{i-1},d)$ and thus we do not need to make any changes to $T$. It remains to discuss how to determine if $e$ exists. To this end, we still run the above algorithm. If the procedure of Lemma~\ref{lem:vdd_bs} returns $-1$, then we know that $e$ does not exist. Otherwise, the algorithm will eventually computes a point $v_1$. Let $s_j\in S_{i-1}$ be the site that claims $v_1$ in $\vd(S_{i-1},d)$. Note that $s_j$ is the site claiming $\overline{v_0v_0'}$. To have $s_j$, we can associate each Voronoi region of $\vd(S_{i-1},d)$ with its site in $T$. As such, once $\overline{v_0v_0'}$ is computed, $s_j$ can be obtained as well. With $s_j$, observe that $e=\emptyset$ if and only if $d_{s_j}(v_1) \leq d_{s_i}(v_1)$. Hence, we can determine whether $e=\emptyset$ in additional $O(\log m)$ time by computing $d_{s_j}(v_1)$ and $d_{s_i}(v_1)$ using the GH data structure.

In summary, we can compute $\vd(S_i,d)$ from $\vd(S_{i-1},d)$ in $O(\log n\log m+\log^2 m)$ time. Therefore, $\vd(S,d)$ can be constructed in $O(n\log n\log m+n\log^2 m)$ time.

\subsection{Computing $\boldsymbol{\vd(S,P_R)}$ implicitly}
\label{sec:vdextend}

Given $\vd(S,d)$, we now present an algorithm to compute $\vd(S,P_R)$ by extending $\vd(S,d)$ into $P_R$. We will compute $\vd(S,P_R)$ in the following implicit manner.

For each vertex of $\vd(S,P_R)$, if it is not a vertex of $P_R$, then
it is either a Voronoi vertex (which is equidistant to three sites of $S$) or an intersection of an Voronoi edge of $\vd(S,P_R)$ with the boundary of $P_R$, and we will compute its location explicitly. For each Voronoi edge of $\vd(S,P_R)$, it is a portion of the bisector $B(s,t)$ of two sites $s$ and $t$. Since its size may be $\Omega(m)$, instead of computing it explicitly, we will only store its two vertices and the two sites $(s,t)$. In this way, we essentially will compute $\vd(S,P_R)$ as a graph such that each edge is represent by the two sites that define it; for differentiation, we call it the {\em abstract graph} of $\vd(S,P_R)$ and use $\avd(S,P_R)$ to denote it. Note that such an abstract graph has been considered in the literature for unweighted case, e.g.,\cite{ref:OhVo20,ref:AgarwalIm18}. Note also that since $\vd(S)$ has $O(n)$ Voronoi vertices and Voronoi edges~\cite{ref:AronovOn89,ref:HershbergerAn99}, $\vd(S,P_R)$ also has $O(n)$ Voronoi vertices and Voronoi edges, and thus $\avd(S,P_R)$ has $O(n)$ vertices and edges.

Given $\vd(S,d)$, we let $d_1,d_2,\ldots,d_h$ denote the segments of $d$ partitioned by the vertices of $\vd(S,d)$ from top to bottom, with $h=O(n)$. For each $d_i$, let $s_i$ denote the site of $S$ claiming $d_i$ in $\vd(S,d)$. Note that $s_i$ and $s_j$ for $i\neq j$ may refer to the same physical site of $S$, but we consider them two occurrences of the site. Indeed, each $s_i$ is associated with its corresponding $d_i$.
For differentiation, we let $S_d$ refer to the multiset $\{s_1,s_2,\ldots,s_h\}$.

To compute $\vd(S,P_R)$, we will use the continuous Dijkstra's approach~\cite{ref:HershbergerAn99}. Imagine that each site $s_i$ generates a wavelet and expands the wavelet through $d_i$ into $P_R$. During the propagation, we maintain a wavefront. The wavefront at ``time'' $t$ consists of all points of $P_R$ whose weighted distance from its nearest neighbor in $S_d$ is equal to $t$. We propagate the wavefront in the temporal order. Some sites in the wavefront will be eliminated when two adjacent bisectors intersect.
Our algorithm is similar to that in \cite[Section 5.2]{ref:HershbergerAn99} in the high level.

For any two sites $s_i$ and $s_j$ of $S_d$, each of their half-bisectors can intersect $d$ at most once by Lemma~\ref{lem:xmonotone}. However, we claim that only one half-bisector can intersect the segment of $d$ between $d_i$ and $d_j$. To see this, if we consider the Voronoi diagram $\vd(\{s_i,s_j\},d)$, then $d_i$ must be in the Voronoi region of $s_i$ and $d_j$ must be in the Voronoi region of $s_j$. Hence, $B(s_i,s_j)$ has exactly one intersection with $d$ between $d_i$ and $d_j$. We let $B_1(s_i,s_j)$ denote the half-bisector of $B(s_i,s_j)$ intersecting $d$ between $d_i$ and $d_j$. We use $B_{ij}$ to denote the portion of $B_1(s_i,s_j)$ inside $P_R$. By Lemmas~\ref{lem:struct2} and \ref{lem:xmonotone}, $B_{ij}$ is an $x$-monotone curve connecting a point on $d$ to a point on $\partial P_R$ and its interior lies in the interior of $P_R$. This also implies that each Voronoi edge of $\vd(S,P_R)$ is a portion of $B_{ij}$ for two sites $s_i,s_j\in S$. Thus each Voronoi edge is $x$-monotone and its two vertices can be naturally defined as {\em left vertex} and {\em right vertex}, respectively.


For any three sites $s_i,s_j,s_k$ with $i<j<k$, while in general the two bisectors $B(s_i,s_j)$ and $B(s_j,s_k)$ may intersect twice, we have the observation that $B_{ij}$ and $B_{jk}$ can only intersect once, as proved in the following lemma.

\begin{observation}\label{obser:bisectorintersect}
For any three sites $s_i,s_j,s_k\in S_d$ with $i<j<k$, $B_{ij}$ and $B_{jk}$ can intersect at most once.
\end{observation}
\begin{proof}
Assume to the contrary that $B_{ij}$ and $B_{jk}$ intersect at two points $p_1$ and $p_2$. Then, we have $d_{s_i}(p_1)=d_{s_j}(p_1)=d_{s_k}(p_1)$. If we consider the Voronoi diagram $\vd(\{s_i,s_j,s_k\})$, $p_1$ must be a vertex of the diagram and thus $\pi(s_i,p_1)$ must be in the Voronoi region of $s_i$ since the Voronoi region is star-shaped by Observation~\ref{obser:star}. Similarly, $\pi(s_j,p_1)$ must be in the Voronoi region of $s_j$ and $\pi(s_k,p_1)$ must be in the Voronoi region of $s_k$.
Since $i<j<k$ and $p_1\in B_{ij}\cap B_{jk}$, by the definitions of $B_{ij}$ and $B_{jk}$, the intersection $\pi(s_j,p_1)\cap d$ must be above $\pi(s_i,p_1)\cap d$ and below $\pi(s_k,p_1)\cap d$; see Figure~\ref{fig:twointersection}.

\begin{figure}[t]
\begin{minipage}[t]{\linewidth}
\begin{center}
\includegraphics[totalheight=1.8in]{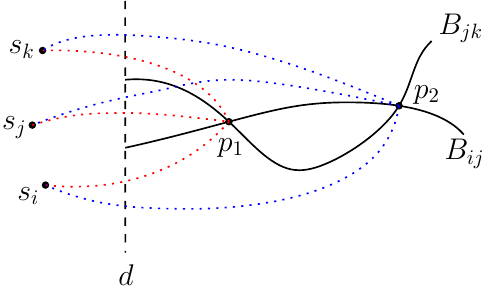}
\caption{The two solid curves are $B_{ij}$ and $B_{jk}$. The red (resp., blue) dotted curves are shortest paths from $p_1$ (resp., $p_2$). One of the red paths must cross one of the blue paths.}
\label{fig:twointersection}
\end{center}
\end{minipage}
\end{figure}

The above also applies to $p_2$, i.e., $p_2$ is a vertex of  $\vd(\{s_i,s_j,s_k\})$, $\pi(s_i,p_2)$ is in the Voronoi region of $s_i$ , $\pi(s_j,p_2)$ is in the Voronoi region of $s_j$, $\pi(s_k,p_2)$ is in the Voronoi region of $s_k$, and the intersection $\pi(s_j,p_2)\cap d$ must be above $\pi(s_i,p_2)\cap d$ and below $\pi(s_k,p_2)\cap d$. Then, one of the three paths $\pi(s_i,p_1)$, $\pi(s_j,p_1)$, and $\pi(s_k,p_1)$ must cross one of the three paths $\pi(s_i,p_2)$, $\pi(s_j,p_2)$, and $\pi(s_k,p_2)$; see Figure~\ref{fig:twointersection}. But this is not possible as Voronoi regions are star-shaped.
\end{proof}

We maintain the wavefront by a list $L$ of the generators of its wavelets ordered along the wavefront.
Initially, $L=\{s_1,s_2,\ldots,s_k\}$ with each $s_i$ associated with $d_i$. For each pair $(s_i,s_{i+1})$ of adjacent sites in $L$, we let $a_i$ be the common endpoint of $d_{i}$ and $d_{i+1}$. We create a vertex for $\avd(S,P_R)$ at $a_i$. We maintain an event queue $Q$. Initially, we compute $Q$ as follows.
For each triple of adjacent sites $(s_{i-1},s_i,s_{i+1})$ in $L$, we compute the intersection $B_{i-1,i}\cap B_{i,i+1}$ (if it exists) and add it to $Q$ as an {\em event} point with its ``key'' set to the weighted geodesic distance from it to $s_i$. Lemma~\ref{lem:bisectorintersection} below gives an algorithm to implement this {\em bisector-intersection operation} in $O(\log^2 m)$ time.  In addition, for each pair $(s_i,s_{i+1})$ of adjacent sites, we create an edge for $\avd(S,P_R)$ with its left vertex set to $a_i$.

During the algorithm, as long as $Q$ is not empty, we take an event $p\in Q$ with the smallest key value. Assume that $p=B_{ij}\cap B_{jk}$ for three sites $s_i,s_j,s_k\in L$ with $i<j<k$. We process the event as follows.
We create a vertex at $p$ for $\avd(S,P_R)$, and make it the right vertex of the edge of $\avd(S,P_R)$ defined by $(s_i,s_j)$ (resp.,  $(s_j,s_k)$). In addition, we create a new edge for $\avd(S,P_R)$ defined by $(p_i,p_k)$ and set $p$ as its left vertex.
Next, we delete $s_j$ from $L$ and also delete its associated events from $Q$. Furthermore, for the right neighbor $s_{k'}$ of $s_k$ in $L$, $(s_i,s_k,s_{k'})$ becomes a new triple of adjacent sites. We compute the intersection $B_{i,k}\cap B_{k,k'}$ (if it exists) and add it to $Q$ as an event point with its key set to the weighted geodesic distance from it to $s_i$. We do the same for the left neighbor of $s_i$.

Once $Q$ becomes empty, if $L\neq \emptyset$, then no triple of adjacent sites has a bisector intersection in $P_R$, meaning that for every pair of adjacent sites $(s_i,s_j)$ of $L$, the intersection $B_{ij}\cap P_R$ is a vertex of $\avd(S,P_R)$. Hence, we compute $B_{ij}\cap P_R$, which can be done in $O(\log^2 m)$ time by Lemma~\ref{lem:bisectorintersection}, and make it the right vertex of the edge of $\avd(S,P_R)$ defined by the pair $(s_i,s_j)$. This completes the algorithm.

\begin{lemma}\label{lem:bisectorintersection}
Suppose that a GH data structure for $P$ has been constructed. Then, for any two sites $s_i,s_j\in S_d$, the intersection $B_{ij}\cap \partial P_R$ can be computed in $O(\log^2 m)$ time. For any three sites $s_i,s_j,s_k\in S_d$ with $i<j<k$, in $O(\log^2 m)$ time we can compute the bisector intersection $B_{ij}\cap B_{jk}$ or report that it does not exist.
\end{lemma}
\begin{proof}
We first present an algorithm to compute the intersection $B_{ij}\cap \partial P_R$, denoted by $p^*$. Without loss of generality, we assume that $i<j$.

By definition, the bisector $B(s_i,s_j)$ intersects $d$ at least once (because $B_{ij}$ intersects $d$) but at most twice by Lemma~\ref{lem:xmonotone}.
If we extend the last edge of $\pi(s_i,s_j)$ beyond $s_j$, let $s_j'$ be the first point of $\partial P$ hit by the extension; define $s_i'$ similarly. Depending on whether $\overline{s_is_i'}$ or $\overline{s_js_j'}$ crosses $d$, there are two cases. According to the proof of Lemma~\ref{lem:algointerdiagonal}, in $O(\log m)$ time we can determine whether one of them crosses $d$ (and if so compute the intersection).

\paragraph{Neither $\boldsymbol{\overline{s_is_i'}}$ nor $\boldsymbol{\overline{s_js_j'}}$ crosses $\boldsymbol{d}$.}
If neither $\overline{s_is_i'}$ nor $\overline{s_js_j'}$ crosses $d$, then according to the proof of Lemma~\ref{lem:algointerdiagonal},   $B(s_i,s_j)$ intersects $d$ exactly once, which is $B_{ij}\cap d$. In this case, $B(s_i,s_j)$ intersects $\partial P_R$ also once, and the intersection is $p^*$. Hence, $B_{ij}$ partitions $P_R$ into two portions, one portion consisting of all points $p$ with $d_{s_i}(p)\leq d_{s_j}(p)$ and the other consisting of all points $p$ with $d_{s_i}(p)> d_{s_j}(p)$. Using this property, we can do binary search on the vertices of $\partial P_R$ to find the edge $e^*$ containing $p^*$.

Specifically, let $a$ and $b$ be the upper and lower endpoints of $d$, respectively. We orient $\partial P_R$ from $a$ to $b$ clockwise around $P_R$. In each iteration of the binary search, we consider a vertex $p\in \partial P_R$. If $d_{s_i}(p)\leq d_{s_j}(p)$, then $e^*$ is after $p$; otherwise, $e^*$ is before $p$. Computing $d_{s_i}(p)$ and $d_{s_j}(p)$ can be done in $O(\log m)$ time using the GH data structure.
As there are $O(\log m)$ iterations in the binary search, the total time for computing $e^*$ is $O(\log^2 m)$.

After $e^*$ is found, $p^*$ can be computed on $e^*$ in additional $O(\log^2 m)$ time by an algorithm similar to that for computing the intersection of $B(s,t)$ and $d$ in the non-intersection case of the proof of Lemma~\ref{lem:algointerdiagonal}, i.e., using the funnel $F_s(e^*)$.

\paragraph{One of $\boldsymbol{\overline{s_is_i'}}$ and $\boldsymbol{\overline{s_js_j'}}$ crosses $\boldsymbol{d}$.}
If one of $\overline{s_is_i'}$ and $\overline{s_js_j'}$ crosses $d$, then it is possible that $B(s_i,s_j)$ intersects $d$ twice. In this case we can still find $e^*$ by binary search on the vertices of $\partial P_R$. But the situation is more subtle because in this case $B(s_i,s_j)$ may intersect $\partial P_R$ twice and partition $P_R$ into three portions (two portions consisting of all points closer to one of $s_i$ and $s_j$ and the third portion consisting of all points closer to the other site). We use the following strategy.


Without loss of generality, we assume that $\overline{s_is_i'}$ crosses $d$. Then, according to the proof of Lemma~\ref{lem:algointerdiagonal}, if $B(s_i,s_j)$ has two intersections with $\partial P_R$, then $s_j'$ is between them on $\partial P_R$. Furthermore, since $i<j$ and $\overline{s_is_i'}$ crosses $d$, $p^*$ is before $s_i'$; see Figure.~\ref{fig:binarysearch}. Hence, we can still compute $p^*$ by doing binary search on the portion of $\partial P_R$ between $a$ and $s_i'$ if the location of $s_i'$ is known. However, since we do not know the location of $s_i'$,\footnote{Note that we could construct a ray-shooting data structure for $P$ in the preprocessing and then $s_i'$  can be computed in $O(\log m)$ time using ray-shooting queries~\cite{ref:HershbergerA95,ref:ChazelleRa94}. However, we want to restrict our preprocessing work to the GH data structure only as this is required in the statement of Lemma~\ref{lem:diasepneighborquery}.} we use the following approach instead.

\begin{figure}[t]
\begin{minipage}[t]{\linewidth}
\begin{center}
\includegraphics[totalheight=1.8in]{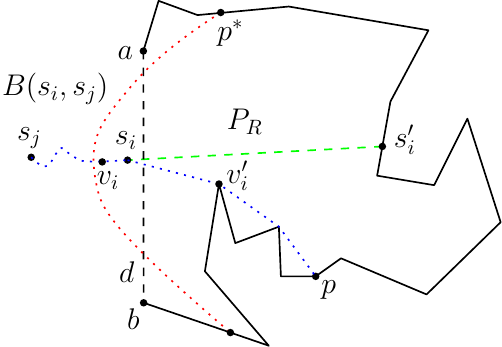}
\caption{The red curve is the bisector $B(s_i,s_j)$. The shortest paths $\pi(s_j,s_i)$ and $\pi(s_i,q)$ are shown with blue curves. The green dashed segment is the extension of $\overline{v_is_i}$.}
\label{fig:binarysearch}
\end{center}
\end{minipage}
\end{figure}

We do binary search on the vertices of $\partial P_R$. In each iteration, we are given a vertex $p\in \partial P_R$ and  we wish to know whether $p$ is before $p^*$ or after. We compute the anchor $v_i$ of $s_i$ in $\pi(s_j,s_i)$ and the anchor $v_i'$ of $s_i$ in $\pi(s_i,p)$; see Figure.~\ref{fig:binarysearch}. Both anchors can be computed in $O(\log m)$ time using the GH data structure. Observe that the path $v_i\rightarrow s_i\rightarrow v_i'$ is a right turn if and only if $p$ is after $s_i'$ on $\partial P_R$ (recall that $\partial P_R$ is oriented from $a$ to $b$). If $p$ is after $s_i'$, then it is also after $p^*$, so we search the left of $\partial P_R$. Otherwise, $p$ is before $s_i'$ and we need to further determine whether $p$ is before or after $p^*$. This can be done by comparing $d_{s_i}(p)$ with $d_{s_j}(p)$ as in the above case. As such, each iteration of the binary search can be done in $O(\log m)$ time. Therefore, we can compute $p^*$ in a total of $O(\log^2 m)$ time.

\paragraph{Computing the intersection $\boldsymbol{B_{ij}\cap B_{jk}}$.}
We now turn to the problem of computing the intersection $B_{ij}\cap B_{jk}$ for three sites $s_i,s_j$, and $s_k$ with $i<j<k$; let $q^*$ denote the intersection point. Note that if $q^*$ exists, then it is the common intersection of $B_{ij}$, $B_{jk}$, and $B_{ik}$.

Let $q=B_{ik}\cap \partial P_R$. We can compute $q$ in $O(\log^2 m)$ time as discussed above. Let $z$ denote the endpoint of $B_{ik}$ on $d$, which can be computed in $O(\log^2 m)$ time by Lemma~\ref{lem:algointerdiagonal}. Recall that $B_{ik}$ is $x$-monotone. Since $B_{ij}\cap B_{jk}$ consists of at most one point by Observation~\ref{obser:bisectorintersect}, for any point $p\in B_{ik}$ with $p\neq q^*$, if $p$ is left of $q^*$, then $d_{s_j}(p)< d_{s_i}(p)$, and if $p$ is right of $q$, then $d_{s_j}(p)> d_{s_i}(p)$. Hence, $q^*$ exists if and only if $d_{s_j}(q)\geq d_{s_i}(q)$. As such, we can determine whether $q^*$ exists in $O(\log^2 m)$ time. In the following, we assume that $q^*$ exists and give an $O(\log^2 m)$ time algorithm to compute it.
We first check whether $d_{s_i}(q)=d_{s_j}(q)$. If so, then $q^*=q$ and we are done. Below, we assume $q\neq q^*$.

Let $c_i$ be the junction vertex of $\pi(s_i,q)$ and $\pi(s_i,z)$; see Figure~\ref{fig:bisectorintersect}. Notice that each of the two paths $\pi(c_i,z)$ and $\pi(c_i,q)$ is a concave chain, and for each point $p\in B_{ik}$, its anchor $p'$ must be in one of the two chains and $\overline{pp'}$ is tangent to the chain containing $p'$. Let $\gamma_i$ be the concatenation of the two chains. Let $q_i$ be the anchor of $q^*$ in $\pi(s_i,q^*)$. Hence, $q_i\in \gamma_i$ and $\overline{q^*q_i}$ is tangent to the chain of $\gamma_i$ containing $q_i$.

Similarly, let $c_k$ be the junction vertex of $\pi(s_k,q)$ and $\pi(s_k,z)$. Define $\gamma_k$ and $q_k$ similarly. As above, $q_k\in \gamma_k$ and $\overline{q^*q_k}$ is tangent to the chain of $\gamma_k$ containing $q_k$; see Figure~\ref{fig:bisectorintersect}.

By Lemma~\ref{lem:struct2}, $B_{ik}$ is partitioned into {\em arcs} by {\em breakpoints} such that each arc is a portion of a hyperbola and all points $p$ of the same arc have the same anchor $p_i$ in $\pi(s_i,p)$ and the same anchor $p_k$ in $\pi(s_k,p)$, i.e., the hyperbola is {\em generated} by the two vertices $p_i\in \gamma_i$ and $p_k\in \gamma_k$, and each breakpoint is the intersection of $B_{ik}$ with the extension segment of the last edge of $\pi(s_i,p_i)$ for some vertex $p_i\in \gamma_i$ or $\pi(s_k,p_k)$ for some vertex $p_k\in \gamma_k$ (we say that the breakpoint is {\em defined} by $p_i$ or $p_k$). To compute $q^*$, we will compute the arc $\alpha^*$ of $B_{ik}$ containing $q^*$, after which $q^*$ can be obtained in $O(1)$ time. To compute $\alpha^*$, it suffices to find the two vertices $p_i^*\in \gamma_i$ and $p_k^*\in \gamma_k$ that generate $\alpha^*$. Note that $p_i^*$ is actually $q_i$ and $p_k^*$ is $q_k$. In what follows, we show how to compute $q_i$ and $q_k$ in $O(\log^2 m)$ time.

\begin{figure}[t]
\begin{minipage}[t]{\linewidth}
\begin{center}
\includegraphics[totalheight=1.8in]{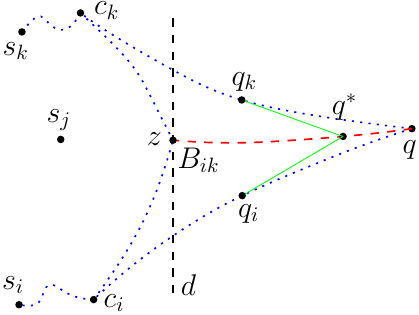}
\caption{The red dashed curve is the bisector $B_{ik}$. The blue curves are shortest paths.}
\label{fig:bisectorintersect}
\end{center}
\end{minipage}
\end{figure}

We orient $\gamma_i$ from $z$ to $c_i$ along $\pi(z,c_i)$ and then to $q$ along $\pi(c_i,q)$. For each point $p\in \gamma_i$, we define its {\em tangent ray} as a ray from $p$ and toward $B_{ik}$ such that its supporting line is tangent to the chain of $\gamma_i$ containing $p$.
We parameterize $\gamma_i$ over $[0, 1]$ from $z$ to $q$ such that each value of $[0,1]$ corresponds to the slope of the tangent ray at a point on $\gamma_i$ (note that if $p$ is a vertex of $\gamma_i$, then the tangent ray at $p$ is not unique; to handle this case, we can conceptually consider each vertex of $\gamma_i$ as a small circle of infinitely small radius).
Similarly, we orient $\gamma_k$ from $q$ to $z$, and parameterize it over $[0, 1]$.

For each point $p_i\in \gamma_i$, define the function $f(p_i)$ as the parameter for the point $p_k\in \gamma_k$ such that the tangent ray of $p_i$ on $\gamma_i$ intersects the tangent ray of $p_k$ on $\gamma_k$ at a point in $B_{jk}$; see Figure~\ref{fig:bisectorintersect10}. Note that the tangent ray of $p_i$ must intersect $B_{ik}$. One can verify that $f$ is a monotone decreasing function.

\begin{figure}[t]
\begin{minipage}[t]{\linewidth}
\begin{center}
\includegraphics[totalheight=1.8in]{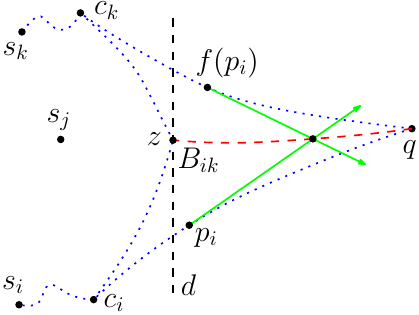}
\caption{Illustrating the definition of $f(p_i)$. The green rays are tangent rays.}
\label{fig:bisectorintersect10}
\end{center}
\end{minipage}
\end{figure}

Given a point $p_i\in \gamma_i$, we show that $f(p_i)$ can be computed in $O(\log m)$ time by binary search on the vertices of $\gamma_k$ using the monotone property of $f$. We first compute $d_{s_i}(p_i)$ in $O(\log m)$ time using the GH data structure.

Given any $p_k\in \gamma_k$, we can determine whether $f(p_i)>p_k$, $f(p_i)<p_k$, or $f(p_i)=p_k$ in the following way. If the tangent ray $\rho_{p_i}$ at $p_i$ on $\gamma_i$ intersects the tangent ray $\rho_{p_k}$ at $p_k$ on $\gamma_k$, let $p$ be the intersection point. Then, if $d_{s_i}(p_i)+|\overline{p_ip}|=d_{s_k}(p_k)+|\overline{p_kp}|$, then $p_k=f(p_i)$ and we are done. Otherwise, observe that $f(p_i)<p_k$ if and only if $d_{s_i}(p_i)+|\overline{p_ip}|<d_{s_k}(p_k)+|\overline{p_kp}|$. If $\rho_{p_i}$ does not intersect $\rho_{p_k}$, then since both rays intersect $B_{ik}$, either $\rho_{p_i}$ intersect the backward extension of $\rho_{p_k}$ or $\rho_{p_k}$ intersects the backward extension of $\rho_{p_i}$. In the former case we have $f(p_i)>p_k$ while in the latter case we have $f(p_i)<p_k$.
Therefore, once $d_{s_k}(p_k)$ is available, whether $f(p_i)>p_k$, $f(p_i)<p_k$, or $f(p_i)=p_k$ can be determined in $O(1)$ time.

Using the GH data structure, in $O(\log m)$ time we can compute the junction vertex $c_k$ and also a tree structure to represent the path $\pi(c_k,q)$ (resp., $\pi(c_k,z)$) so that we can do binary search on $\gamma_k$. Further, when we do the binary search on $\pi(c_k,q)$ (resp., $\pi(c_k,z)$), it is possible to keep track of $d_{s_k}(p_k)$ for each vertex $p_k\in \gamma_k$ considered in the binary search. In this way, using the monotone property of $f$, $f(p_i)$ can be computed in $O(\log m)$ time.

We now show how to compute $q_i$ in $O(\log^2 m)$ time ($q_k$ can be found similarly). We do binary search on $\gamma_i$. As above, using the GH data structure, in $O(\log m)$ time we can compute the junction vertex $c_i$ and also obtain a tree structure to represent the path $\pi(c_i,q)$ (resp., $\pi(c_i,z)$) so that we can do binary search on $\gamma_i$. In each iteration of the binary search, we consider a vertex $p_i$ of $\gamma_i$. We first compute $f(p_i)$ in $O(\log m)$ time using the above algorithm. Note that when computing $f(p_i)$, we use the tangent ray from $p_i$ that is the extension of the last segment of $\pi(s_i,p_i)$, which can be computed in $O(\log m)$ time using the GH data structure. With $f(p_i)$, the intersection $p$ between $B_{ik}$ and the above tangent ray from $p_i$ can be obtained. We compare $d_{s_i}(p)$ with $d_{s_j}(p)$. Both values can be computed in $O(\log m)$ time using the GH data structure.
If $d_{s_i}(p)=d_{s_j}(p)$, then $q^*$ is $p$ and we can terminate the algorithm. If $d_{s_i}(p)<d_{s_j}(p)$, then $q^*$ is to the  right of $p$ on $B_{ik}$ and thus we search the portion of $\gamma_i$ after $p_i$. Symmetrically, if $d_{s_i}(p)>d_{s_j}(p)$, then we search the portion of $\gamma_i$ before $p_i$. In this way, $q_i$ can be found within $O(\log m)$ iterations. As each iteration takes $O(\log m)$ time, the total time for computing $q_i$ is $O(\log^2 m)$.

In summary, the intersection point $q^*=B_{ij}\cap B_{jk}$ can be computed in $O(\log^2 m)$ time. The lemma thus follows.
\end{proof}

For the time analysis, notice that each event can be processed in $O(\log^2 m+\log |Q|)$ time. As $\avd(S,P_R)$ has $O(n)$ vertices and processing each event will produce a vertex of $\avd(S,P_R)$ and generate $O(1)$ new events for $Q$, the number of events that have ever appeared in $Q$ is $O(n)$. Hence, insertions and deletions for $Q$ take $O(\log n)$ time each. Once $Q$ becomes $\emptyset$, the rest of the algorithm takes $O(n\log^2 m)$ time. Therefore, the total time of the algorithm is bounded by $O(n\log n+n\log^2 m)$.

\subsection{Point locations}
\label{sec:point_location}
Given a query point $q\in P_R$, we can compute its nearest neighbor in $S$ by locating the Voronoi region of $\vd(S,P_R)$ containing $q$. Although we do not have $\vd(S,P_R)$ explicitly, we can find its Voronoi region containing $q$ using $\avd(S,P_R)$ instead. Since we already have the locations of the vertices of $\avd(S,P_R)$, we can obtain an embedding of $\avd(S,P_R)$. By Lemma~\ref{lem:xmonotone}, every edge of $\avd(S,P_R)$ is $x$-monotone. Hence, we can construct a point location data structure for $\avd(S,P_R)$ by using the algorithm of Edelsbrunner, Guibas, and Stolfi~\cite{ref:EdelsbrunnerOp86}, which takes $O(n)$ time to construct as the size of $\avd(S,P_R)$ is $O(n)$.

Given a query point $q\in P_R$, a primitive operation in the query algorithm of \cite{ref:EdelsbrunnerOp86} is to determine whether $q$ is below or above an edge $e\in \avd(S,P_R)$ belonging to $B_{ij}$ for two sites $s_i,s_j\in S$. As $B_{ij}$ is $x$-monotone, it suffices to determine whether $q$ is below or above $B_{ij}$.
This problem can be easily solved by comparing $d(s_i,q)$ with $d(s_j,q)$ in the unweighted case since the bisector intersects $d$ only once~\cite{ref:AgarwalIm18}. However, in our problem, since $B_{ij}\subseteq B(s_i,s_j)$ and $B(s_i,s_j)$ may intersect $d$ twice, we can not simply compare $d_{s_i}(q)$ with $d_{s_j}(q)$. Instead, we propose the following method to resolve the issue.

\begin{lemma}
Suppose that a GH data structure for $P$ is already available; then given a point $q\in P_R$ and two sites $s_i,s_j\in S$ with $i<j$, we can determine whether $q$ is below or above $B_{ij}$ in $O(\log m)$ time.
\end{lemma}
\begin{proof}
We utilize our algorithm for Lemma~\ref{lem:algointerdiagonal}.
If we extend the last edge of $\pi(s_i,s_j)$ beyond $s_j$, let $s_j'$ be the first point on $\partial P$ hit by the extension. Define $s_i'$ similarly. As argued in the proof of Lemma~\ref{lem:algointerdiagonal}, if $B(s_i,s_j)$ intersects $d$ twice, then one of $\overline{s_is_i'}$ and $\overline{s_js_j'}$ must cross $d$, which can be determined in $O(\log m)$ time.

If neither $\overline{s_is_i'}$ nor $\overline{s_js_j'}$ crosses $d$, then $B_{ij}=B(s_i,s_j)\cap P_R$. Since $i<j$, $q$ is below $B_{ij}$ if and only if $d_{s_i}(q)\leq d_{s_j}(q)$, which can be determined in $O(\log m)$ time using the GH data structure. Otherwise, as argued in  the proof of Lemma~\ref{lem:algointerdiagonal}, only one of $\overline{s_is_i'}$ and $\overline{s_js_j'}$ can cross $d$. Without loss of generality, we assume that $\overline{s_is_i'}$ cross $d$.

Let $z_i=d\cap \overline{s_is_i'}$; see Figure~\ref{fig:pointlocation}. The line segment $\overline{z_is_i'}$ partitions $P_R$ into two subpolygons; let $P_1$ denote the one locally above $\overline{z_is_i'}$ and $P_2$ the other. According to the analysis of Lemma~\ref{lem:algointerdiagonal}, since $i<j$, $B_{ij}=P_1\cap B(s_i,s_j)$. Hence, if $q\in P_1$, then as above, $q$ is below $B_{ij}$ if and only if $d_{s_i}(q)\leq d_{s_j}(q)$. Therefore, we can solve the problem in $O(\log m)$ time if $q\in P_1$. If $q\in P_2$, then it is obviously true that $q$ is below $B_{ij}$ as $B_{ij}\subseteq P_1$.

\begin{figure}[t]
\begin{minipage}[t]{\linewidth}
\begin{center}
\includegraphics[totalheight=1.8in]{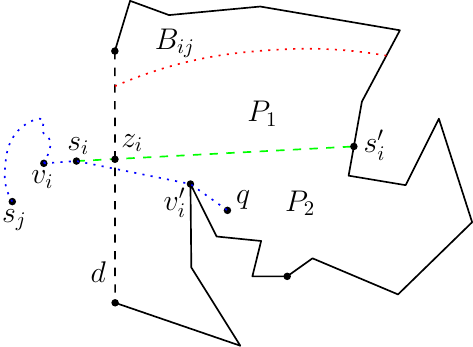}
\caption{The red curve is $B_{ij}$. The shortest paths $\pi(s_j,s_i)$ and $\pi(s_i,q)$ are shown with blue curves. The green dashed segment is the extension of $\overline{v_is_i}$.}
\label{fig:pointlocation}
\end{center}
\end{minipage}
\end{figure}

It remains to discuss how to determine whether $q$ is in $P_1$ or $P_2$. To this end, we compute the anchor $v_i$ of $s_i$ in $\pi(s_j,s_i)$ and the anchor $v_i'$ of $s_i$ in $\pi(s_i,q)$. Both anchors can be computed in $O(\log m)$ time using the GH data structure. Observe that the path $v_i\rightarrow s_i\rightarrow v_i'$ is a right turn if and only if $q\in P_2$; see Figure~\ref{fig:pointlocation}. Consequently, we can determine whether $q$ is in $P_1$ or $P_2$ in $O(\log m)$ time.

In summary, we can determine whether $q$ is above or below $B_{ij}$ in $O(\log m)$ time.
\end{proof}

With the above lemma, since the query algorithm of \cite{ref:EdelsbrunnerOp86} needs $O(\log n)$ primitive operations, the Voronoi region of $\vd(S,P_R)$ containing $q$ and thus the nearest neighbor of $q$ in $S$ can be computed in $O(\log n\log m)$ time.

This completes the proof of Lemma~\ref{lem:diasepneighborquery}.

\section{The unweighted simple polygon case}
\label{sec:simple_pol}

In this section, we present our SSSP algorithm for unweighted $\G$ when $P$ is a simple polygon.

To solve the problem, a natural idea is to see whether the algorithms for the Euclidean case can be adapted. Two algorithms are known that can solve the Euclidean case in $O(n\log n)$ time~\cite{ref:CabelloSh15,ref:ChanAl16}. The algorithm in \cite{ref:ChanAl16} relies on constructing a grid in the plane, which seems inherently not applicable anymore in the polygon setting.

We therefore adapt the approach of \cite{ref:CabelloSh15}. After introducing notation, we first describe the algorithm at a high level (omitting implementation details), then prove its correctness, and finally present the full implementation and running-time analysis.

\subsection{Notation}
\label{sec:notation}
We follow the notation already defined in Sections~\ref{sec:intro} and \ref{sec:pre}, including $m$, $n$, $S$, $P$, $s$, $\G$, $\vd(S)$, $\E(u,v)$, $\R(u)$.

We define $S_i$ for $i = 0, 1, 2, \ldots$ to be the set of vertices of $\G$ whose shortest path lengths from $s$ in $\G$ are equal to $i$. Clearly, $S_0=\{s\}$. Define $S_{\leq i} = \cup_{j \leq i} S_j$.

Our goal is to compute the sets $S_i$, for $i=0,1,2,\ldots$. In our following discussion, we will focus on describing the algorithm to compute these sets. By slightly modifying the algorithm, we can also compute the predecessor information of the shortest paths.

For any subset $S'\subseteq S$ and a point $p\in P$, as before, let $\beta_p(S')$ denote the nearest neighbor of $p$ in $S'$ and $d(S',p)=\min_{q\in S'}d(q,p)$.

Let $V(P)$ denote the set of vertices of $P$. Vertices of $P$ are also referred to as {\em obstacle vertices}.
For ease of exposition, we make the following general position assumptions: (1) the shortest paths in $P$ between any two points in $V(P)\cup S$ are unique, and (2) no point in $V(P)\cup S$ is equidistant (in geodesic distance) to two other points in $V(P)\cup S$.
Note that the second assumption implies that no obstacle vertex lies on a Voronoi edge of $\vd(S)$, which further means that each Voronoi edge of $\vd(S)$ is one-dimensional. The second assumption also implies that the nearest neighbor of a site $p\in S$ in any subset $S'\subseteq S$, i.e., $\beta_p(S')$, is unique.

We use $\DT(S)$ to denote the dual graph of $\vd(S)$. Specifically, the vertex set of $\DT(S)$ is $S$, and two sites $u,v\in S$ have an edge in $\DT(S)$ if their Voronoi regions in $\vd(S)$ share a Voronoi edge, i.e., $\E(u,v)\neq\emptyset$. Note that if $P$ is the entire plane, then $\DT(S)$ is the Delaunay triangulation of $S$. For any site $u \in S$, we define $\N(u)$ as the set of vertices of $\DT(S)$ adjacent to $u$; we also refer to sites of $\N(u)$ as {\em neighbors} of $u$ in $\DT(S)$. For a subset $S' \subseteq S$, by slightly abusing the notation, we let $\N(S') = \bigcup_{u \in S'} \N(u)$.

\subsection{Algorithm description}
We first construct the geodesic Voronoi diagram $\vd(S)$. The algorithm then proceeds in iterations. Initially, we have $S_0=\{s\}$. In the $i$-th iteration for $i\geq 1$, we compute the set $S_i$ as follows. We assume that $S_{i-1}$ is available.
Let $Q=\N(S_{i-1})\setminus S_{\leq i-1}$, i.e., the sites in
$S$ adjacent to those of $S_{i-1}$ in $\DT(S)$ but not in $S_{\leq i-1}$.
For each site $p\in Q$, we remove $p$ from $Q$ and determine whether $d(S_{i-1},p)\leq 1$. If so, we add $p$ to $S_i$. In addition, we  update $Q=Q\cup (\N(p)\setminus S_{\leq i})$, i.e., add the sites of $\N(p)$ that are not in $Q\cup S_{\leq i}$ to $Q$ for the current $S_i$.
We repeat this until $Q$ becomes empty. The $i$-th iteration finishes once $Q$ become empty.
If $S_i=\emptyset$, then the algorithm stops; otherwise, we proceed with the next iteration. See Algorithm~\ref{algo:simple} for the pseudocode.

\begin{algorithm}
    \caption{SSSP algorithm for the simple polygon case} \label{algo:simple}
    \KwIn{A simple polygon $P$, sites $S \subseteq P$, source site $s \in S$.}
    \KwOut{$S_i$ for $i = 0, 1, \ldots$}
    $i \gets 0$\;
    $S_0 \gets \{s\}$\;
    Build $\vd(S)$ and $\DT(S)$ \label{line:euc_dt}\;
    \While{$S_i \neq \emptyset$}{
        $i \gets i + 1$\;
        $S_i \gets \emptyset$\;
        $Q \gets \N(S_{i - 1}) \setminus S_{\leq i-1}$ \label{line:euc_Q1}\;
        \While{$Q \neq \emptyset$ \label{line:euc_Q_while}}{
            $p \gets pop(Q)$\;
            \If{$\df(S_{i - 1}, p) \leq 1$ \label{line:euc_check}}{
                $S_i \gets S_i \cup \{p\}$\; \label{line:add}
                $Q \gets Q \cup (\N(p) \setminus S_{\leq i})$ \label{line:euc_Q2}\;
            }
        }
    }
    \Return{$S_i$ for $i = 0, 1, \ldots$\;}
\end{algorithm}

\subsection{Algorithm correctness}
We now prove that the algorithm is correct. For convenience, we use $S'_i$ to refer to $S_i$ computed by our algorithm and use $S_i$ to refer to the ``correct'' set of sites whose shortest path lengths from $s$ in $\G$ are $i$. Our goal is to prove $S_i=S'_i$ for all $i=0,1,2,\ldots$. We prove it by induction. Initially, $S_0=S'_0$ obviously holds since both of them consist of the single site $s$.

Assuming that $S_j'=S_j$ for all $j<i$, we prove $S_i=S_i'$ as follows. First of all, since $S_j'=S_j$ for all $j<i$, our algorithm guarantees that every site of $S_i'$ is at distance $i$ from $s$ in $\G$ and thus $S_i'\subseteq S_i$. In the following, we show that every site $v\in S_i$ is also in $S_i'$. For this, notice that if $v$ is ever added to $Q$ in our algorithm, then $v$ will be added to $S_i'$ because it will eventually be popped from $Q$ and pass the check on Line~\ref{line:euc_check}. Hence, it suffices to argue that $v$ will be added to $Q$ in the $i$-th iteration of the algorithm.

We will prove by Lemma~\ref{lem:path} below that there exists a site $u\in S_{i-1}=S_{i-1}'$ such that $\DT(S)$ has a $u$-$v$ path $\omega$ whose internal vertices are all in $S_i$ (note that {\em internal vertices} refer to vertices of $\omega$ excluding $u$ and $v$). Let the vertices of the path $\omega$ from $u$ to $v$ be $w_0,w_1,w_2,\ldots,w_k$ where $w_0 = u$ and $w_k=v$. We claim that all vertices of $w_t$, $1\leq t\leq k$, will be added to $Q$ in the $i$-th iteration. Indeed, since $\omega$ is a path of $\DT(S)$, we have $w_1\in \N(u)$. Therefore, according to our algorithm, $w_1$ must be added to $Q$ on Line~\ref{line:euc_Q1}. Because $w_1 \in S_i$, after $w_1$ is popped from $Q$, it will be added to $S_i'$. Since $w_2\in \N(w_1)$, $w_2$ will be added to $Q$ on Line \ref{line:euc_Q2} no later than right after $w_1$ is added to $S'_i$ on Line~\ref{line:add}. Following the same argument, $w_3,w_4,\ldots,w_k=v$ will all be added to $Q$ and thus added to $S_i'$ as well.

This proves $S_i=S_i'$ and thus the correctness of the algorithm.

\begin{restatable}{lemma}{lempath}
    \label{lem:path}
    For any $i > 0$ and any site $v \in S_i$,
    there exist a site $u \in S_{i - 1}$ and a $u$-$v$ path in $G(S) \cap \DT(S)$ such that all internal vertices of the path are in $S_i$.
\end{restatable}
\begin{proof}
    We extend the proof of \cite[Lemma 1]{ref:CabelloSh15} for the Euclidean case to our geodesic case. We remark that this proof is also applicable to the general polygonal domain case.

    Let $u =\beta_v(S_{i-1})$, i.e., $u$ is the nearest neighbor of $v$ in $S_{i-1}$. Consider the shortest path $\pif(u, v)$. For the sake of simplicity, we assume that $\pi(u,v)$ does not contain any Voronoi vertex of $\vd(S)$ (otherwise, we could replace $v$ by a point arbitrarily close to $v$).
    Let $w_0, w_1, \ldots, w_k$ where $w_0 = u$ and $w_k = v$ be the sites whose Voronoi regions in $\vd(S)$ intersect $\pif(u, v)$
    in the order from $u$ to $v$.
    Note that it is possible that $\pif(u, v)$ reenters the same Voronoi region multiple times,
    so it could be that $w_j = w_{j'}$ for $j \neq j'$. Since $\pif(u, v)$ does not intersect any Voronoi vertex of $\vd(S)$,
    we know that for every $0\leq j\leq k-1$, $\pif(u, v)$ crosses a common edge of $\R(w_j)$ and $\R(w_{j+1})$, meaning that $w_j$ and $w_{j+1}$ are connected by an edge in $\DT(S)$. Therefore, $\DT(S)$ has a $u$-$v$ path  $\omega$ with vertices $u=w_0, w_1, \ldots, w_k=v$.

    We next show that $\omega$ is also a path in $\G$. To this end, it suffices to prove that $d(w_j,w_{j+1})\leq 1$ for all $0\leq j\leq k-1$. Let $c$ be the midpoint of $\pif(u, v)$ and let $r = \df(u, c) = \df(c, v)$. Since $v \in S_i$ and $u = \beta_v(S_{i-1})$, it holds that $\df(u, v) \leq 1$ and thus $r \leq \frac{1}{2}$.
    For any $w_j$ ($0\leq j\leq k$), by definition, there exists a point $q \in \R(w_j) \cap \pif(u, v)$.
    Without loss of generality, assume that $q \in \pif(u, c)$ (otherwise, the proof is symmetric).
    We find that:
    \[\df(c, w_j) \os 1 \leq \df(c, q) + \df(q, w_j) \os 2 \leq  \df(c, q) + \df(q, u) \os 3 = r.\]
    where (1) is due to the triangle inequality,
    (2) is due to $q \in \R(w_j)$,
    and (3) is due to $q \in \pif(u, c)$.
        By the triangle inequality, we now have $\df(w_j, w_{j + 1}) \leq \df(w_j, c) + \df(c, w_{j + 1}) \leq  2 r \leq 1$.

    The above proves that $\omega$ is a path in $G(S)\cap \DT(S)$. In the following, we prove that all internal vertices of $\omega$ are in $S_i$, which will prove the lemma.

    Consider an internal vertex $w_j$ of $\omega$. Our goal is to show that $w_j\in S_i$.
    First of all, since $u\in S_{j-1}$ and $\df(u, w_j) \leq \df(u, c) + \df(c, w_j) \leq  2 r \leq 1$, we obtain that
    $w_j \in S_{\leq i}$. We next argue that $w_j$ cannot be in $S_{\leq i-1}$. Assume to the contradiction that $w_j \in S_{\leq  i-1}$. As before, pick a point $q \in \R(w_j) \cap \pif(u, v)$.
    We can derive:
    \[\df(w_j, v) \leq \df(w_j, q) + \df(q, v) \leq  \df(u, q) + \df(q, v) = \df(u, v) \leq 1.\]
    If $w_j \in S_{i - 1}$, then $\df(w_j, v) \leq \df(u, v)$ contradicts with the fact that $u$ is the unique nearest neighbor of $v$ in $S_{i-1}$.
    If $w_j \in S_{\leq i - 2}$, then $\df(w_j, v) \leq 1$ implies that $v\in S_{\leq i-1}$,
    contradicting with $v \in S_i$. Therefore, we obtain that $w_j$ must be in $S_i$.
\end{proof}

\subsection{Algorithm implementation and time analysis}
We now discuss how to implement the algorithm efficiently.

First of all, to compute $\vd(S)$ and thus $\DT(S)$, Aronov~\cite{ref:AronovOn89} gave an algorithm of $O((n+m)\log(n+m)\log m)$ time. Papadopoulou and Lee later presented an improved algorithm that runs in $O((n+m)\log (n+m))$ time~\cite{ref:PapadopoulouA98}. Recently, Oh and Ahn~\cite{ref:OhVo20} proposed a new algorithm of $O(m+n\log n\log^2 m)$ time and Oh~\cite{ref:OhOp19} finally solved the problem in $O(m+n\log n)$ time.
Note that the combinatorial size of $\vd(S)$ is $O(n+m)$~\cite{ref:AronovOn89}. In addition, since $\DT(S)$ has $n$ vertices and is a planar graph, $\DT(S)$ has $O(n)$ edges.

Notice that throughout the whole algorithm, each edge of $\DT(S)$ is considered at most four times (twice for each vertex) in the construction of the neighborhood sets.
Specifically, the edge $(u,v)$ for $u \in S_i$ will be considered by Line \ref{line:euc_Q2} in the $i$-th iteration and by Line~\ref{line:euc_Q1} in the $i + 1$-th iteration.
There are $O(n)$ edges in $\DT(S)$,
so there are at most $O(n)$ insertions (and thus pop operations) to $Q$.
Hence, the while loop on Line \ref{line:euc_Q_while} has at most $O(n)$ iterations
throughout the whole algorithm. It remains to determine the time complexity for the operation of checking whether $d(S_{i-1},p)\leq 1$ on Line~\ref{line:euc_check}.

Notice that $d(S_{i-1},p)\leq 1$ if and only if $S_{i-1}$ has a point $q$ with $d(q,p)\leq 1$. Hence, the problem can be solved by a geodesic unit-disk range emptiness (GUDRE) query.
Using Theorem~\ref{theo:rangeempty}, we construct a GUDRE data structure for $S_{i - 1}$ in $O(|S_{i - 1}|\cdot (\log^3 m+ \log n\log^2 m))$ time at the beginning of each $i$-th iteration of Algorithm~\ref{algo:simple}, assuming that a GH data structure has been built for $P$ in the beginning of the entire algorithm. Using a GUDRE query,
whether $d(S_{i-1},p)\leq 1$ can be determined in $O(\log n\log^2 m)$ time by Theorem~\ref{theo:rangeempty}. Note that no deletions on $S_{i-1}$ are needed.
Since $\sum_i|S_{i - 1}|=n$, the overall time for constructing the GUDRE data structure for $S_{i - 1}$ in all iterations of the algorithm is $O(n(\log^3 m+\log n\log^2 m))$.

We conclude that the total time of the algorithm is $O(m+n\log n + n(\log^3 m+\log n\log^2 m))$, which is $O(m+n\log n\log^2 m)$ as explained in the proof of the following theorem.


\begin{theorem}
Given a simple polygon $P$ of $m$ vertices, a set $S$ of $n$ points in $P$, and a source point $s\in S$, we can compute shortest paths from $s$ to all points of $S$ in the geodesic unit-disk graph of $S$ in $O(m+n\log n\log^2 m)$ time.
\end{theorem}
\begin{proof}
We argue that $m+n\log n + n(\log^3 m+\log n\log^2 m)=O(m+n\log n\log^2 m)$. It suffices to show that $n\log^3 m=O(m+n\log n\log^2 m)$. Indeed, if $m<n^2$, then $n\log^3 m=O(n\log n\log^2 m)$; otherwise, $n\log^3 m=O(m)$. Hence, it follows that $n\log^3 m=O(m+n\log n\log^2 m)$.
\end{proof}

Note that when $m=O(1)$, the runtime of our algorithm is $O(n\log n)$, which matches the optimal $O(n\log n)$ time algorithm for the Euclidean case unit-disk graph SSSP problem~\cite{ref:ChanAl16}.

\section{The unweighted polygonal domain case}
\label{sec:pol_holes}

In this section, we present our SSSP algorithm for $G(S)$ when $P$ is a polygonal domain with holes.
In fact, the framework of Algorithm~\ref{algo:simple} continues to apply: its correctness hinges on Lemma~\ref{lem:path}, whose proof extends directly to polygonal domains (as remarked in the proof of the lemma).
The main difficulty is algorithmic.
Unlike the simple polygon setting, we do not have an efficient analogue of Theorem~\ref{theo:rangeempty} for polygonal domains (even without deletions) that would allow us to evaluate the condition $d(S_{i-1},p)\leq 1$ on Line~\ref{line:euc_check}.

As before, the algorithm proceeds in iterations, with the $i$th iteration responsible for computing $S_i$.
However, we distinguish between iterations based on the size of $S_i$.
If $|S_i|\ge \sqrt{n}$, we classify the iteration as \emph{heavy}; otherwise, it is \emph{light}.
We design separate procedures for these two types, referred to as the \emph{heavy-iteration procedure} and the \emph{light-iteration procedure}, respectively.
At the start of iteration $i$, the value of $|S_i|$ is not known in advance.
Our strategy is therefore to always begin with the light-iteration procedure.
If, during its execution, more than $\sqrt{n}$ sites are discovered for $S_i$, we immediately abort the light-iteration procedure and switch to the heavy-iteration procedure.

A useful observation is that there can be at most $\sqrt{n}$ heavy iterations, since the sets $S_1, S_2, \ldots$ are disjoint, their union is $S$, and each heavy iteration has at least $\sqrt{n}$ sites.



In the following, we first present the heavy-iteration procedure as it is relatively easy. We follow the notation and the discussion in Section~\ref{sec:notation}.

\subsection{Heavy-iteration procedure}
\label{sec:heavy}

The procedure starts with computing the geodesic Voronoi diagram of $S_{i-1}$ in $P$, denoted by $\vd(S_{i-1})$, which can be done in $O((|S_{i-1}|+m)\log (|S_{i-1}|+m))$ time by the algorithm of Hershberger and Suri~\cite{ref:HershbergerAn99}. Note that each Voronoi region of $\vd(S_{i-1})$ is a {\em shortest path map} with respect to the site of the region and the combinatorial complexity of $\vd(S_{i-1})$ is $O(|S_{i-1}|+m)$~\cite{ref:HershbergerAn99}. We then construct a point location data structure on the diagram in $O(|S_{i-1}|+m)$ time~\cite{ref:KirkpatrickOp83,ref:EdelsbrunnerOp86}. Finally, for each site $p\in S\setminus S_{\leq i-1}$, using $\vd(S_{i-1})$ and the point location data structure, we can compute the nearest neighbor $\beta_p(S_{i-1})$ and its geodesic distance from $p$ in $O(\log (m+|S_{i-1}|))$ time~\cite{ref:HershbergerAn99}; note that their geodesic distance is $d(S_{i-1},p)$. If $d(S_{i-1},p)\leq 1$, we add $p$ to $S_i$.
This computes $S_i$ in $O((|S_i|+m)\log (|S_{i-1}|+m)+n\log (|S_{i-1}|+m))$ time, which is bounded by $O((n+m)\log (n+m))$.

Since the number of heavy iterations in the entire algorithm is at most $\sqrt{n}$, the total time of all heavy-iteration procedures in the algorithm is $O(\sqrt{n}(n+m)\log (n+m))$.

\paragraph{Remark.} Instead of first constructing $\vd(S_{i-1})$ and then computing $d(p,S_{i-1})$ using the point location queries, we can use the algorithm of Hershberger and Suri~\cite{ref:HershbergerAn99} to directly compute $d(S_{i-1},p)$ for all sites $p\in  S\setminus S_{\leq i-1}$. Indeed, we can add the sites of $S\setminus S_{\leq i-1}$ to $P$ as new obstacle vertices when applying Hershberger and Suri's algorithm with $S_{i-1}$ as the starting points to generate wavelets. The algorithm can directly compute $d(S_{i-1},p)$ for all obstacle vertices $p$ in the new $P$ (with sites of $S\setminus S_{\leq i-1}$ as the new obstacle vertices).
Since the complexity of the new $P$ is $O(m+|S\setminus S_{\leq i-1}|)$, which is bounded by $O(n+m)$, the runtime of Hershberger and Suri's algorithm is $O((n+m)\log (n+m))$. In this way, we can still compute $S_i$ in $O((n+m)\log (n+m))$ time in each heavy-iteration procedure.

\subsection{Light-iteration procedure}
\label{sec:small}
As mentioned above, each $i$-th iteration of our algorithm will start with running the light-iteration procedure. The procedure will stop if either all sites of $S_i$ have been computed and $|S_i|<\sqrt{n}$, or the algorithm detects that $|S_i|\geq \sqrt{n}$. In the former case, the $i$-th iteration finishes. In the latter case, we switch to the heavy-iteration procedure. We will show that the total time of the light-iteration procedure in the entire algorithm is $O(\sqrt{n}(n+m)\log (n+m))$.

The light-iteration procedure runs in rounds, and each round computes at least one new point for $S_i$. If a round does not compute any new point for $S_i$, then we can assert that all sites of $S_i$ have been computed. If after a round $|S_i|$ becomes larger than or equal to $\sqrt{n}$, then we terminate the light-iteration procedure and switch to the heavy-iteration procedure. In the following, we first give some intuition and introduce several new concepts.

In each round, we search in the neighborhoods of the sites of $S_{i - 1}$ and also the neighborhoods of the sites of the current $S_i$ (i.e., the sites of $S_i$ we have computed thus far). The observation is that we do not need the geodesic Voronoi diagram of $S_{i-1}$ in the entire polygon $P$ to conduct this search, and instead we just need to compute the Voronoi diagram for $S_{i-1}$ in the region that is the union of the Voronoi regions $\R(v)$ for all $v\in S_{i - 1}\cup S_i\cup \N(S_{i - 1} \cup S_i)$. However, it turns out that this is still not quite restrictive enough. For example, it is possible for some site $v \in S$ to have $\Omega(n)$ neighbors in $\DT(S)$ and $\Omega(m)$ obstacle vertices in $\R(v)$. We could end up with having a quadratic $O(n m)$ term in the runtime if we compute a Voronoi diagram involving $\R(v)$ once for every one of its neighbors.

To circumvent the above issue, we will show that we do not have to use the entire $\R(v)$ to compute $d(S_{i-1},v)$ for a site $v\in S\setminus S_{\leq i-1}$. More specifically, if there is a shortest path $\pi(S_{i-1},v)$ from a site of $S_{i - 1}$ to $v$, the observation is that it is possible that we only need to use a ``small'' portion of $\R(v)$ to find $\pi(S_{i-1},v)$. More specifically, suppose that $\pi(S_{i-1},v)$ enters $\R(v)$ through $\R(u)$ for a neighbor $u \in \N(v)$.
Consider the Voronoi edge $\E(u, v)$ of $\vd(S )$ between $u$ and $v$.
This edge may be comprised of multiple connected components $e_1, \ldots, e_k$
where each component is a sequence of hyperbolic arcs. For each component $e$, let $a$ and $b$ be its two endpoints, each of which is either a Voronoi vertex of $\vd(S)$ or an intersection between $e$ and the boundary of $P$.
Note that since $a$ is a point on the boundary of $\R(v)$, $\pi(v,a)$ must be in $\R(v)$. Similarly, $\pi(v,b)$ is also in $\R(v)$.
Observe that the cycle $\pif(v, a)$, $\pif(v, b)$, and $e$ divides $P$ into two portions, and exactly one of them must be inside $\R(v)$. We use $\Delta_v(e)$ to denote that region and call it a {\em pseudo-triangle}.\footnote{Note that if shortest paths between $v$ and $a$ are not unique, then we let $\pi(v,a)$ refer to an arbitrary one and fix it for the definition of $\Delta_v(e)$. We do the same for $\pif(v, b)$.}
We are interested in $\Delta_v(e)$ mainly due to the following observation.


\begin{restatable}{observation}{obser:trian_cont}\label{lem:trian_cont}
 For any point $p$ in $e$, there exists a shortest path $\pif(p, v)$ that is inside the pseudo-triangle $\Delta_v(e)$.
\end{restatable}
\begin{proof}
If $p$ is an endpoint of $e$, the lemma obviously holds following the definition of $\Delta_v(e)$. In the following, we assume that $p$ is in the interior of $e$.

First of all, since $e$ is on the boundary of $\R(v)$, $\pi(v,p)$ must be inside $\R(v)$.
Since $p$ is in the interior of $e$, if we move on $\pi(v,p)$ from $p$ to $v$, we will enter the interior of $\R(v)$ and thus also the interior of $\Delta_v(e)$. This means that if $\pif(v,p)$ exits $\Delta_v(e)$, then it must cross either $\pi(v,a)$ or $\pi(v,b)$. But this is not possible since no two shortest paths from the same point can cross each other~\cite{ref:HershbergerAn99}.
\end{proof}

\begin{figure}
    \centering
    \includegraphics[height = 9cm]{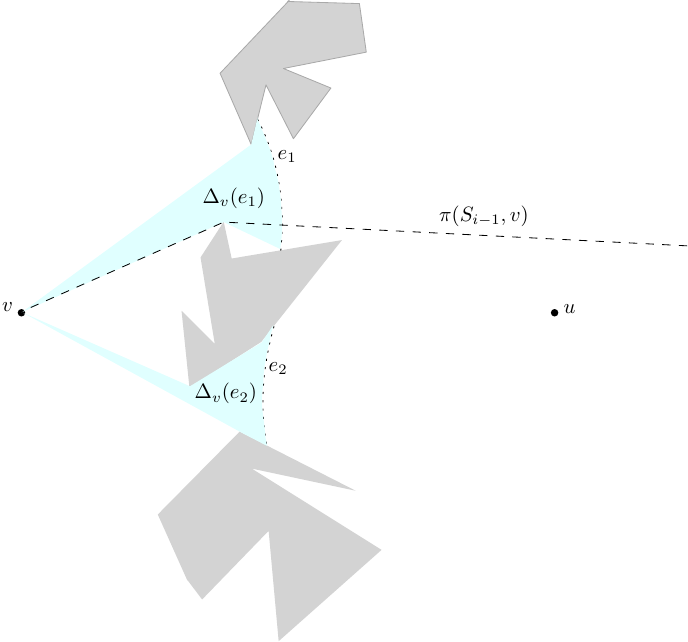}
    \caption{The gray polygons are obstacles.
    The path $\pi(S_{i - 1}, v)$ (dashed polyline) passes through the connected component $e_1$ (dotted curve) of $\E(u, v)$,
    so $\pi(S_{i - 1}, v) \cap \R(v)$ is contained in $\Delta_v(e_1)$ (the upper light blue region).}
    \label{fig:through_triangle}
\end{figure}

Observation~\ref{lem:trian_cont} implies that if we want to include the shortest path $\pi(S_{i-1},v)$ through the Voronoi edge $\E(u, v)$,
the only portion of $\R(v)$ that we need to include when constructing the Voronoi diagram for $S_{i-1}$ is union of the pseudo-triangles $\Delta_v(e)$ for all components $e\subseteq \E(u,v)$.
See Figure~\ref{fig:through_triangle}.
Motivated by this, for each site $u$, we define $\bigstar(u)$ to be the union of $\R(u)$ and all $\Delta_v(e)$ for all $v \in \N(u)$ and all connected components $e\subseteq \E(u, v)$.
See Figure~\ref{fig:star} for an illustration.
We use the symbol $\bigstar$ because this region looks like a star with $\R(u)$ in the center and pseudo-triangles $\Delta_v(e)$ in the periphery; we refer to these pseudo-triangles $\Delta_v(e)$ as {\em peripheral pseudo-triangles} of $\bigstar(u)$.
By including $\bigstar(u)$ in our ``work region'' for computing the Voronoi diagram for $S_{i-1}$, we ensure that for all $v \in \N(u)$,
if $\pif(S_{i - 1}, v)$ crosses $\E(u, v)$, then this path will be included in our work region.

\begin{figure}
    \centering
    \includegraphics[height = 7cm]{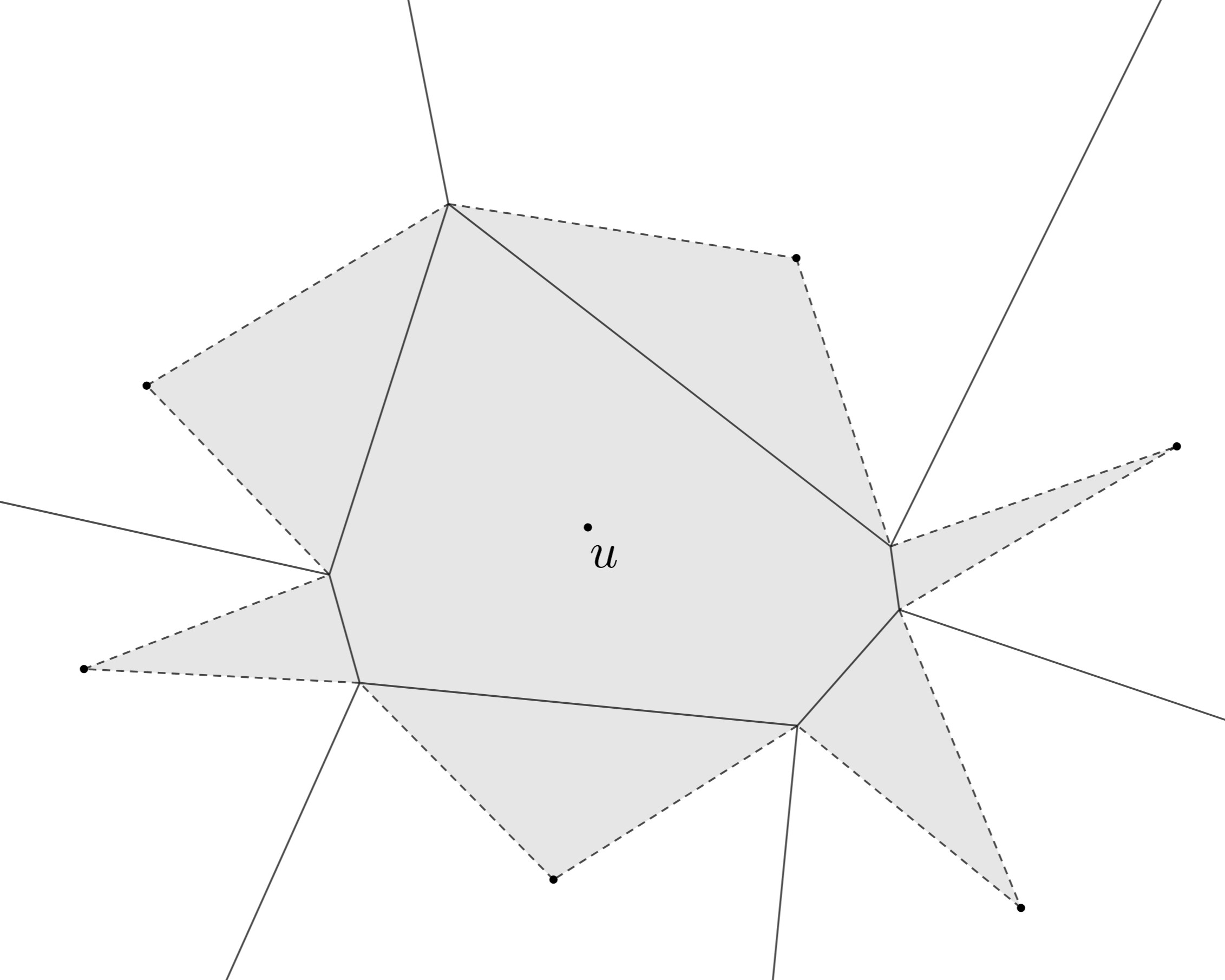}
    \caption{Illustrating $\bigstar(u)$.
        The Voronoi diagram is drawn with solid segments. $\bigstar(u)$ is drawn shaded and bounded by dashed segments.}
    \label{fig:star}
\end{figure}

By slightly abusing the notation, for any subset $S'\subseteq S$, we define $\bigstar(S')=\cup_{u\in S'}\bigstar(u)$. If $\Delta_v(e)$ is a pseudo-triangle in $\bigstar(S')$ with $v\not\in S'$, then $\Delta_v(e)$ is also in the ``periphery'' of $\bigstar(S')$ and therefore we refer to $\Delta_v(e)$ as a {\em peripheral pseudo-triangle} of $\bigstar(S')$.
Note that each $\bigstar(u)$ for $u \in S'$ is a polygon with holes,
so $\bigstar(S')$ is a union of polygons with holes.
In particular, the boundary of $\bigstar(S')$ does not include any curved Voronoi edges of $\vd(S)$.

In what follows, we first give the description of the light-iteration procedure, then prove its correctness, and finally analyze the time complexity.

\subsubsection{Procedure description}
With the above discussion, we describe our light-iteration procedure as follows. In the beginning, we set $S_i=\emptyset$. Each round of the procedure works as follows.

We first compute the region $\F=\bigstar(S_{i-1}\cup S_i)$ with the current $S_i$ (i.e., the set of sites that have been added to $S_i$ so far); we call $\F$ {\em the work region} for this round. We also compute the subset $S_{\F}$ of sites of $S$ in $\F$ that are not in $S_{\leq i-1}\cup S_i$ (again, $S_i$ refers to the current $S_i$).
We then construct the geodesic Voronoi diagram for the sites $S_{i-1}$ but in the region $\F$ only, denoted by $\vd(S_{i-1},\F)$, and build a point location data structure for it~\cite{ref:KirkpatrickOp83,ref:EdelsbrunnerOp86}.
Using the point location data structure, we compute $d_{\F}(S_{i-1},p)$ for each point $p \in S_{\F}$
where $d_\F(\cdot, \cdot)$ is the geodesic distance in $\F$.
If $d_{\F}(S_{i-1},p)\leq 1$, we add $p$ to $A$ (which is initialized to $\emptyset$ at the beginning of each round). After all points $p\in S_{\F}$ are considered, we do the following. If $A\neq \emptyset$, then we add all sites of $A$ to $S_i$. If $|S_i|\geq \sqrt{n}$, then we terminate the light-iteration procedure and switch to the heavy-iteration procedure; otherwise, we proceed with the next round of the light-iteration procedure. On the other hand, if $A=\emptyset$, then this round does not add any new site for $S_i$ and we terminate the light-iteration procedure and also finish the $i$-th iteration of the algorithm.

Algorithm~\ref{algo:combined} gives the pseudocode of the entire algorithm. Note that, as remarked in Section~\ref{sec:heavy}, instead of computing $d_{\F}(S_{i-1},p)$ for all sites $p\in S_{\F}$ by first constructing $\vd(S_{i-1},\F)$ and then using point location queries, we could alternatively compute $d_{\F}(S_{i-1},p)$ directly by applying Hershberger and Suri's algorithm~\cite{ref:HershbergerAn99} with sites of $S_{\F}$ as new obstacle vertices.

\begin{algorithm}
    \caption{SSSP in geodesic unit-disk graphs in polygonal domains} \label{algo:combined}
    \KwIn{A polygonal domain $P$, sites $S \subseteq P$, source site $s \in S$.}
    \KwOut{$S_i$ for $i = 0, 1, \ldots$}
    $i \gets 0$\;
    $S_0 \gets \{s\}$\;
    Construct $\vd(S)$ and $\DT(S)$ for $S$ in $P$ \label{line:comb_vddt}\;
    \While{$S_i \neq \emptyset$}{
        $i \gets i + 1$\;
        $S_i \gets \emptyset$\;

        \tcp{light-iteration procedure. Early termination if $|S_i| \geq \sqrt{n}$.}
        \Do{$|S_i| < \sqrt{n}$ and $A \neq \emptyset$ \label{line:small_end}}{
            \label{line:small_start}
            Compute $\F \gets \bigstar(S_{i - 1} \cup S_i)$ and also compute $S_{\F}$, the set of all sites $p\in \F$ with $p\not\in S_{\leq i-1}\cup S_i$\;\label{line:workregion}
            Build the geodesic Voronoi diagram $\vd(S_{i-1},\F)$ for sites of $S_{i - 1}$ in $\F$ and construct a point location data structure for $\vd(S_{i-1},\F)$\;\label{line:vd}
            Using $\vd(S_{i-1},\F)$, compute $d_{\F}(S_{i - 1}, p)$ for each site $p\in S_{\F}$\; \label{line:comdis}
            $A \gets \{p \in S_{\F}:  d_{\F}(S_{i - 1}, p) \leq 1\}$\;\label{line:getA}
            $S_i \gets S_i \cup A$\;\label{line:update}
        }
        \tcp{heavy-iteration procedure.}
        \If{$|S_i| \geq \sqrt{n}$ \label{line:large_start}}{
            Build the geodesic Voronoi diagram $\vd(S_{i-1})$ for the sites $S_{i - 1}$ in $P$ and construct a point location data structure for $\vd(S_{i-1})$\;
            $S_i \gets \emptyset$\;
            Using $\vd(S_{i-1})$, compute $d(S_{i-1},p)$ for each site $p\in S\setminus S_{\leq i-1}$\;
            $S_i \gets \{p\in S\setminus S_{\leq i-1}: d(S_{i-1},p)\leq 1\}$\;
        }
    }
    \Return{$S_0, S_1, \ldots$\;}
\end{algorithm}

\subsubsection{Proving the correctness}
We now argue the correctness of the light-iteration procedure.
The procedure is terminated due to one of the two cases on Line~\ref{line:small_end}: (1) $|S_i|\geq \sqrt{n}$; (2) $|A|=\emptyset$.
First of all, according to our algorithm, for each site $p\in A$, it is guaranteed that $d_{\F}(S_{i-1},p)\leq 1$. Because $\F \subseteq P$, $d(S_{i-1},p)\leq d_{\F}(S_{i-1},p)$. Therefore, we have $d(S_{i-1},p)\leq 1$. Because all points added to $S_i$ are indeed at distance $i$ from $s$ in $\G$, if the algorithm switches to the heavy-iteration procedure, it is true that $|S_i|\geq \sqrt{n}$.
The correctness proof is completed with the following lemma.

\begin{restatable}{lemma}{lemsmallcor} \label{lem:small_cor}
    If $A=\emptyset$ after a round in the light-iteration procedure, then the set $S_i$ has been correctly computed and $|S_i|<\sqrt{n}$.
\end{restatable}
\begin{proof}
Suppose that $A=\emptyset$ after a round of the light-iteration procedure. Let $S_i'$ be the $S_i$ in Algorithm~\ref{algo:combined} and let $S_i$ refer to the ``true'' set of sites whose distance from $s$ in $\G$ is equal to $i$. Our goal is to prove that $S_i=S_i'$. As discussed above, whenever the algorithm adds a point $p$ to $S_i'$ in Line~\ref{line:update}, it is guaranteed that $p\in S_i$. Hence, $S_i'\subseteq S_i$ holds. It remains to show that $S_i\subseteq S_i'$.

Assume to the contradiction that $S_i\not\subseteq S_i'$ and pick $v = \argmin_{v' \in S_i \setminus S'_i} \df(S_{i - 1}, v')$.
As in the proof of Lemma~\ref{lem:path}, let $u = \beta_v(S_{i-1})$ and $u=w_0,w_1, \ldots, w_k=v$ be the sites of $S$ whose Voronoi regions in $\vd(S)$ intersect $\pif(u, v)$ in the order from $u$ to $v$.
In the following, we first argue that $w_j\in S'_i$ for all $1\leq j\leq k-1$.

Following the same proof as Lemma~\ref{lem:path} (recall that the proof is applicable to the polygonal domain case), we have $w_j\in S_i$.
Assume for the sake of contradiction that $w_j \notin S_i'$.
Then, we have $w_j\in S_i\setminus S_i'$, and we can further derive:
\[\df(u, v) \os 1 = \df(S_{i - 1}, v) \os 2 \leq \df(S_{i - 1}, w_j) \os 3 \leq \df(u, w_j),\]
where (1) is due to $u = \beta_v(S_{i-1})$, (2) is due to $v = \argmin_{v' \in S_i \setminus S_i'} \df(S_{i - 1}, v')$ and $w_j\in S_i\setminus S_i'$, and (3) is due to $u \in S_{i - 1}$.
Similarly, by considering any point $q \in \R(w_j) \cap \pif(u, v)$, we can derive:
\[\df(u, w_j) \os 1 \leq \df(u, q) + \df(q, w_j) \os 2\leq \df(u, q) + \df(q, v) \os 3 = \df(u, v),\]
where (1) is due to the triangle inequality, (2) is due to $q \in \R(w_j)$, and (3) is due to $q \in \pif(u, v)$.
Taken together, we get that $\df(u, v) = \df(u, w_j)$.
However, this is in contradiction with our general position assumption that
no point in $S$ is equidistant to two other points in~$S$.



The above proves that $w_j\in S_i'$ for all $1\leq j\leq k-1$. In the following, we argue that $\pi(u,v)$ must be in $\F$, where $\F$ is defined by Line~\ref{line:workregion} in the beginning of this round.
By definition, 
$\pif(u, v)$ is contained in the union of $\R(w_j)$ for all $0\leq j\leq k$.
Since $w_0=u \in S_{i - 1}$, ${\R(w_0)} \subseteq \bigstar(w_0) \subseteq \F$.
Similarly, for all $j$, $1\leq j\leq k-1$, since $w_j \in S_i'$, we have $\R(w_j) \subseteq \bigstar(w_j) \subseteq \F$.
It remains to show that $\pif(u, v) \cap \R(v) \subseteq \F$.
We know that $\pif(u, v)$ enters $\R(v)$ through some point $q \in \E(w_{k - 1}, v)$. Let $e$ be the connected component of $\E(w_{k - 1},v)$ that contains $q$. By Observation~\ref{lem:trian_cont}, the subpath of $\pi(u,v)$ between $q$ and $v$, which is exactly $\pif(u, v) \cap \R(v)$, must be in the pseudo-triangle $\Delta_v(e)$. By definition, $\Delta_v(e)\subseteq \bigstar(w_{k-1})$. Since $\bigstar(w_{k-1}) \subseteq \F$, we have $\Delta_v(e)\subseteq\F$, and thus $\pif(u, v) \cap \R(v) \subseteq \F$. Hence, we conclude that $\pif(u, v) \subseteq \F$.

We find that $v$ must have been added to $A$ on Line~\ref{line:getA} because
\[d_\F(S_{i - 1}, v) \os 1 \leq d_\F(u, v) \os 2 = d(u, v) \os 3 = d(S_{i - 1}, v) \os 4 \leq 1,\]
where (1) is due to $u \in S_{i - 1}$,
(2) is due to $\pif(u, v) \subseteq \F$,
(3) is due to $u = \beta_v(S_{i - 1})$,
and (4) is due to $v \in S_i$.
This contradicts with $A = \emptyset$, completing our proof that $S_i\subseteq S_i'$ and therefore $S_i=S_i'$.

Finally, to see that $|S_i|<\sqrt{n}$, first notice that right before the current round of the algorithm, $|S_i|<\sqrt{n}$ must hold since otherwise the algorithm would have been terminated. After the round, since $A=\emptyset$, this round does not add any new point to $S_i$, so we still have $|S_i|<\sqrt{n}$. The lemma thus follows.
\end{proof}

\subsubsection{Time analysis}

We now analyze the runtime of the light-iteration procedure. We show that the total time of it in all iterations of the entire algorithm is $O(\sqrt{n}(n+m)\log (n+m))$.

First of all, notice that in each iteration, the light-iteration procedure will run for at most $\sqrt{n}$ rounds. This is because we add at least one point to $S_i$ in each round (except possibly the last) and we do not start a new round if more than $\sqrt{n}$ points have been added to $S_i$. In the following discussion, let $S_i'$ be the $S_i$ in Algorithm~\ref{algo:combined} and let $S_i$ refer to the ``true'' set of sites whose distance from $s$ in $\G$ is equal to $i$.

The step in Line~\ref{line:workregion} can be done in $O(|\F|+|S_{i-1}|+|S_i'|)$ time,\footnote{Note that after $\vd(S)$ is computed by the algorithm of \cite{ref:HershbergerAn99}, the algorithm also decomposes each Voronoi cell $\R(v)$ of $\vd(S)$ into a shortest path map with respect to $v$ such that given any point $p\in \R(v)$, a shortest path $\pi(v,p)$ must be in $\R(v)$ and can be found in time linear in the number of edges of the path~\cite{ref:HershbergerAn99}. Using these results, Line~\ref{line:workregion} can be done in $O(|\F|+|S_{i-1}|+|S_i'|)$ time.}
where $|\F|$ is the combinatorial complexity of $\F$. Computing the geodesic Voronoi diagram $\vd(S_{i-1},\F)$ on Line~\ref{line:vd} can be done in $O((|S_{i-1}|+|\F|)\log (|S_{i-1}|+|\F|))$ by the algorithm of Hershberger and Suri~\cite{ref:HershbergerAn99}. We can write the time complexity as
$O((|S_{i-1}|+|\F|)\log (n+m))$ since $|S_{i-1}|\leq n$ and $|\F|=O(m+n)$. Hence, the total time of these two steps is $O(|S'_i|+(|S_{i-1}|+|\F|)\log (n+m))$.
Observe that $S'_i$ always grows after each round, implying that $\F$ is monotonically growing. Let $\F_i$ refer to the $\F$ at the last round of the procedure. Then, since there are at most $\sqrt{n}$ iterations and $S_i'\subseteq S_i$, the total time of these two steps in all rounds of the $i$-th iteration is bounded by $O(\sqrt{n}(|S_i|+(|S_{i-1}|+|\F_i|)\log (n+m)))$. Clearly, $\sum_{i}(|S_i|+|S_{i-1}|)\leq 2n$. Hence, the overall time complexity of these two steps in the whole algorithm is $O(\sqrt{n}(n+\sum_i|\F_i|)\log(n+m))$.
We will show later that $\sum_i|\F_i|=O(n+m)$, which leads to the result that the overall runtime of these two steps in the entire algorithm is bounded by $O(\sqrt{n}(n+m)\log (n+m))$.

The step on Line~\ref{line:comdis} takes $O(|S_{\F}|\log (n+m))$ time since the complexity of $\vd(S_{i-1},\F)$ is $O(n+m)$. It remains to prove an upper bound for $|S_{\F}|$. Observe that $S_{\F}\subseteq \N(S_{i-1})\cup \N(S_i)$. Since there are at most $\sqrt{n}$ rounds in the $i$-th iteration, the sum of $S_{\F}$ in all rounds of the $i$-th iteration is $O(\sqrt{n}(|\N(S_{i-1})|+|\N(S_i)|))$.
Note that $\sum_i|\N(S_i)|$ is linear in the size of $\DT(S)$, which is $O(n)$. Hence, $\sum_i|\N(S_i)|+\sum_i|\N(S_{i-1})|=O(n)$. Therefore, the sum of $S_{\F}$ in the whole algorithm is $O(n^{3/2})$, and thus the total time of Line~\ref{line:comdis} in the whole algorithm is $O(n^{3/2}\log (n+m))$.

The steps in Lines \ref{line:getA} and \ref{line:update} both take $O(|S_{\F}|)$ time. As the sum of $S_{\F}$ in the whole algorithm is $O(n^{3/2})$, the total time of these two steps in the entire algorithm is $O(n^{3/2})$.

In summary, the time complexity of the light-iteration procedure in the whole algorithm is  $O(\sqrt{n}(n+m)\log (n+m))$, provided that $\sum_i|\F_i|=O(n+m)$, which will be proved in the rest of this section.

\paragraph{Proving $\boldsymbol{\sum_i|\F_i|=O(n+m)}$.}
One tool we use to argue the complexity is that an obstacle vertex is shared by at most two pseudo-triangles and each pseudo-triangle will be included in $\F_i$ in at most one iteration of the algorithm. However, this is not exactly true. Consider a pseudo-triangle $\Delta_v(e)$ with $a$ and $b$ as the two endpoints of $e$. By definition, $\Delta_v(e)$ is bounded by $\pi(v,a)\cup \pi(v,b)\cup e$. It is possible that the two paths $\pi(v,a)$ and $\pi(v,b)$ share a common subpath of non-zero length. Specifically, let $c$ denote the point of $\pi(v,a)$ farthest from $v$ such that the subpath $\pi(v,c)$ of $\pi(v,a)$ is also a subpath of $\pi(v,b)$. We call $c$ the {\em cusp} of $\Delta_v(e)$. See Figure~\ref{fig:cusp} for an example.
If $c\neq v$, then the obstacle vertices in the interior of $\pi(v,c)$ may be shared by multiple pseudo-triangles. The consequence would be that those obstacle vertices could be included in $\F_i$ for multiple iterations, which would lead to a quadratic bound on $\sum_i|\F_i|$.

\begin{figure}
    \centering
    \includegraphics[height = 5.5cm]{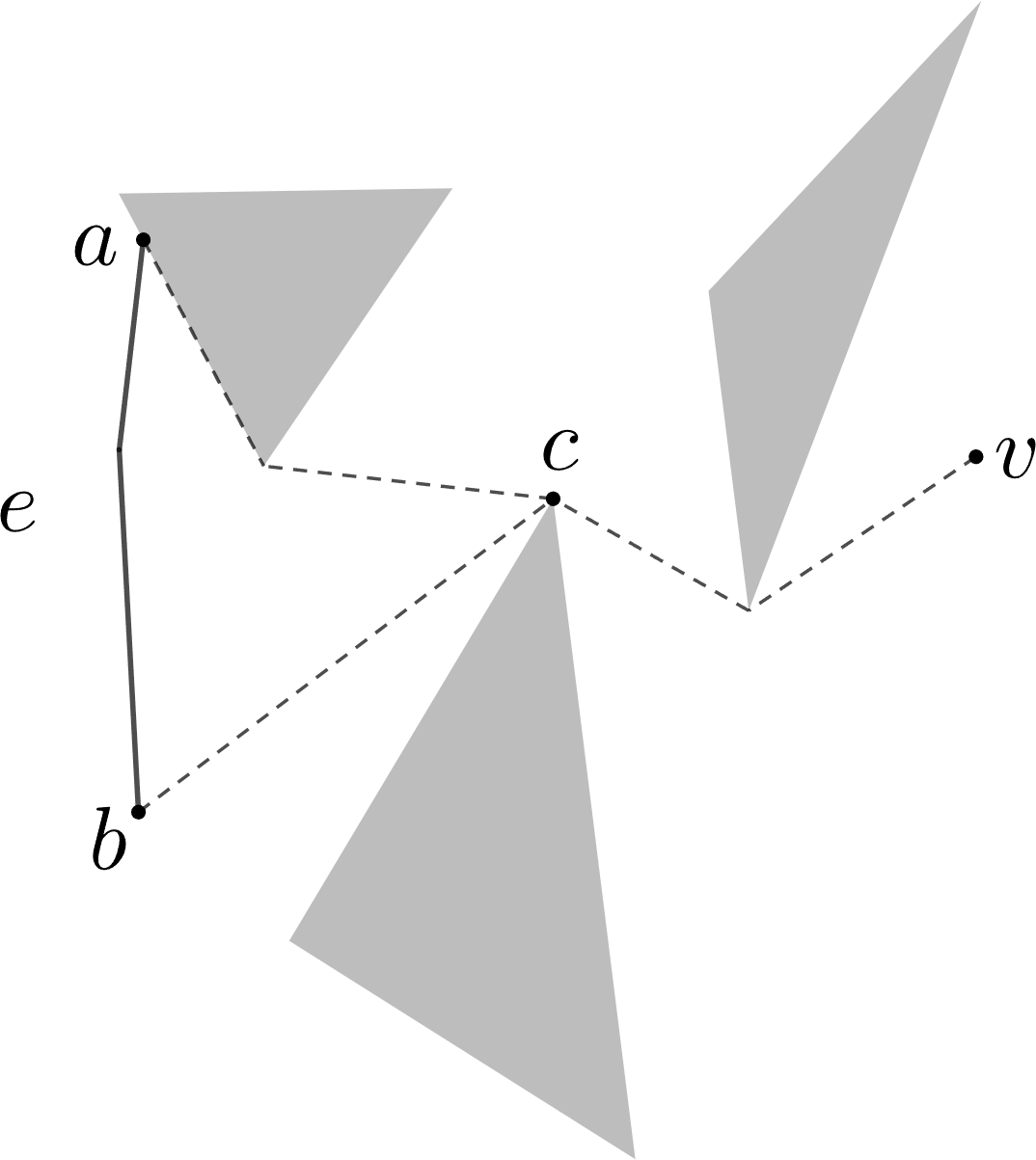}
    \caption{Example of $\Delta_v(e)$ with the cusp $c$.}
    \label{fig:cusp}
\end{figure}

To address this issue, we precompute $\df(c, v)$, which is the length of $\pi(c,v)$, and store it with $\Delta_v(e)$. Let $\Delta_v'(e)$ denote the region of $\Delta_v(e)$ bounded by $\pi(c,a)\cup \pi(c,b)\cup e$. We call $\Delta_v'(e)$ the {\em core} of $\Delta_v'(e)$ and call the path $\pi(v,c)$ excluding the cusp $c$ the {\em tail} of $\Delta_v(e)$. For each peripheral pseudo-triangle $\Delta_v(e)$ of $\F_i$ on Line~\ref{line:workregion}, we only include its core in $\F$ (so the tail is excluded from $\F$). In this way, obstacle vertices in the tail of $\Delta_v(e)$ will not be included in $\F_i$ for multiple iterations.
After the diagram $\vd(S_{i-1},\F)$ is computed, we have the information $d_{\F}(u,c)$, where $u$ is the site of $S_{i-1}$ whose region in $\vd(S_{i-1},\F)$ contains $c$. Hence, $u$ is the nearest neighbor of $c$ in $S_{i-1}$ and thus is also the nearest neighbor of $v$. Using the precomputed value $d(c,v)$, we can obtain $d_{\F}(S_{i-1},v)=d_{\F}(u,v)=d_{\F}(u,c)+d(c,v)$. In this way, the task on Line~\ref{line:comdis} can still be accomplished without including the tails of the peripheral pseudo-triangles in $\F$.

We say that a point is in the {\em interior} of a pseudo-triangle $\Delta_v(e)$ if the point is not in $\pi(v,a)\cup \pi(v,b)\cup e$. We have the following observation.

\begin{restatable}{observation}{lemtriandisj} \label{lem:trian_disj}
    All pseudo-triangles in $P$ are interior disjoint.
\end{restatable}
\begin{proof}
Consider a pseudo-triangle $\Delta_v(e)$. As $\Delta_v(e) \subseteq \R(v)$, it is trivially true that $\Delta_v(e)$ is
interior disjoint to any pseudo-triangle outside $\R(v)$. Hence, it is sufficient to argue that $\Delta_v(e)$ is interior disjoint to any other pseudo-triangle $\Delta_v(e')$ inside $\R(v)$. Since $e$ and $e'$ are interior disjoint, if $\Delta_v(e)$ and $\Delta_v(e')$ are not interior disjoint, they must cross along their other edges. However, this is not possible since their other edges all belong to shortest paths from $v$ and no two shortest paths from $v$ can cross each other~\cite{ref:HershbergerAn99}. Therefore, $\Delta_v(e)$ and $\Delta_v(e')$ must be interior disjoint.
\end{proof}

The following lemma finally proves the upper bound for $\sum_i|\F_i|$.

\begin{restatable}{lemma}{lemsmalltime} \label{lem:small_time}
    $\sum_i|\F_i|=O(m+n)$.
\end{restatable}
\begin{proof}
Consider a pseudo-triangle $\Delta = \Delta_v(e)$ with $a$ and $b$ as the two endpoints of $e$.

Let $m_o(\Delta)$ be the number of obstacle vertices in the interior of $\Delta$.
By Observation~\ref{lem:trian_disj}, the interiors of all pseudo-triangles are disjoint. Hence, it holds that $\sum_{\Delta} m_o(\Delta) = O(m)$.

    Let $V_\pi(\Delta)$ denote the set of vertices in $\pif(a, c)$ and $\pif(b, c)$ excluding $a$, $b$, and $c$, where $c$ is the cusp of $\Delta$. Let $m_\pi(\Delta)=|V_\pi(\Delta)|$. Observe that each vertex of $V_\pi(\Delta)$ is an obstacle vertex and it is in $V_\pi(\Delta')$ for at most one other pseudo-triangle $\Delta'$.
    Therefore, we have $\sum_{\Delta} m_\pi(\Delta) = O(m)$.


    Let $m_v$ be the number of obstacle vertices in $\R(v)$. As the combinatorial size of $\vd(S)$ is $O(m+n)$~\cite{ref:HershbergerAn99}, we get that $\sum_v m_v = O(m+n)$.

    Let $m_\bigstar(u)$ be the number of obstacle vertices in $\bigstar(u)$ excluding the vertices in the tails of the peripheral pseudo-triangles of $\bigstar(u)$ (we exclude those tail vertices because we have excluded the tail vertices of $\Delta_v(e)$ from $\F$ if $\Delta_v(e)$ is a peripheral pseudo-triangle of $\F$).
    We have:
    \[m_\bigstar(u) \leq m_u + \sum_{v \in \N(u), e \in \E(u, v)} (m_o(\Delta_v(e)) + m_\pi(\Delta_v(e)) + 3),\]
    where the $+3$ comes from the fact that we excluded the three vertices $a,b,c$ from $m_\pi(\cdot)$ for each pseudo-triangle $\Delta_v(e)$.

    Observe that $\sum_{u \in S} \sum_{v \in \N(u), e \in \E(u, v)} 3 = O(m + n)$ because the combinatorial complexity of $\vd(S)$ is $O(m + n)$~\cite{ref:HershbergerAn99}.
    Because $\sum_{\Delta} m_o(\Delta) = O(m)$, $\sum_{\Delta} m_\pi(\Delta) = O(m)$, $\sum_u m_u = O(m+n)$,
    and each pseudo-triangle can be a peripheral pseudo-triangle of $\bigstar(u)$ for a single site $u$,
    we obtain $\sum_{u \in S} m_\bigstar(u) = O(m+n)$.

    Define $m_i = \sum_{u \in S_{i - 1} \cup S_i} m_\bigstar(u)$. Let $S_i'$ refer to $S_i$ in the last round of the light-iteration procedure. Hence, $S_i'\subseteq S_i$. Recall that $\F_i=\bigstar(S_{i - 1} \cup S_i')$ (again, only the core is included in $\F_i$ for each peripheral pseudo-triangle of $\F_i$). Clearly, $|\F_i|=O(m_i)$. Since $\sum_{u \in S} m_\bigstar(u) = O(m+n)$ and the subsets $S_0,S_1,\ldots$ are pairwise disjoint, we obtain $\sum_im_i=O(m+n)$. Consequently, we have $\sum_i|\F_i|=O(m+n)$.

    The lemma thus follows.
\end{proof}

With Lemma~\ref{lem:small_time}, we conclude that the total time of the light-iteration procedure in the whole algorithm is $O(\sqrt{n}(n+m)\log (n+m))$, and as discussed, so is the total time of the whole algorithm.

The following theorem summarizes our result in this section.

\begin{theorem}
Given a polygonal domain $P$ of $m$ vertices, a set $S$ of $n$ points in $P$, and a source point $s\in S$, we can compute shortest paths from $s$ to all points of $S$ in the geodesic unit-disk graph of $S$ in $O(\sqrt{n}(n+m)\log (n+m))$ time.
\end{theorem}


\bibliographystyle{plainurl}
\bibliography{references}

\appendix




\end{document}